\documentclass[letterpaper, 11pt]{llncs}

\usepackage[margin=1in]{geometry}
\usepackage[bookmarks, colorlinks=true, plainpages = false, citecolor = blue,linkcolor=red,urlcolor = blue, filecolor = blue]{hyperref}
\usepackage{url}

\usepackage{amsmath,amsfonts,amssymb,bm, verbatim,mathtools}

\usepackage{color,graphicx,appendix}

\usepackage{subfigure}
\usepackage{etoolbox}
\usepackage{tikz}
\usepackage{xr,xspace}
\usepackage{todonotes}
\usepackage{paralist}
\usepackage{caption}

\newcommand{\mycomment}[1]{}
\newcommand{\eproof}{\hfill $\Box$}

\usepackage{xspace,prettyref}
\newcommand{\diverge}{\to\infty}

\newcommand{\reals}{{\mathbb{R}}}

\newrefformat{eq}{(\ref{#1})}
\newrefformat{chap}{Chapter~\ref{#1}}
\newrefformat{sec}{Section~\ref{#1}}
\newrefformat{algo}{Algorithm~\ref{#1}}
\newrefformat{fig}{Fig.~\ref{#1}}
\newrefformat{tab}{Table~\ref{#1}}
\newrefformat{rmk}{Remark~\ref{#1}}
\newrefformat{clm}{Claim~\ref{#1}}
\newrefformat{def}{Definition~\ref{#1}}
\newrefformat{cor}{Corollary~\ref{#1}}
\newrefformat{lmm}{Lemma~\ref{#1}}
\newrefformat{prop}{Proposition~\ref{#1}}
\newrefformat{app}{Appendix~\ref{#1}}
\newrefformat{hyp}{Hypothesis~\ref{#1}}
\newrefformat{thm}{Theorem~\ref{#1}}

\newcommand{\pth}[1]{\left( #1 \right)}

\newcommand{\tr}{\widetilde{r}}

\newcommand{\tx}{{\widetilde{x}}}

\newcommand{\tC}{{\widetilde{C}}}

\newcommand{\tY}{{\widetilde{Y}}}

\newcommand{\calC}{{\mathcal{C}}}

\newcommand{\calE}{{\mathcal{E}}}
\newcommand{\calF}{{\mathcal{F}}}

\newcommand{\calH}{{\mathcal{H}}}

\newcommand{\calL}{{\mathcal{L}}}

\newcommand{\calN}{{\mathcal{N}}}
\newcommand{\calO}{{\mathcal{O}}}

\newcommand{\calR}{{\mathcal{R}}}
\newcommand{\calS}{{\mathcal{S}}}

\newcommand{\calV}{{\mathcal{V}}}

\newcommand{\argmin}{{\rm argmin}}
\newcommand{\argmax}{{\rm argmax}}

\begin{document}

\title{Fault-Tolerant Multi-Agent Optimization-- Part III
\thanks{This research is supported in part by National Science Foundation awards NSF 1329681 and 1421918.
Any opinions, findings, and conclusions or recommendations expressed here are those of the authors
and do not necessarily reflect the views of the funding agencies or the U.S. government.}}

\author{Lili Su \hspace*{1in} Nitin Vaidya}
\institute{Department of Electrical and Computer Engineering, and\\
Coordinated Science Laboratory\\
University of Illinois at Urbana-Champaign\\
Email:\{lilisu3, nhv\}@illinois.edu}

\maketitle

\begin{center}

Technical Report\\

\vskip \baselineskip

\today
~
\end{center}

\begin{abstract}
We study fault-tolerant distributed optimization of a sum of convex (cost) functions with real-valued scalar
input/output in the presence of crash faults or Byzantine faults.
In particular, the goal is to optimize
a global cost function $\frac{1}{n}\sum_{i\in \calV} h_i(x)$, where $\calV=\{1, \ldots, n\}$ is the collection of agents, and $h_i(x)$ is agent $i$'s local cost function, which is initially
known only to agent $i$.
 This problem finds its applications in the domain of fault-tolerant large scale distributed machine learning, where data are generated at different locations and some data may be lost during processing or be tampered by malicious local data managers.  The global cost function $\frac{1}{n}\sum_{i\in \calV} h_i(x)$ captures the requirement that, in distributed machine learning, the system tries to take full advantage of all the data generated at different locations. 
Since the above global cost function cannot be optimized exactly in presence of crash faults or Byzantine faults, we define two weaker versions of the problem for crash faults and Byzantine faults, respectively.\\

When some agents may crash, the local functions/data stored at these agents may not always available to the system. In this scenario, the goal for the weaker problem is to generate an output that is an optimum of a function formed as
$$C\pth{\sum_{i\in \calN} h_i(x)+\sum_{i\in \calF} \alpha_i h_i(x)},$$ where $\calN$ is the set of non-faulty agents, $\calF$ is the set of faulty agents (crashed agents), $0\le \alpha_i\le 1$ for each $i\in \calF$ and $C$ is a normalization constant such that $C\pth{|\calN|+\sum_{i\in \calF} \alpha_i}=1$. We present an iterative algorithm in which each agent only needs to perform local computation, and send one message per iteration.\\


When some agents may be Byzantine, the system cannot take full advantage of the data kept by non-faulty agents.
The goal for the associated weaker problem is to generate an output that is an optimum of a function formed as
$$\sum_{i\in \calN}\alpha_i h_i(x),$$
such that $\alpha_i\geq 0$ for each $i\in \calN$ and $\sum_{i\in \calN}\alpha_i=1$.
We present an iterative algorithm, where only local computation is needed and only one message per agent is sent in each iteration, that ensures that at least $|\calN|-f$ agents have weights ($\alpha_i$'s) that are lower bounded by $\frac{1}{2(|\calN|-f)}$.\\

The obtained results can be generalized to asynchronous systems as well.

\end{abstract}

~

~

\newpage

\section{System Model and Problem Formulation}
\label{sec:intro}
The system under consideration is synchronous, and
 consists of $n$ agents connected by a complete communication network.
Our results can be generalized to asynchronous system. We postpone the discussion of this generalization to the end of this report.
The set of agents is $\calV=\{1,\cdots,n\}$.
We assume that $n>3f$ for reasons that will be clearer soon.
We say that a function $h: \mathbb{R}\rightarrow \mathbb{R}$ is {\em admissible} if (i) $h(\cdot)$ is convex, and continuously
differentiable,
 (ii) the set $\arg\min_{x\in\mathbb{R}} h(x)$ containing the optima of $h(\cdot)$
is non-empty and compact (i.e., bounded and closed), (iii) the magnitude of the gradient is bounded by $L$, i.e., $|h^{\prime}(x)|\le L,\, \forall  x\in \reals$, and the derivative $h^{\prime}(\cdot)$ is $L$--Lipschitz continuous.
Each agent $i\in \calV$ is initially provided with an {\em admissible} local cost function $h_i: \mathbb{R}\rightarrow\mathbb{R}$. Ideally, the system goal is to optimize the {\em average} of all the local functions, and have all the agents to reach agreement on the optimum $x$. In particular, each agent should output an {\em identical} value $\tx\in\mathbb{R}$
that minimizes
\begin{align}
\label{problem1}
\frac{1}{n}\sum_{i\in \calV} h_i(x).
\end{align}
This problem finds its applications in the domain of large scale distributed machine learning, where data are generated at different locations and the data center at each location is not allowed to transmit all the locally collected data to other centers either due to transmission capacity constraint or due to privacy issue. 

This problem is well-studied in the scenario where each agent is reliable throughout any execution of an algorithm \cite{Duchi2012,Nedic2009,Tsianos2012}. In this work, we consider the fault-tolerant version of this problem. In particular, we consider the setting where up to $f$ of the $n$ agents may crash or be Byzantine faulty. Let $\calF$ denote the set of faulty agents, and let $\calN = \calV - \calF$ denote the set of non-faulty agents. For each $t\ge 0$, let $\calN[t]$ be the collection of agents that have not been crashed till the end of iteration $t$, with $\calN[0]=\calV$. Note that $\calN[t+1]\subseteq \calN[t]$ for $t\ge 0$, and that $\lim_{t\diverge} \calN[t]=\calN$.
The set $\calF$ of faulty agents may be chosen by an adversary arbitrarily. Let $|\calF|=\phi$. Note that $\phi\leq f$ and $|\calN|\geq n-f$.
The presence of crashed or Byzantine faulty agents makes it impossible to design an algorithm that solves (\ref{problem1}) for all admissible local cost functions (this is shown formally in Part I of this report \cite{su2015byzantine}).  Therefore, for crash fault and Byzantine fault, respectively, we study two weaker versions of the problem, namely, Problem 1 and Problem 2 in Figure \ref{fig:prob}. Problem 1 is proposed in this report, and Problem 2 is first introduced in Part I of this report \cite{su2015byzantine}.\\

When some agents may crash, the local functions/data stored at these agents may only be partially visible or even invisible to the system -- as agent $i$ may crash at any time during an execution. Problem 1 requires that the output $\tx$ be an optimum of a function formed as
$$C\pth{\sum_{i\in \calN} h_i(x)+\sum_{i\in \calF} \alpha_i h_i(x)},$$ where $\calN$ is the set of non-faulty agents, $\calF$ is the set of faulty agents (crashed agents), $0\le \alpha_i\le 1$ for each $i\in \calF$ and $C$ is a normalization constant such that $C\pth{|\calN|+\sum_{i\in \calF} \alpha_i}=1$. 

When some agents may be Byzantine, the system cannot take full advantage of the data kept by non-faulty agents. In addition, among the non-faulty agents, the system may put more weights to some agents than the others. Then, the desired goal is to {\em maximize} the number of weights ($\alpha_i$'s)
that are bounded away from zero.
With this in mind, Problem 2 in Figure \ref{fig:prob} is introduced in Part I of this report \cite{su2015byzantine}. In Problem 2, note that ${\bf 1}\{\alpha_i>\beta\}$ is an indicator function that outputs 1 if $\alpha_i>\beta$, and 0 otherwise.
Essentially, Problem 2 requires that at least $\gamma$
weights must exceed a threshold $\beta$, where $\beta\geq 0$. Thus, $\beta,\gamma$ are parameters of Problem 2, capturing how the data collected by non-faulty agents are utilized by the system.

%

%
%
%
%

\begin{figure}
\begin{tabular}{|l|l|l|}\hline
\begin{minipage}[t]{0.50\textwidth}
{\bf Problem 1}\\
\begin{eqnarray*}
\tx & \in & \arg \min_{x\in\mathbb{R}}\quad C\pth{\sum_{i\in \calN} h_i(x)+\sum_{i\in \calF} \alpha_i h_i(x)}\\
\text{such that} && \nonumber \\
&&\forall i\in\calF, ~ 0\le \alpha_i\le 1 \text{~~and~~} \nonumber \\
&&C\pth{|\calN|+\sum_{i\in \calF} \alpha_i}=1 \nonumber
\end{eqnarray*}
\end{minipage}
&
\begin{minipage}[t]{0.5\textwidth}
{\bf Problem 2 with parameters $\beta,\gamma$, $\beta\geq 0$}\\
\begin{eqnarray*}
\tx & \in & \arg \min_{x\in\mathbb{R}}\quad \sum_{i\in \calN} \alpha_i h_i(x)\\
\text{such that} && \nonumber \\
&&\forall i\in\calN, ~ \alpha_i\geq 0,  \nonumber \\
&&\sum_{i\in \calN}\alpha_i=1, \text{~~and~~} \nonumber \\
&&\sum_{i\in\calN} {\bf 1}(\alpha_i>\beta) ~ \geq ~ \gamma \nonumber
\end{eqnarray*}

\end{minipage}
~\\
\hline
\end{tabular}
\caption{Problem formulations: All non-faulty agents
must output an identical value $\tx\in\mathbb{R}$ that satisfies the constraints specified in each problem formulation.}
\label{fig:prob}
\end{figure}

We will say that Problem 1 or 2 is solvable if there exists an algorithm that will
find a solution for the problem (satisfying all its constraints) for all admissible local cost functions,
and all possible behaviors of faulty agents.
Our problem formulations require that all non-faulty agents output asymptotically identical $\tx\in\mathbb{R}$, while
satisfying the constraints imposed by the problem (as listed in Figure \ref{fig:prob}).
Thus, the traditional fault-tolerant consensus \cite{impossible_proof_lynch} problem, which also imposes a similar
{\em agreement} condition, is a special case
of our optimization problem.\footnote{Interested readers are referred to Part I of this report \cite{su2015byzantine} for formal proof. } 
Therefore,
the lower bound of $n>3f$ for Byzantine consensus \cite{impossible_proof_lynch} also applies
to our problem. Hence we assume that $n>3f$.\\

We prove the following key results:
\begin{itemize}
\item (Theorem 2) We provide a simple iterative algorithm that solves Problem 1. In each iteration of this algorithm, each agent only needs to perform local computation and send one message.
\item (Theorem 4) We present a simple iterative algorithm that solves Problem 2 with $\beta\leq \frac{1}{2(|\calN|-f)}$ and $\gamma\leq |\calN|-f$. In each iteration of this algorithm, each agent only needs to perform local computation and send one message.
\end{itemize}
In our proposed algorithms, the local estimates at all non-faulty agents are identical in the limit. 

The rest of the report is organized as follows. Related work is summarized in Section \ref{related work}.  Two algorithms are proposed in Section \ref{sec: algorithm crash}, wherein the first algorithm solves Problem 1 with two rounds of information exchange in each iteration. In contrast, the second algorithm solves Problem 1 with one message sent per agent in each iteration.
Section \ref{sec: algorithm byzantine} presents a simple iterative algorithm that solves Problem 2 with $\beta=\frac{1}{2\pth{|\calN|-f}}$ and $\gamma=|\calN|-f$. Similar to the second algorithm in Section \ref{sec: algorithm crash}, this proposed algorithm only requires one message sent per agent in each iteration. Section \ref{conclusion and discussion} discusses the generalization of the obtained results to asynchronous systems, and concludes the report.

\section{Related Work}\label{related work}

Fault-tolerant consensus \cite{PeaseShostakLamport} is a special case of the optimization problem considered in this report. There is a significant body of work on fault-tolerant consensus, including \cite{Dolev:1986:RAA:5925.5931,Chaudhuri92morechoices,mostefaoui2003conditions,fekete1990asymptotically,LeBlanc2012,vaidya2012iterative,friedman2007asynchronous}.
The optimization algorithms presented in this report use Byzantine consensus as a component.

Convex optimization, including distributed convex optimization, also has a long history \cite{bertsekas1989parallel}. However, we are not aware of
prior work that obtains the results presented in this report except \cite{su2015byzantine,DBLP:journals/corr/SuV15a}.
Primal and dual decomposition methods that led themselves naturally to a distributed paradigm are well-known \cite{Boyd2011}. 
There has been significant research on a variant of distributed optimization problem \cite{Duchi2012,Nedic2009,Tsianos2012}, in which the global objective $h(x)$ is a summation of $n$ convex functions, i.e, $h(x)=\sum_{j=1}^n h_j(x)$, with function $h_j(x)$ being known to the $j$-th agent. The need for robustness for distributed optimization problems has received some attentions recently \cite{Duchi2012,kailkhura2015consensus,zhang2014distributed,marano2009distributed,su2015byzantine,DBLP:journals/corr/SuV15a}. In particular, Duchi et al.\ \cite{Duchi2012} studied the impact of random communication link faults on the convergence of distributed variant of dual averaging algorithm. Specifically, each realizable link fault pattern considered in \cite{Duchi2012} is assumed to admit a doubly-stochastic matrix which governs
the evolution dynamics of local estimates of the optimum.

We considered Byzantine fault in \cite{su2015byzantine} and \cite{DBLP:journals/corr/SuV15a}. Both \cite{su2015byzantine} and \cite{DBLP:journals/corr/SuV15a} considered synchronous system. \cite{su2015byzantine} showed that at most $|\calN|-f$ non-faulty functions can have non-zero weights. This observation led to the formulation of Problem 2 in Fig. \ref{fig:prob}. Six algorithms were proposed in \cite{su2015byzantine}. In contrast, we also showed \cite{DBLP:journals/corr/SuV15a} that sufficient redundancy in the input functions (each input function is not exclusively kept by a single agent), it is possible to solve (\ref{problem1}), where the summation is over all input functions. In addition, a simple low-complexity iterative algorithm was proposed in \cite{DBLP:journals/corr/SuV15a}, and a tight topological condition for the existence of such iterative algorithms is identified.

In other related work, significant attempts have been made to solve the problem of distributed hypothesis testing in the presence of Byzantine attacks \cite{kailkhura2015consensus,zhang2014distributed,marano2009distributed}, where Byzantine sensors may transmit fictitious observations aimed at confusing the decision maker to arrive at a judgment that is in contrast with the true underlying distribution. Consensus based variant of distributed event detection, where a centralized data fusion center does not exist, is considered in \cite{kailkhura2015consensus}. In contrast, in this paper, we focus on the Byzantine attacks on the multi-agent optimization problem.


\section{Mutil-Agent Optimization with Crash fault}
\label{sec: algorithm crash}

Algorithm 1 and its correctness proof contain the key ideas and intuition of this report.

\subsection{Algorithm 1: Two-Round of Information Exchange per Iteration}

In Algorithm 1, each agent $j$ maintains two variables: the local estimate $x_j$ and the auxiliary variable $s_j$, with $x_j[t]$ and $s_j[t]$ representing these two variables at the end of iteration $t$, and $x_j[0]$ being the system input at agent $j$ and $s_j[0]=0$. In each iteration $t\ge 1$, there are two rounds of information exchange. In the first round, (1) agent $j$ requests all the agents (including itself) to compute the gradients of their local functions at $x_j[t-1]$; (2) after receiving $x_j[t-1]$, a non-faulty agent $i$ computes $h_i^{\prime}(x_j[t-1])$ and sends it back to agent $j$; (3) agent $j$ collects the requested gradients and updates the auxiliary variable $s_j$. In the second round, all the non-faulty agents exchange their auxiliary variables $s_j[t]$'s and update their local estimate as an {\em average} of all received auxiliary variables.

Let $\{\lambda[t]\}_{t=0}^{\infty}$ be a sequence of stepsizes chosen beforehand such that $\lambda[t]\ge 0$ ad $\lambda[t]\ge \lambda[t+1]$ for each $t\ge 0$, $\lim_{t\diverge}\lambda[t]=0 $,  $\sum_{t=0}^{\infty} \lambda[t]=\infty$ and $\sum_{t=0}^{\infty} \lambda^2[t]<\infty$.

\paragraph{}
\vspace*{8pt}\hrule
~

{\bf Algorithm 1} for agent $j$ at iteration $t$:
~
\vspace*{4pt}\hrule

\begin{list}{}{}

\item[{\bf Step 1:}] Send $x_j[t-1]$ to all the agents (including agent $j$ itself).\\
~
\item[{\bf Step 2:}]
Upon receiving $x_i[t-1]$ from agent $i$, compute $h_j^{\prime}(x_i[t-1])$ --
the gradient of function $h_j(\cdot)$ at $x_i[t-1]$ -- and send it back to agent $i$. \\
~
\item[{\bf Step 3:}]
Let $\calR^{1}_j[t-1]$ denote the set of gradients of the form $h_i^{\prime}(x_j[t-1])$ received as a result of step 1 and step 2.
%
Update $s_j$ as
\begin{align}
\label{update z crash 1}
s_j[t]=x_j[t-1]-\frac{\lambda[t-1]}{\left |\calR^{1}_j[t-1]\right |}\pth{\sum_{i\in \calR^{1}_j[t-1]}h_i^{\prime}(x_j[t-1])}.
\end{align}
~

\item[{\bf Step 4:}] Send $s_j[t]$ to all the agents (including agent $j$ itself). \\
~
\item[{\bf Step 5:}] 
    Let $\calR^{2}_j[t-1]$ denote the set of auxiliary variables $s_i[t]$ received as a result of step 4.\\
    Update $x_j$ as
\begin{align}
\label{update x crash 1}
x_j[t]= \frac{1}{\left | \calR^{2}_j[t-1] \right |}\sum_{i\in \calR^{2}_j[t-1]} s_i[t].
\end{align}

\end{list}

\hrule

~

\vskip \baselineskip

Steps 1, 2 and 3 correspond to the first round of information exchange, and step 4 corresponds to the second round of information exchange.
We will show that Algorithm 1 correctly solves Problem 1. Intuitively speaking, the first round of information exchange corresponds to the standard gradient-method iterate, which drives each local estimate to a global
optimum; the second round of information exchange forces all local estimates at non-faulty agents to reach consensus.  Algorithm 2 will achieve a similar goal with a single round of exchange.\\

Recall that $\calN$ is the set of non-faulty agents and $\calF$ is the set of faulty agents that may crash at any time during an execution. In Problem 1, the system goal is to optimize
$$C\pth{\sum_{i\in \calN} h_i(x)+\sum_{i\in \calF} \alpha_i h_i(x)},$$ where $\calN$ is the set of non-faulty agents, $\calF$ is the set of faulty agents (crashed agents), $0\le \alpha_i\le 1$ for each $i\in \calF$ and $C$ is a normalization constant such that $C\pth{|\calN|+\sum_{i\in \calF} \alpha_i}=1$. Given $\calN$ and $\calF$, the normalization constant $C$ and the crashed agents' coefficients $\alpha_i$ depend on when the faulty agents crash during an execution. 
For given $\calN$ and $\calF$, let $\calC$ be the collection of potential system objectives, formally defined as follows:
\begin{align}
\nonumber
\calC\triangleq \Big{\{}~~p(x)~:~ p(x)&=C\pth{\sum_{i\in \calN} h_i(x)+\sum_{i\in \calF} \alpha_i h_i(x)}, \\
~~&\forall i\in\calF, ~0\le \alpha_i\le 1,~ C\pth{|\calN|+\sum_{i\in \calF} \alpha_i}=1.~~\Big{\}}
\label{valid collection crash}
\end{align}
Each $p(x)\in \calC$ is called a valid function.
Since $\forall i\in\calF, ~0\le \alpha_i\le 1$, it holds that $\frac{1}{n}\le C\le \frac{1}{|\calN|}$ for each valid function. Note that $\frac{1}{|\calN|}\sum_{i\in \calN}h_i(x)\in \calC$ is a valid function. For ease of future reference, we let $\widetilde{p}(x)\triangleq\frac{1}{|\calN|}\sum_{i\in \calN}h_i(x)$. Define $Y\triangleq \cup_{p(x)\in \calC} \argmin~ p(x)$. The characterization of $Y$ is presented in the following two lemmas.

\begin{lemma}
\label{valid convex crash}
$Y$ is a convex set.
\end{lemma}

\begin{lemma}
\label{valid closed crash}
$Y$ is a closed set.
\end{lemma}
The proofs of Lemma \ref{valid convex crash} and Lemma \ref{valid closed crash} are presented in Appendix \ref{appendix: valid convex crash} and Appendix \ref{appendix: valid closed crash}, respectively. \\

In addition, since $Y$ is convex, $Dist\pth{\cdot, Y}$ is also convex.

\subsubsection{Asymptotic Consensus under Algorithm 1}
We first show that asymptotic consensus among the non-faulty agents is achieved under Algorithm 1. The following proposition is used in proving consensus.

\begin{proposition}
\label{crash sum 0}
Let $0\le b<1$. Define $\ell(t)=\sum_{r=0}^{t-1}\lambda[r] b^{t-r}.$ The limit of $\ell(t)$ exists and
$$\lim_{t\diverge} \ell(t)=0.$$
\end{proposition}
Proposition \ref{crash sum 0} is proved in Appendix \ref{appendix: crash sum 0}. \\

Recall that for each $t\ge 0$, $\calN[t]$ is the collection of agents that have not been crashed till the end of iteration $t$. Note that $\calN[t+1]\subseteq \calN[t]$ for $t\ge 0$, and that $\lim_{t\diverge} \calN[t]=\calN$.
Denote $M(t)=\max_{i\in \calN[t]} x_i[t]$ and $m(t)=\min_{i\in \calN[t]} x_i[t]$.

\begin{lemma}
\label{consensus alg1}
Under Algorithm 1, the sequence $\{M[t]-m[t]\}_{t=0}^{\infty}$ converges and
$$\lim_{t\diverge} \pth{M[t]-m[t]}~=~0.$$
\end{lemma}
\begin{proof}

Let $i, j\in \calN[t]$ such that 
$x_i[t]=M[t]$, and $x_j[t]=m[t]$.
\begin{align}
\label{iterative m crash 1}
\nonumber
M[t]-m[t]&=x_i[t]-x_j[t]\\
\nonumber
&= \frac{1}{\left | \calR^{2}_i[t-1] \right |}\sum_{k\in \calR^{2}_i[t-1]} s_k[t] - \frac{1}{\left | \calR^{2}_j[t-1] \right |}\sum_{p\in \calR^{2}_j[t-1]} s_p[t] ~~~~~~~~\text{by}~(\ref{update x crash 1})\\
\nonumber
&=\min\left \{\frac{1}{\left | \calR^{2}_i[t-1] \right |}, \frac{1}{\left | \calR^{2}_j[t-1] \right |}\right\} \pth{\sum_{k\in \calR^{2}_i[t-1]} s_k[t]-\sum_{p\in \calR^{2}_j[t-1]} s_p[t]}\\
\nonumber
&\quad+\pth{\frac{1}{\left | \calR^{2}_i[t-1] \right |}-\min\left\{\frac{1}{\left | \calR^{2}_i[t-1] \right |}, \frac{1}{\left | \calR^{2}_j[t-1] \right |}\right\} }\sum_{k\in \calR^{2}_i[t-1]} s_k[t]\\
&\quad - \pth{\frac{1}{\left | \calR^{2}_j[t-1] \right |}-\min\left\{\frac{1}{\left | \calR^{2}_i[t-1] \right |}, \frac{1}{\left | \calR^{2}_j[t-1] \right |}\right\} }\sum_{p\in \calR^{2}_j[t-1]} s_p[t].
\end{align}


Assume that $\left | \calR^{2}_i[t-1] \right |\ge \left | \calR^{2}_j[t-1] \right |$. The case that $\left | \calR^{2}_i[t-1] \right |< \left | \calR^{2}_j[t-1] \right |$ can be shown similarly.

We can simplify (\ref{iterative m crash 1}) as follows.
\begin{align}
\label{im1}
M[t]-m[t]=\frac{1}{\left | \calR^{2}_i[t-1] \right |}\pth{\sum_{k\in \calR^{2}_i[t-1]} s_k[t]-\sum_{p\in \calR^{2}_j[t-1]} s_p[t]}- \pth{\frac{1}{\left | \calR^{2}_j[t-1] \right |}-\frac{1}{\left | \calR^{2}_i[t-1] \right |} }\sum_{p\in \calR^{2}_j[t-1]} s_p[t].
\end{align}
For each $k$, wet get
\begin{align}
\label{alg1 consensus 1}
\nonumber
s_k[t]&=x_k[t-1]-\frac{\lambda[t-1]}{\left |\calR^{1}_k[t-1]\right |}\pth{\sum_{i\in \calR^{1}_k[t-1]}h_i^{\prime}(x_k[t-1])} \\
\nonumber
&\le x_k[t-1]+\frac{\lambda[t-1]}{\left |\calR^{1}_k[t-1]\right |}\pth{\sum_{i\in \calR^{1}_k[t-1]}L}~~~\text{since }|h_k^{\prime}(x)|\le L, \forall x\in \reals, \forall i\in \calV \\
\nonumber
&=x_k[t-1]+\frac{\lambda[t-1]}{\left |\calR^{1}_k[t-1]\right |}\left| \calR^{1}_k[t-1]\right|L\\
&=x_k[t-1]+\lambda[t-1]L\le M[t-1]+\lambda[t-1]L.
\end{align}
Similarly, for each $k\in \calN[t]$, it holds that
\begin{align}
\label{alg1 consensus 2}
s_k[t]\ge m[t-1]-\lambda[t-1]L.
\end{align}

We bound the two terms in the right hand side of (\ref{im1}) separately.
For the first term of (\ref{im1}) , we get
\begin{align}
\label{a1}
\nonumber
&\frac{1}{\left | \calR^{2}_i[t-1] \right |}\pth{\sum_{k\in \calR^{2}_i[t-1]} s_k[t]-\sum_{p\in \calR^{2}_j[t-1]} s_p[t]}\\
\nonumber
&=\frac{1}{\left | \calR^{2}_i[t-1] \right |}\pth{\sum_{k\in \calR^{2}_i[t-1]\cap \calR^{2}_j[t-1]} s_k[t]+\sum_{k\in \calR^{2}_i[t-1]-\calR^{2}_j[t-1]} s_k[t]-\sum_{p\in \calR^{2}_j[t-1]\cap \calR^{2}_i[t-1]} s_p[t]-\sum_{p\in \calR^{2}_j[t-1]- \calR^{2}_i[t-1]} s_p[t]}\\
\nonumber
&=\frac{1}{\left | \calR^{2}_i[t-1] \right |}\pth{\sum_{k\in \calR^{2}_i[t-1]-\calR^{2}_j[t-1]} s_k[t]-\sum_{p\in \calR^{2}_j[t-1]- \calR^{2}_i[t-1]} s_p[t]}\\
\nonumber
&\overset{(a)}{\le} \frac{1}{\left | \calR^{2}_i[t-1] \right |}\pth{\sum_{k\in \calR^{2}_i[t-1]-\calR^{2}_j[t-1]} \pth{M[t-1]+\lambda[t-1]L}-\sum_{p\in \calR^{2}_j[t-1]- \calR^{2}_i[t-1]} \pth{m[t-1]-\lambda[t-1]L}}\\
&= \frac{\left |\calR^{2}_i[t-1]-\calR^{2}_j[t-1]\right |}{\left | \calR^{2}_i[t-1] \right |} \pth{M[t-1]+\lambda[t-1]L}-\frac{\left |\calR^{2}_j[t-1]-\calR^{2}_i[t-1]\right |}{\left | \calR^{2}_i[t-1] \right |} \pth{m[t-1]-\lambda[t-1]L}.
\end{align}
Inequality $(a)$ holds due to (\ref{alg1 consensus 1}) and (\ref{alg1 consensus 2}).

For the second term of (\ref{im1}), we get
\begin{align}
\label{a2}
\nonumber
- \pth{\frac{1}{\left | \calR^{2}_j[t-1] \right |}-\frac{1}{\left | \calR^{2}_i[t-1] \right |} }\sum_{p\in \calR^{2}_j[t-1]} s_p[t]~&=~-\frac{\left | \calR^{2}_i[t-1] \right |-\left | \calR^{2}_j[t-1] \right |}{\left | \calR^{2}_j[t-1] \right | \left | \calR^{2}_i[t-1] \right |}\sum_{p\in \calR^{2}_j[t-1]} s_p[t]\\
\nonumber
&\le ~-\frac{\left | \calR^{2}_i[t-1] \right |-\left | \calR^{2}_j[t-1] \right |}{\left | \calR^{2}_j[t-1] \right | \left | \calR^{2}_i[t-1] \right |}\sum_{p\in \calR^{2}_j[t-1]} \pth{m[t-1]-\lambda[t-1]L}\\
&=-\frac{\left | \calR^{2}_i[t-1] \right |-\left | \calR^{2}_j[t-1] \right |}{\left | \calR^{2}_i[t-1] \right |} \pth{m[t-1]-\lambda[t-1]L}.
\end{align}

By (\ref{a1}) and  (\ref{a2}), (\ref{im1}) can be further bounded as
\begin{align}
\label{alg ub}
\nonumber
M[t]-m[t]&=\frac{1}{\left | \calR^{2}_i[t-1] \right |}\pth{\sum_{k\in \calR^{2}_i[t-1]} s_k[t]-\sum_{p\in \calR^{2}_j[t-1]} s_p[t]}- \pth{\frac{1}{\left | \calR^{2}_j[t-1] \right |}-\frac{1}{\left | \calR^{2}_i[t-1] \right |} }\sum_{p\in \calR^{2}_j[t-1]} s_p[t]~~~\text{by }(\ref{im1})\\
\nonumber
&\le \frac{\left |\calR^{2}_i[t-1]-\calR^{2}_j[t-1]\right |}{\left | \calR^{2}_i[t-1] \right |} \pth{M[t-1]+\lambda[t-1]L}-\frac{\left |\calR^{2}_j[t-1]-\calR^{2}_i[t-1]\right |}{\left | \calR^{2}_i[t-1] \right |} \pth{m[t-1]-\lambda[t-1]L}\\
\nonumber
&\quad-\frac{\left | \calR^{2}_i[t-1] \right |-\left | \calR^{2}_j[t-1] \right |}{\left | \calR^{2}_i[t-1] \right |} \pth{m[t-1]-\lambda[t-1]L}~~~\text{by}~(\ref{a1})~\text{and}~  (\ref{a2})\\
\nonumber
&\overset{(a)}{=}\frac{\left |\calR^{2}_i[t-1]-\calR^{2}_j[t-1]\right |}{\left | \calR^{2}_i[t-1] \right |} \pth{M[t-1]+\lambda[t-1]L}-\frac{\left |\calR^{2}_i[t-1]-\calR^{2}_j[t-1]\right |}{\left | \calR^{2}_i[t-1] \right |} \pth{m[t-1]-\lambda[t-1]L}\\
\nonumber
&=\frac{\left |\calR^{2}_i[t-1]-\calR^{2}_j[t-1]\right |}{\left | \calR^{2}_i[t-1] \right |} \pth{M[t-1]-m[t-1]+2\lambda[t-1]L}\\
\nonumber
&\overset{(b)}{\le}\frac{f}{n-f} \pth{M[t-1]-m[t-1]+2\lambda[t-1]L}\\
& \le \pth{\frac{f}{n-f}}^{t}\pth{M[0]-m[0]}+2L\pth{\sum_{r=0}^{t-1}\pth{\frac{f}{n-f}}^{t-r}\lambda[r]}.
\end{align}
Equality $(a)$ is true because that
\begin{align*}
&\left |\calR^{2}_j[t-1]-\calR^{2}_i[t-1]\right |+\left |\calR^{2}_i[t-1]\right |- \left |\calR^{2}_j[t-1]\right |\\
&=\left |\calR^{2}_j[t-1]\right |-\left |\calR^{2}_j[t-1]\cap \calR^{2}_i[t-1]\right |+\left |\calR^{2}_i[t-1]\right |- \left |\calR^{2}_j[t-1]\right |\\
&=\left |\calR^{2}_i[t-1]\right |-\left |\calR^{2}_j[t-1]\cap \calR^{2}_i[t-1]\right |\\
&=\left |\calR^{2}_i[t-1]-\calR^{2}_j[t-1]\right |.
\end{align*}
Since $|\calF|\le f$, it holds that $\left |\calR^{2}_i[t-1]-\calR^{2}_j[t-1]\right |\le f$ and that $\left |\calR^{2}_i[t-1]\right |\ge n-f$. Thus,
$$\frac{\left |\calR^{2}_i[t-1]-\calR^{2}_j[t-1]\right |}{\left | \calR^{2}_i[t-1] \right |}\le \frac{f}{n-f},$$
and inequality $(b)$ holds.\\

It follows from Proposition \ref{crash sum 0} that
$$\lim_{t\diverge} 2L\pth{\sum_{r=0}^{t-1}\pth{\frac{f}{n-f}}^{t-r}\lambda[r]}=0,$$
where $b=\frac{f}{n-f}$.
Thus, taking limit sup on both sides of (\ref{alg ub}), we get
\begin{align*}
\limsup_{t\diverge}\pth{M[t]-m[t]}\le \lim_{t\diverge} \pth{\pth{\frac{f}{n-f}}^{t}\pth{M[0]-m[0]}}+2L\lim_{t\diverge} \pth{\sum_{r=0}^{t-1}\pth{\frac{f}{n-f}}^{t-r}\lambda[r]}=0.
\end{align*}
On the other hand, by definition of $\pth{M[t]-m[t]}$, for each $t\ge 0$ we get $\pth{M[t]-m[t]}\ge 0$. Thus
$$\liminf_{t\diverge} \pth{M[t]-m[t]}\ge 0.$$
Then, we obtain
$$\limsup_{t\diverge}\pth{M[t]-m[t]}\le 0\le \liminf_{t\diverge} \pth{M[t]-m[t]}.$$
Thus, the limit of $\pth{M[t]-m[t]}$ exists and
$$\lim_{t\diverge}\pth{M[t]-m[t]}=0.$$

\eproof
\end{proof}
Recall that $M[t]=\max_{i\in \calN[t]} x_i[t]$ and $m[t]=\min_{i\in \calN[t]} x_i[t]$. Lemma \ref{consensus alg1} implies that asymptotic consensus is achieved under Algorithm 1.
The following lemma is used in the correctness proof of Algorithm 1.
\begin{lemma}
\label{alg1 finite ub crash}
Under Algorithm 1, the following holds.
$$\sum_{t=0}^{\infty} \lambda[t]\pth{M[t]-m[t]}<\infty.$$
\end{lemma}

\begin{proof}
Since
\begin{align*}
\sum_{t=0}^{\infty} \lambda[t]\pth{M[t]-m[t]}=\lambda[0]\pth{M[0]-m[0]}+\sum_{t=1}^{\infty} \lambda[t]\pth{M[t]-m[t]},
\end{align*}
and $\lambda[0]\pth{M[0]-m[0]}<\infty$, 
to show Lemma \ref{alg1 finite ub crash}, it is enough to show that 
$$\sum_{t=1}^{\infty} \lambda[t]\pth{M[t]-m[t]}<\infty.$$

\begin{align}
\label{alg1 finiteness CB}
\nonumber
\sum_{t=1}^{\infty} \lambda[t]\pth{M[t]-m[t]}&\le
\sum_{t=1}^{\infty} \lambda[t]\pth{ \pth{\frac{f}{n-f}}^{t}\pth{M[0]-m[0]}+2L\pth{\sum_{r=0}^{t-1}\pth{\frac{f}{n-f}}^{t-r}\lambda[r]}}~~~\text{by (\ref{alg ub})}\\
\nonumber
&=\pth{M[0]-m[0]}\sum_{t=1}^{\infty} \lambda[t]\pth{\frac{f}{n-f}}^{t}
+2L\sum_{t=1}^{\infty}\sum_{r=0}^{t-1}\pth{\pth{\frac{f}{n-f}}^{t-r}\lambda[r]\lambda[t]}\\
\nonumber
&\overset{(a)}{\le} \pth{M[0]-m[0]}\sum_{t=1}^{\infty} \lambda[t]\pth{\frac{f}{n-f}}^{t}
+L\sum_{t=1}^{\infty}\sum_{r=0}^{t-1}\pth{\pth{\frac{f}{n-f}}^{t-r}\pth{\lambda^2[r]+\lambda^2[t]}}\\
\nonumber
&= \pth{M[0]-m[0]}\sum_{t=1}^{\infty} \lambda[t]\pth{\frac{f}{n-f}}^{t}+L\sum_{t=1}^{\infty}\lambda^2[t]\sum_{r=0}^{t-1}\pth{\frac{f}{n-f}}^{t-r}\\
&\quad +L\sum_{t=1}^{\infty}\sum_{r=0}^{t-1}\pth{\pth{\frac{f}{n-f}}^{t-r}\lambda^2[r]}
\end{align}
Inequality $(a)$ holds because $\lambda[t]\lambda[r]\le \frac{\lambda^2[t]+\lambda^2[r]}{2}$. We bound the three terms in the RHS of (\ref{alg1 finiteness CB}) separately.

\paragraph{The first term of (\ref{alg1 finiteness CB}): }
Since $\lambda[t]\le \lambda [0]$ for each $t\ge 1$, we have
\begin{align}
\label{1CB finiteness t1}
\nonumber
\pth{M[0]-m[0]}\sum_{t=1}^{\infty} \lambda[t]\pth{\frac{f}{n-f}}^{t}&\le
\pth{M[0]-m[0]} \lambda[0]\sum_{t=1}^{\infty}\pth{\frac{f}{n-f}}^{t}\\
\nonumber
&\le \pth{M[0]-m[0]} \lambda[0] \frac{1}{1-\frac{f}{n-f}} \\
&=\pth{M[0]-m[0]} \lambda[0] \frac{n-f}{n-2f}<\infty.
\end{align}

\paragraph{The second term of (\ref{alg1 finiteness CB}):}
\begin{align}
\label{1CB finiteness t2}
\nonumber
L\sum_{t=1}^{\infty}\lambda^2[t]\sum_{r=0}^{t-1}\pth{\frac{f}{n-f}}^{t-r}&=
L\sum_{t=1}^{\infty}\lambda^2[t]\sum_{r=1}^{t}\pth{\frac{f}{n-f}}^{r}\\
\nonumber
&\le L \sum_{t=1}^{\infty}\lambda^2[t]\sum_{r=0}^{\infty}\pth{\frac{f}{n-f}}^{r}\\
\nonumber
&=L\sum_{t=1}^{\infty}\lambda^2[t]\frac{1}{1-\frac{f}{n-f}}\\
\nonumber
&= \frac{n-f}{n-2f}L\sum_{t=1}^{\infty}\lambda^2[t]\\
&<\infty
 \end{align}
The last inequality follows from the fact that $\sum_{t=1}^{\infty}\lambda^2[t]\le \sum_{t=0}^{\infty}\lambda^2[t]<\infty$.

\paragraph{The third term of (\ref{alg1 finiteness CB}):}
For any fixed $T$, we get
\begin{align*}
L\sum_{t=1}^{T}\sum_{r=0}^{t-1}\pth{\pth{\frac{f}{n-f}}^{t-r}\lambda^2[r]}&=L\sum_{r=0}^{T-1} \lambda^2[r] \sum_{t=1}^{T} \pth{\frac{f}{n-f}}^{t}\\
&\le L\sum_{r=0}^{T-1} \lambda^2[r] \sum_{t=0}^{\infty} \pth{\frac{f}{n-f}}^{t} \\
&=\frac{n-f}{n-2f}L\sum_{r=0}^{T-1} \lambda^2[r].
\end{align*}
Let $T\diverge$, we get
\begin{align}
\label{1CB finiteness t3}
L\sum_{t=1}^{\infty}\sum_{r=0}^{t-1}\pth{\pth{\frac{f}{n-f}}^{t-r}\lambda^2[r]}&=\frac{n-f}{n-2f}L\sum_{r=0}^{\infty} \lambda^2[r]~<~\infty.
\end{align}
%
We get
\begin{align*}
\sum_{t=1}^{\infty} \lambda[t] \pth{M[t]-m[t]}&\le
 \pth{M[0]-m[0]}\sum_{t=1}^{\infty} \lambda[t]\pth{\frac{f}{n-f}}^{t}+L\sum_{t=1}^{\infty}\lambda^2[t]\sum_{r=0}^{t-1}\pth{\frac{f}{n-f}}^{t-r}\\
&\quad +L\sum_{t=1}^{\infty}\sum_{r=0}^{t-1}\pth{\pth{\frac{f}{n-f}}^{t-r}\lambda^2[r]}~~\text{by}~(\ref{alg1 finiteness CB})\\
&<\infty+\infty+\infty=\infty ~~~\text{by}~(\ref{1CB finiteness t1}), (\ref{1CB finiteness t2})~ \text{and}~ (\ref{1CB finiteness t3})
\end{align*}
proving the lemma.

\eproof
\end{proof}
By Lemma \ref{alg1 finite ub crash}, we know there exists some constant $C_1$ such that for any constant $t\ge 0$
\begin{align}
\label{alg11}
\sum_{\tau=t}^{\infty}\lambda[\tau]L\pth{M[\tau]-m[\tau]}\le \sum_{\tau=0}^{\infty}\lambda[\tau]L\pth{M[\tau]-m[\tau]}\le C_1.
\end{align}

The following corollary is an immediate consequence of Lemma \ref{alg1 finite ub crash}.
\begin{corollary}
\label{alg1 cor1}
Under Algorithm 1,
\begin{align}
\label{alg1 cor11}
\lim_{t\diverge} \lambda[t]\pth{M[t]-m[t]}=0,
\end{align}
and
\begin{align}
\label{alg1 cor12}
\lim_{t\diverge} \sum_{\tau=t}^{\infty}\lambda[\tau]\pth{M[\tau]-m[\tau]}=0.
\end{align}
\end{corollary}
\begin{proof}
By Lemma \ref{alg1 finite ub crash}, (\ref{alg1 cor11}) holds trivially. Now we prove (\ref{alg1 cor12}).\\

Let $F=\sum_{\tau=0}^{\infty} \lambda[\tau](M[\tau]-m[\tau])$, and let $\{F_t\}_{t=0}^{\infty}$ be a sequence such that for each $t$,
$$F_t=\sum_{\tau=0}^{t-1} \lambda[\tau](M[\tau]-m[\tau]).$$
Since $M[\tau]-m[\tau]\ge 0$ for each $\tau\ge 0$, by construction, it holds that $F_t\le F_{t+1}$ and that
$F_t\le F$ for each $t\ge 0$. Thus, by MCT, we know that
$$\lim_{t\diverge} F_t=F.$$
Now, let
$$R_t\triangleq F-F_t=\sum_{\tau=0}^{\infty} \lambda[\tau](M[\tau]-m[\tau])-\sum_{\tau=0}^{t-1} \lambda[\tau](M[\tau]-m[\tau])=\pth{\sum_{\tau=t}^{\infty} \lambda[\tau](M[\tau]-m[\tau])}.$$ By Lemma \ref{alg1 finite ub crash}, we know that $F<\infty$. Thus the sequence $R_t$ is well-defined.
In addition, since the sequence $F_t$ converges, then the sequence $R_t$ also converges. So, we get
\begin{align*}
\lim_{t\diverge} \pth{\sum_{\tau=t}^{\infty} \lambda[\tau](M[\tau]-m[\tau])}=\lim_{t\diverge} R_t=\lim_{t\diverge} \pth{F-F_t}=F-\lim_{t\diverge} F_t=F-F=0,
\end{align*}
proving (\ref{alg1 cor12}).

\eproof
\end{proof}

\subsubsection{Optimality of Algorithm 1}

\begin{definition}
\label{BS def resilient}
Given a sequence $\{x[t]\}_{t=0}^{\infty}$, a sequence of gradients $\{g[t]\}_{t=0}^{\infty}$, and a set of stepsizes $\{\lambda[t]\}_{t=0}^{\infty}$
we say $x[t]$ is a resilient point with respect to gradient $g[t]$ if one of the following items is true:
\begin{list}{}{}
\item
\hspace*{1in}*
$~ x[t]\in Y~\text{and}~~  \pth{x[t]-\lambda[t] g[t]}\notin Y ,$\\
\hspace*{1in}*
$~ x[t]> \max Y~\text{and}~~\pth{x[t]-\lambda[t] g[t]}< \min Y,$\\
\hspace*{1in}*
$~ x[t]< \min Y~\text{and}~~\pth{x[t]-\lambda[t] g[t]}> \max Y.$
\end{list}
\end{definition}
Since by Lemma \ref{valid closed crash}, we know that set $Y$ is closed. Thus $\max Y$ and $\min Y$ exist, and Definition \ref{BS def resilient} is well-defined over set $Y$.\\

Let $\{z[t]\}_{t=0}^{\infty}$ be a sequence of estimates such that
\begin{align}
\label{crash sequence z}
z[t]=x_{j_{t}}[t], ~~\text{where}~j_{t}\in \argmax_{j\in \calN[t]} Dist\pth{x_j[t], Y}.
\end{align}
From the definition, there is a sequence of agents $\{j_t\}_{t=0}^{\infty}$ associated with the sequence $\{z[t]\}_{t=0}^{\infty}$. 

\begin{lemma}
\label{accu point}
If there exists $c>0$ such that
$$\lim_{t\diverge} Dist\pth{z[t], ~Y}=c,$$
then at least one of the following two statements is true.\\
(A.1) There exists a subsequence $\{z[t_k]\}_{k=0}^{\infty}$ such that $z[t_k]<\min Y$ for all $k\ge 0$.\\
(A.2) There exists a subsequence $\{z[t_k^{\prime}]\}_{k=0}^{\infty}$ such that $z[t_k^{\prime}]>\max Y$ for all $k\ge 0$.\\
In addition, at least one of $(\min Y-c)$ or $(\max Y+c)$ is an accumulation point of $\{z[t]\}_{t=0}^{\infty}$.
\end{lemma}

\begin{proof}
Since $\lim_{t\diverge} Dist\pth{z[t], Y}=c>0$, there exists $m$ 
such that $z[t]\notin Y$ for each $t\ge m$.
Otherwise, there exists a subsequence $\{z[t_k]\}_{k=0}^{\infty}$ such that $z[t_k]\in Y$ for each $k\ge 0$. By definition of $Dist\pth{\cdot, Y}$, we have, $Dist\pth{z[t_k], Y}=0$ for each $k\ge 0$. Then
$$c=\lim_{t\diverge} Dist\pth{z[t_k], Y}=0,$$ contradicting the assumption that $c>0$.

Since $z[t]\notin Y$ for each $t\ge m$, at least one of the following two statements is true.\\
(A.1) There exists a subsequence $\{z[t_k]\}_{k=0}^{\infty}$ such that $z[t_k]<\min Y$ for all $k\ge 0$.\\
(A.2) There exists a subsequence $\{z[t_k^{\prime}]\}_{k=0}^{\infty}$ such that $z[t_k^{\prime}]>\max Y$ for all $k\ge 0$.\\

By symmetry, WLOG, assume (A.1) is true. Then, for each $y\in Y$ and each $k\ge 0$, we have
\begin{align*}
z[t_k]<\min Y\le y.
\end{align*}
Thus,
$$\left | z[t_k]-y\right |=y-z[t_k].$$
Minimizing over $y\in Y$, we have
$$Dist\pth{z[t_k], Y}=\min_{y\in Y}\left | z[t_k]-y\right |=\min_{y\in Y} (y-z[t_k])=\min Y-z[t_k].$$
Thus,
\begin{align}
\label{limit z}
z[t_k]=\min Y-  Dist\pth{z[t_k], Y}.
\end{align}

Recall that the limit of $Dist\pth{z[t], Y}$ exists and $\lim_{t\diverge} Dist\pth{z[t], Y}=c$, and note that $\{Dist\pth{z[t_k], Y} \}_{k=0}^{\infty}$ is a subsequence of $\{Dist\pth{z[t], Y} \}_{t=0}^{\infty}$. Thus, the limit of $Dist\pth{z[t_k], Y}$ exists, and
$$\lim_{k\diverge} Dist\pth{z[t_k], Y}= \lim_{t\diverge} Dist\pth{z[t], Y}=c.$$
Therefore, the limit of $z[t_k]$ exists, and
\begin{align}
\label{BS limit}
\nonumber
\lim_{k\diverge}z[t_k]&= \lim_{k\diverge}\pth{\min Y-  Dist\pth{z[t_k], Y}}\\
\nonumber
&=\min Y- \lim_{k\diverge} Dist\pth{z[t_k], Y}\\
&=\min Y-c.
\end{align}
Thus, $\pth{\min Y -c}$ is an accumulation point of $\{z[t]\}_{t=0}^{\infty}$.\\

Similarly, if (A.2) is true, i.e.,  there exists a subsequence $\{z[t_k^{\prime}]\}_{k=0}^{\infty}$ such that $z[t_k^{\prime}]>\max Y$ for all $k\ge 0$, and we can show that $\pth{\max Y+c}$ is an accumulation point of $\{z[t]\}_{t=0}^{\infty}$.\\

Therefore, Lemma \ref{accu point} has been proved.

\eproof
\end{proof}

Recall that $\{z[t]\}_{t=0}^{\infty}$ is a sequence of estimates such that (\ref{crash sequence z}) holds, and that there is a sequence of agents $\{j_t\}_{t=0}^{\infty}$ associated with the sequence $\{z[t]\}_{t=0}^{\infty}$. 

\begin{lemma}
\label{optimial z}
If
\begin{align}
\label{alg1 z}
\lim_{t\diverge} Dist\pth{z[t], Y}=0,
\end{align}
 then for each non-faulty agent $i$ in $\calN$, the sequence $\{Dist\pth{x_i[t], Y}\}_{t=0}^{\infty}$ converges and $$\lim_{t\diverge} Dist\pth{x_i[t], Y}=0.$$
\end{lemma}
\begin{proof}
For each $i\in \calN$, we have
\begin{align*}
Dist\pth{x_i[t], Y}\le \max_{j\in \calN[t] }Dist\pth{x_j[t], Y}=Dist\pth{z[t], Y}~~~\text{by}~(\ref{crash sequence z})
\end{align*}
Taking limit sup on both sides, we get
$$\limsup_{t\diverge} Dist\pth{x_i[t], Y}\le \limsup_{t\diverge}Dist\pth{z[t], Y}=0~~~\text{by}~(\ref{alg1 z}) $$
Thus, for each $i\in \calN$, the sequence $\{Dist\pth{x_i[t], Y}\}_{t=0}^{\infty}$ converges and $$\lim_{t\diverge} Dist\pth{x_i[t], Y}=0.$$

\eproof
\end{proof}

Lemma \ref{accu point} and Lemma \ref{optimial z} derived in proof of Algorithm 1 apply to Algorithm 2 and Algorithm 3 also.

\begin{theorem}
\label{optimal alg1}
The sequence $\{Dist\pth{z[t], Y}\}_{t=0}^{\infty}$ converges and $$\lim_{t\diverge} Dist\pth{z[t], Y}=0.$$
\end{theorem}
\begin{proof}


Recall that $\calN[t-1]$ is the set of agents that do not crash by the end of iteration $t-1$. There may exists an agent $j$ that crashes during the execution of iteration $t$. If agent $j$ crashes after performing step 3 in Algorithm 1, then $s_j[t]$ is well-defined. On the contrary, if agent $j$ crashes before step 3 is conducted, then $s_j[t]$ is not well-defined. In this case, we define $Dist\pth{s_j[t], Y}=0$ for ease of exposition. With this convention,  $\min_{j\in \calN[t-1]} Dist\pth{s_j[t], Y}$ is well-defined.

Let $j_{t-1}^{\prime}\in \calN[t-1]$ such that
\begin{align}
\label{auxilary 1}
\max_{j\in \calN[t-1]} Dist\pth{s_j[t], Y}=Dist\pth{s_{j_{t-1}^{\prime}}[t], Y}.
\end{align}
We get
\begin{align}
\label{distance Y}
\nonumber
Dist\pth{z[t], Y}&=\max_{j\in \calN[t]}~ Dist\pth{x_j[t], Y}~~~\text{due to}~(\ref{crash sequence z})\\
\nonumber
&=\max_{j\in \calN[t]}~ Dist\pth{ \frac{1}{\left |  \calR_j^{2}[t-1]\right |}\sum_{i\in \calR_j^{2}[t-1]} s_i[t], ~ Y}~~\text{by}~(\ref{update x crash 1})\\
\nonumber
&\le  \max_{j\in \calN[t]}~\frac{1}{\left |  \calR_j^{2}[t-1]\right |} \sum_{i\in \calR_j^{2}[t-1]}Dist\pth{s_i[t], ~ Y} ~~~\text{since $Dist\pth{\cdot, Y}$ is convex}\\
\nonumber
&\le \max_{j\in \calN[t]}~\max_{i\in \calR_j^{2}[t-1]} Dist\pth{s_i[t], Y}\\
\nonumber
&\le \max_{j\in \calN[t-1]} Dist\pth{s_j[t], Y}.\\
\nonumber
&=Dist\pth{s_{j_{t-1}^{\prime}}[t], Y}\\
%
&=\inf_{y\in Y}\left |x_{j_{t-1}^{\prime}}[t-1]-\frac{\lambda[t-1]}{\left |  \calR_{j_{t-1}^{\prime}}^{1}[t-1] \right |}\pth{\sum_{i\in \calR_{j_{t-1}^{\prime}}^{1}[t-1] } h_i^{\prime}(x_{j_{t-1}^{\prime}}[t-1])}-y \right |~~\text{by}~(\ref{update z crash 1}).
\end{align}
%
%
%
Recall that $j_{t}^{\prime}$ is defined as (\ref{auxilary 1}).
Note that for each $t\ge 0$, there exists a non-faulty agent $j_{t}^{\prime}$ such that (\ref{distance Y}) holds, and there exists a sequence of agents $\{j_{t}^{\prime}\}_{t=0}^{\infty}$.
Let $\{x[t]\}_{t=0}^{\infty}$ be a sequence of estimates such that
\begin{align}
\label{d1}
x[t]=x_{j_{t}^{\prime}}[t].
\end{align}
Let $\{g[t]\}_{t=0}^{\infty}$ be a sequence of gradients such that
\begin{align}
\label{d2}
g[t]=\frac{1}{\left |  \calR_{j_{t}^{\prime}}^{1}[t] \right |}\pth{\sum_{i\in \calR_{j_{t}^{\prime}}^{1}[t] } h_i^{\prime}(x_{j_{t}^{\prime}}[t])}.
\end{align}


\paragraph{\bf If $x[t-1]=x_{j_{t-1}^{\prime}}[t-1]$ is a resilient point} with respect to the gradient $g[t-1]$,
 by Definition \ref{BS def resilient}, we can bound (\ref{distance Y}) further as follows
\begin{align}
\label{distance Y1}
Dist\pth{z[t], Y}&\le \inf_{y\in Y}\left |x_{j_{t-1}^{\prime}}[t-1]-\frac{\lambda[t-1]}{\left |  \calR_{j_{t-1}^{\prime}}^{1}[t-1] \right |}\pth{\sum_{i\in \calR_{j_{t-1}^{\prime}}^{1}[t-1] } h_i^{\prime}(x_{j_{t-1}^{\prime}}[t-1])}-y \right |\le \lambda[t-1] L.
\end{align}

\paragraph{\bf If $x_{j_{t-1}^{\prime}}[t-1]$ is not a resilient point} with respect to the gradient $g[t-1]$,  then from Definition \ref{BS def resilient}, we know that
\begin{itemize}
\item[$B 1$:] if $x_{j_{t-1}^{\prime}}[t-1]\in Y$, then $x_{j_{t-1}^{\prime}}[t-1]-\frac{\lambda[t-1]}{\left |  \calR_{j_{t-1}^{\prime}}^{1}[t-1] \right |}\pth{\sum_{i\in \calR_{j_{t-1}^{\prime}}^{1}[t-1] } h_i^{\prime}(x_{j_{t-1}^{\prime}}[t-1])}\in Y,$
\item[$B 2$:] if $x_{j_{t-1}^{\prime}}[t-1]< \min Y$, then $x_{j_{t-1}^{\prime}}[t-1]-\frac{\lambda[t-1]}{\left |  \calR_{j_{t-1}^{\prime}}^{1}[t-1] \right |}\pth{\sum_{i\in \calR_{j_{t-1}^{\prime}}^{1}[t-1] } h_i^{\prime}(x_{j_{t-1}^{\prime}}[t-1])}\le \max Y,$
\item[$B 3$:] if $x_{j_{t-1}^{\prime}}[t-1]> \max Y$, then $x_{j_{t-1}^{\prime}}[t-1]-\frac{\lambda[t-1]}{\left |  \calR_{j_{t-1}^{\prime}}^{1}[t-1] \right |}\pth{\sum_{i\in \calR_{j_{t-1}^{\prime}}^{1}[t-1] } h_i^{\prime}(x_{j_{t-1}^{\prime}}[t-1])}\ge \min Y$.
\end{itemize}

\vskip \baselineskip

We consider two scenarios: scenario 1  $$x_{j_{t-1}^{\prime}}[t-1]-\frac{\lambda[t-1]}{\left |  \calR_{j_{t-1}^{\prime}}^{1}[t-1] \right |}\pth{\sum_{i\in \calR_{j_{t-1}^{\prime}}^{1}[t-1] } h_i^{\prime}(x_{j_{t-1}^{\prime}}[t-1])}~\in~ Y,$$ and scenario 2
$$x_{j_{t-1}^{\prime}}[t-1]-\frac{\lambda[t-1]}{\left |  \calR_{j_{t-1}^{\prime}}^{1}[t-1] \right |}\pth{\sum_{i\in \calR_{j_{t-1}^{\prime}}^{1}[t-1] } h_i^{\prime}(x_{j_{t-1}^{\prime}}[t-1])}~\notin~ Y.$$
The first scenario can possibly appear in each of $B 1, B 2,$ and $B 3$.
In contrast, the second scenario can only appear in $B 2$ and $B 3$.
\paragraph{Scenario 1: }
Assume that
$$x_{j_{t-1}^{\prime}}[t-1]-\frac{\lambda[t-1]}{\left |  \calR_{j_{t-1}^{\prime}}^{1}[t-1] \right |}\pth{\sum_{i\in \calR_{j_{t-1}^{\prime}}^{1}[t-1] } h_i^{\prime}(x_{j_{t-1}^{\prime}}[t-1])}~\in~ Y,$$
it holds that
$$
\inf_{y\in Y}\left |x_{j_{t-1}^{\prime}}[t-1]-\frac{\lambda[t-1]}{\left |  \calR_{j_{t-1}^{\prime}}^{1}[t-1] \right |}\pth{\sum_{i\in \calR_{j_{t-1}^{\prime}}^{1}[t-1] } h_i^{\prime}(x_{j_{t-1}^{\prime}}[t-1])}-y \right |~=~0~\le~ Dist\pth{z[t-1],Y}.
$$
Thus, (\ref{distance Y}) can be further bounded as
\begin{align}
\nonumber
Dist\pth{z[t], Y}&\le \inf_{y\in Y}\left |x_{j_{t-1}^{\prime}}[t-1]-\frac{\lambda[t-1]}{\left |  \calR_{j_{t-1}^{\prime}}^{1}[t-1] \right |}\pth{\sum_{i\in \calR_{j_{t-1}^{\prime}}^{1}[t-1] } h_i^{\prime}(x_{j_{t-1}^{\prime}}[t-1])}-y \right |\\
&~=~0~\label{distance Y2'}\\
&\le~ Dist\pth{z[t-1],Y}\label{distance Y2}.
\end{align}

\paragraph{Scenario 2: }
Assume that
$$x_{j_{t-1}^{\prime}}[t-1]-\frac{\lambda[t-1]}{\left |  \calR_{j_{t-1}^{\prime}}^{1}[t-1] \right |}\pth{\sum_{i\in \calR_{j_{t-1}^{\prime}}^{1}[t-1] } h_i^{\prime}(x_{j_{t-1}^{\prime}}[t-1])}~\notin~ Y=[\min Y, \max Y].$$
As commented earlier, either $B 2$ holds or $B 3$ holds. In addition, from the assumption of scenario 2, $B 2$ and $B 3$ can be further refined as follows.

\begin{itemize}
\item[$B 2^{\prime}$:]
$x_{j_{t-1}^{\prime}}[t-1]<\min Y~~~\text{and}~~~ x_{j_{t-1}^{\prime}}[t-1]-\frac{\lambda[t-1]}{\left |  \calR_{j_{t-1}^{\prime}}^{1}[t-1] \right |}\pth{\sum_{i\in \calR_{j_{t-1}^{\prime}}^{1}[t-1] } h_i^{\prime}(x_{j_{t-1}^{\prime}}[t-1])}~<\min Y$
\item[$B 3^{\prime}$:]
$x_{j_{t-1}^{\prime}}[t-1]>\max Y~~~\text{and}~~~ x_{j_{t-1}^{\prime}}[t-1]-\frac{\lambda[t-1]}{\left |  \calR_{j_{t-1}^{\prime}}^{1}[t-1] \right |}\pth{\sum_{i\in \calR_{j_{t-1}^{\prime}}^{1}[t-1] } h_i^{\prime}(x_{j_{t-1}^{\prime}}[t-1])}~>\max Y$
\end{itemize}

Suppose $B 2^{\prime}$ is true. As $x_{j_{t-1}^{\prime}}[t-1]<\min Y$, and $\frac{1}{\left |  \calR_{j_{t-1}^{\prime}}^{1}[t-1] \right |}\pth{\sum_{i\in \calR_{j_{t-1}^{\prime}}^{1}[t-1] } h_i^{\prime}(x_{j_{t-1}^{\prime}}[t-1])}$ is the gradient of a valid function at point $x_{j_{t-1}^{\prime}}[t-1]$, from the definition of set $Y$, we know that
\begin{align}
\label{negative gradient}
\frac{1}{\left |  \calR_{j_{t-1}^{\prime}}^{1}[t-1] \right |}\pth{\sum_{i\in \calR_{j_{t-1}^{\prime}}^{1}[t-1] } h_i^{\prime}(x_{j_{t-1}^{\prime}}[t-1])}<0.
\end{align}
In addition, since
$$x_{j_{t-1}^{\prime}}[t-1]-\frac{\lambda[t-1]}{\left |  \calR_{j_{t-1}^{\prime}}^{1}[t-1] \right |}\pth{\sum_{i\in \calR_{j_{t-1}^{\prime}}^{1}[t-1] } h_i^{\prime}(x_{j_{t-1}^{\prime}}[t-1])}~<\min Y,$$
it holds that for any $y\in Y$
\begin{align}
\label{r3}
\nonumber
&\left |   x_{j_{t-1}^{\prime}}[t-1]-\frac{\lambda[t-1]}{\left |  \calR_{j_{t-1}^{\prime}}^{1}[t-1] \right |}\pth{\sum_{i\in \calR_{j_{t-1}^{\prime}}^{1}[t-1] } h_i^{\prime}(x_{j_{t-1}^{\prime}}[t-1])} -y\right |\\
\nonumber
&=y-  x_{j_{t-1}^{\prime}}[t-1]+\frac{\lambda[t-1]}{\left |  \calR_{j_{t-1}^{\prime}}^{1}[t-1] \right |}\pth{\sum_{i\in \calR_{j_{t-1}^{\prime}}^{1}[t-1] } h_i^{\prime}(x_{j_{t-1}^{\prime}}[t-1])}\\
\nonumber
&=\left |y-  x_{j_{t-1}^{\prime}}[t-1]\right |+\frac{\lambda[t-1]}{\left |  \calR_{j_{t-1}^{\prime}}^{1}[t-1] \right |}\pth{\sum_{i\in \calR_{j_{t-1}^{\prime}}^{1}[t-1] } h_i^{\prime}(x_{j_{t-1}^{\prime}}[t-1])}\\
&=\left |y-  x_{j_{t-1}^{\prime}}[t-1]\right |-\lambda[t-1]\left|\frac{1}{\left |  \calR_{j_{t-1}^{\prime}}^{1}[t-1] \right |}\pth{\sum_{i\in \calR_{j_{t-1}^{\prime}}^{1}[t-1] } h_i^{\prime}(x_{j_{t-1}^{\prime}}[t-1])}\right|~~~\text{by}~(\ref{negative gradient})
\end{align}


Similarly, we can show that (\ref{r3}) still holds for the case when $B 3^{\prime}$ is true. Henceforth, we refer (\ref{r3}) as the relation holds for both $B 2^{\prime}$ and $B 3^{\prime}$, i.e., holds under scenario 2.\\

Thus, under scenario 2, we can bound (\ref{distance Y}) as follows
%

\begin{align}
Dist\pth{z[t], Y}&\le \inf_{y\in Y}\left |x_{j_{t-1}^{\prime}}[t-1]-\frac{\lambda[t-1]}{\left |  \calR_{j_{t-1}^{\prime}}^{1}[t-1] \right |}\pth{\sum_{i\in \calR_{j_{t-1}^{\prime}}^{1}[t-1] } h_i^{\prime}(x_{j_{t-1}^{\prime}}[t-1])}-y \right |~~~ \text{by } (\ref{distance Y})\nonumber\\
&= \inf_{y\in Y} \left |y-  x_{j_{t-1}^{\prime}}[t-1]\right |-\lambda[t-1]\left|\frac{1}{\left |  \calR_{j_{t-1}^{\prime}}^{1}[t-1] \right |}\pth{\sum_{i\in \calR_{j_{t-1}^{\prime}}^{1}[t-1] } h_i^{\prime}(x_{j_{t-1}^{\prime}}[t-1])}\right|~~~\text{by}~(\ref{r3})\nonumber\\
&=Dist\pth{x_{j_{t-1}^{\prime}}[t-1], Y}-\lambda[t-1]\left|\frac{1}{\left |  \calR_{j_{t-1}^{\prime}}^{1}[t-1] \right |}\pth{\sum_{i\in \calR_{j_{t-1}^{\prime}}^{1}[t-1] } h_i^{\prime}(x_{j_{t-1}^{\prime}}[t-1])}\right| \nonumber\\
&\le Dist\pth{z[t-1], Y}-\lambda[t-1]\left|\frac{1}{\left |  \calR_{j_{t-1}^{\prime}}^{1}[t-1] \right |}\pth{\sum_{i\in \calR_{j_{t-1}^{\prime}}^{1}[t-1] } h_i^{\prime}(x_{j_{t-1}^{\prime}}[t-1])}\right|\label{distance Y3}\\
&\le Dist\pth{z[t-1], Y} \label{distance Y4}.
\end{align}

The last inequality follows from the fact that
$$Dist\pth{x_{j_{t-1}^{\prime}}[t-1], Y}\le \max_{j\in \calN[t-1]} Dist\pth{x_j[t-1], Y}=Dist\pth{z[t-1], Y}.$$


\vskip 1.5\baselineskip
Combining the above analysis for the case when $x_{j_{t-1}^{\prime}}[t-1]$ is a resilient point or the case when $x_{j_{t-1}^{\prime}}[t-1]$ is not a resilient point,  by (\ref{distance Y1}), (\ref{distance Y2}) and (\ref{distance Y4}), we obtain the following iteration relation
\begin{align}
\label{basic iteration alg1}
Dist\pth{z[t],Y}\le \max \left\{\lambda[t-1] L, ~Dist\pth{z[t-1],Y}\right\}.
\end{align}

\hrule

\vskip 2\baselineskip

Recall from (\ref{d1}) and (\ref{d2}) that $x[t-1]=x_{j_{t-1}^{\prime}}[t-1]$ and $g[t-1]=\frac{1}{\left |  \calR_{j_{t-1}^{\prime}}^{1}[t-1] \right |}\pth{\sum_{i\in \calR_{j_{t-1}^{\prime}}^{1}[t-1] } h_i^{\prime}(x_{j_{t-1}^{\prime}}[t-1])}$.
We consider two cases : case (i) there are infinitely many points in  $\{x[t]\}_{t=0}^{\infty}$ that are resilient with respect to $\{g[t]\}_{t=0}^{\infty}$, and case (ii) there are finitely many points in  $\{x[t]\}_{t=0}^{\infty}$ that are resilient with respect to $\{g[t]\}_{t=0}^{\infty}$, respectively.

%
%

\paragraph{{\bf Case (i):}} There are infinitely many points in  $\{x[t]\}_{t=0}^{\infty}$ that are resilient with respect to $\{g[t]\}_{t=0}^{\infty}$.

Let $\{t_i\}_{i=0}^{\infty}$ be the maximal sequence of such indices. Since $x[t_i]$ is a resilient point with respect to $g[t_i]$ for each $i$, then for each $t_i$, by (\ref{distance Y1}), we get
\begin{align}
\label{e1}
Dist\pth{z[t_i+1],Y}\le \lambda[t_i] L,
\end{align}
and for each $t\not=t_i$ for any $i$, by (\ref{distance Y2}) and (\ref{distance Y4}), we get
\begin{align}
\label{e2}
Dist\pth{z[t+1],Y}
\le Dist\pth{z[t],Y}.
\end{align}
Taking limit sup on both sides of (\ref{e1}) over $i$, we get
\begin{align}
\label{alg1 casei}
0\le \limsup_{i\diverge} ~Dist\pth{z[t_i+1],Y}&\le \limsup_{i\diverge}\lambda[t_i] L=0.
\end{align}
%
For each $\tau> t_0$ and $\tau \notin \{t_i\}_{i=0}^{\infty}$, there exists $t_{i(\tau)}$ such that $t_{i(\tau)}< \tau\le t_{i(\tau)+1}$. Then, we get
\begin{align}
\label{alg1 casei'}
Dist\pth{z[\tau], Y}&\le~Dist\pth{z[t_{i(\tau)}+1],Y}~~~\text{due to}~(\ref{e2})~\text{and that}~ \tau\ge t_{i(\tau)}+1\\
&\le \lambda[t_{i(\tau)}] L~~~\text{by}~(\ref{e1})
\end{align}
Taking the limit sup on both sides of (\ref{alg1 casei'}) over $\tau$, where $\tau>t_0$ and $\tau \notin \{t_i\}_{i=0}^{\infty}$, we get
\begin{align}
\label{alg1 caseii}
\limsup_{\tau\diverge} Dist\pth{z[\tau], Y}\le \limsup_{\tau \diverge}\lambda[t_{i(\tau)}] L=\lim_{\tau\diverge}\lambda[t_{i(\tau)}] L =0.
\end{align}

\vskip \baselineskip

From (\ref{alg1 casei}), we know that $\forall \, \epsilon>0, \exists \, i_0$ such that for all $i\ge i_0$, the following holds.
\begin{align}
\label{converge 1}
\sup \{Dist\pth{z[t_j], Y}, \, t_j\in \{t_i\}_{i=0}^{\infty}, \, j\ge i_0\}=\left  | \sup \{Dist\pth{z[t_j], Y}, \, t_j\in \{t_i\}_{i=0}^{\infty}, \, j\ge i_0\}-0\right |< \epsilon.
\end{align}
From (\ref{alg1 caseii}), we know that $\forall \, \epsilon>0, \exists \, \tau^*, \tau^*\notin \{t_i\}_{i=0}^{\infty}$ such that for all $\tau\ge \tau^*, \tau\notin  \{t_i\}_{i=0}^{\infty}$, the following holds.
\begin{align}
\label{converge 2}
\sup \{Dist\pth{z[\tau], Y}, \, \tau\ge \tau^*, \tau \notin  \{t_i\}_{i=0}^{\infty}\}=\left  | \sup \{Dist\pth{z[\tau], Y}, \, \tau\ge \tau^*, \tau \notin  \{t_i\}_{i=0}^{\infty}\}-0\right |< \epsilon.
\end{align}
Let $t^*=\max\{t_{i_0}, \tau^*\}$. Then for each $\epsilon>0$ and $t\ge t^*$, we have
\begin{align*}
&\sup \{Dist\pth{z[t], Y}, \, t\ge t^*\}\\
&=\sup \pth{\{Dist\pth{z[t], Y}, \, t\in \{t_i\}_{i=0}^{\infty}, t\ge t_{i_0}\}\cup \{Dist\pth{z[t], Y}, \, t\notin \{t_i\}_{i=0}^{\infty}, t\ge \tau^*\}}\\
&=\max \left\{\sup\{Dist\pth{z[t], Y}, \, t\in \{t_i\}_{i=0}^{\infty}, t\ge t_{i_0}\}, \sup\{Dist\pth{z[t], Y}, \, t\notin \{t_i\}_{i=0}^{\infty}, t\ge \tau^*\}\right \}\\
&<\max \{\epsilon, \epsilon \}=\epsilon ~~~\text{by (\ref{converge 1}) and (\ref{converge 2})}.
\end{align*}
Thus, we have
$$ \limsup_{t\diverge} Dist\pth{z[t], Y}=0.$$

Therefore, the limit of $Dist\pth{z[t], Y}$ exists, and
\begin{align*}
\lim_{t\diverge} Dist\pth{z[t], Y}=0.
\end{align*}

\paragraph{{\bf Case (ii):}} There are finitely many points in  $\{x[t]\}_{t=0}^{\infty}$ that are resilient with respect to $\{g[t]\}_{t=0}^{\infty}$.

By the assumption of case (ii) we know that there exists a time index $m_0$ such that for all $t\ge m_0$, each $x[t]$ is not a resilient point with respect to $g[t]$. Then, for $t\ge m_0$, (\ref{e2}) holds. Thus, by MCT, the limit of $Dist\pth{z[t], Y}$ exists. Let $c\ge 0$ be a nonnegative constant such that
\begin{align}
\label{crash case 2 limit alg1}
\lim_{t\diverge} Dist\pth{z[t], Y}=c.
\end{align}
Since $Dist\pth{z[t], Y}\le Dist\pth{z[m_0], Y}$ holds for each $t\ge m_0$, we know that $c<\infty$. \\

\paragraph{Case (ii.a):}
Assume that there are infinitely many time indices $t\ge m_0$ such that
$$x[t]-\lambda[t] g[t]~=~x_{j_{t}^{\prime}}[t]-\frac{\lambda[t]}{\left |  \calR_{j_{t}^{\prime}}^{1}[t] \right |}\pth{\sum_{i\in \calR_{j_{t}^{\prime}}^{1}[t] } h_i^{\prime}(x_{j_{t}^{\prime}}[t])}~\in ~Y.$$
Let $\{t_k\}_{k=0}^{\infty}$ be the maximal sequence of such time indices. By (\ref{distance Y2'}), we have
\begin{align*}
Dist\pth{z[t_k+1], Y}\le 0.
\end{align*}
Thus, the limit of $Dist\pth{z[t_k+1], Y}$ exists and
$$\lim_{k\diverge} Dist\pth{z[t_k+1], Y}=0.$$

Recall from (\ref{crash case 2 limit alg1}) that the limit of $Dist\pth{z[t], Y}$ exists. The limit of $Dist\pth{z[t], Y}$ and the limit of $Dist\pth{z[t_k+1], Y}$ should be identical, i.e.,
\begin{align}
\label{lim bs in}
c=\lim_{t\diverge} Dist\pth{z[t], Y}=\lim_{k\diverge} Dist\pth{z[t_k+1], Y}=0,
\end{align}
proving the theorem.\\

\paragraph{Case (ii.b):}
Assume that there are only finitely many time indices $t\ge m_0$ such that
$$x[t]-\lambda[t] g[t]~=~x_{j_{t}^{\prime}}[t]-\frac{\lambda[t]}{\left |  \calR_{j_{t}^{\prime}}^{1}[t] \right |}\pth{\sum_{i\in \calR_{j_{t}^{\prime}}^{1}[t] } h_i^{\prime}(x_{j_{t}^{\prime}}[t])}~\in ~Y.$$
Then, there exists\footnote{Recall that $m_0$ is the time index such that for each $t\ge m_0$, $x[t]$ is not a resilient point with respect to $g[t]$.} $m^{\prime}\ge m_0$ such that for each $t\ge m^{\prime}\ge m_0$, $x[t]$ is not a resilient point with respect to $g[t]$, and $x[t]-\lambda[t] g[t]\notin Y$.
Thus, for each $t\ge m^{\prime}\ge m_0$, (\ref{distance Y3}) holds, i.e.,
\begin{align*}
Dist\pth{z[t+1], Y}\le Dist\pth{z[t], Y}-\lambda[t]\left|\frac{1}{\left |  \calR_{j_{t}^{\prime}}^{1}[t] \right |}\pth{\sum_{i\in \calR_{j_{t}^{\prime}}^{1}[t] } h_i^{\prime}(x_{j_{t}^{\prime}}[t])}\right|.
\end{align*}


Recall that $0\le c< \infty$ is a nonnegative constant such that
$\lim_{t\diverge}~Dist\pth{z[t],Y}=c$.
{\bf Next we show that $c=0$.} We prove this by contradiction. Suppose $c>0$.
By Lemma \ref{accu point}, we know that either (A.1) is true or (A.2) is true.\\
(A.1) There exists a subsequence $\{z[t_k]\}_{k=0}^{\infty}$ such that $z[t_k]<\min Y$ for all $k\ge 0$.\\
(A.2) There exists a subsequence $\{z[t_k^{\prime}]\}_{k=0}^{\infty}$ such that $z[t_k^{\prime}]>\max Y$ for all $k\ge 0$.\\
We also know that at least one of $(\min Y-c)$ or $(\max Y+c)$ is an accumulation point of $\{z[t]\}_{t=0}^{\infty}$, and no other accumulation points exist.

Let $a=\min Y$, $b=\max Y$ and $\epsilon=\frac{c}{2}$. 
%
%
It can be seen from the proof of Lemma \ref{accu point} that there exists $m$ such that $z[t]\notin Y$ for each $t\ge m$.
We consider three scenarios: (A.1) is true but (A.2) is not true, (A.2) is true but (A.1) is not true, both (A.1) and (A.2) are true.
\paragraph{{\bf When (A.1) holds but (A.2) does not hold:}} That is, there exists a subsequence $\{z[t_k]\}_{k=0}^{\infty}$ such that $z[t_k]<\min Y$ for all $k\ge 0$; and there does not exist a subsequence $\{z[t_k^{\prime}]\}_{k=0}^{\infty}$ such that $z[t_k^{\prime}]>\max Y$ for all $k\ge 0$. Then there exists $m_1\ge m$ such that $z[t]<\min Y$ for each $t\ge m_1\ge m$. From the proof of Lemma \ref{accu point}, we know
\begin{align}
\label{A1-A2}
\lim_{t\diverge} z[t]~=~\min Y -c~=~a-c.
\end{align}
Since (\ref{A1-A2}) holds, there exists $m_1^*\ge m_1\ge m$ such that for all $t\ge m_1^*\ge m_1\ge m$, the following holds.
\begin{align}
\label{limit trapped A1-A2}
|z[t]-\pth{a-c}|\le \epsilon=\frac{c}{2}~~~\iff~~~a-\frac{3c}{2}\le z[t]\le a-\frac{c}{2}.
\end{align}
Since $c>0$, we have $a-\frac{c}{2}<a$. Then, for each $p(\cdot)\in \calC$, $p^{\prime}(a-\frac{c}{2})<0$. Thus,
$$\rho^*\triangleq \sup_{p(\cdot)\in \calC} p^{\prime}(a-\frac{c}{2})\le 0.$$
Let $K=\sum_{j\in \calF} {\bf 1}\{h_j^{\prime}(a-\frac{c}{2})\ge 0\}$.
Define $q(x)$ as follows,
$$q(x)=\frac{1}{|\calN|+K}\pth{\sum_{j\in \calN} h_j(x)+\sum_{j\in \calF} h_j(x){\bf 1}\{h_j^{\prime}(a-\frac{c}{2})\ge 0\}}.$$
It can be easily seen that $q(\cdot)\in \calC$ is a valid function and $$\rho^*=\sup_{p(\cdot)\in \calC} p^{\prime}(a-\frac{c}{2})=q^{\prime}(a-\frac{c}{2})<0.$$

Note that when $t\ge m_1^*\ge m_1\ge m$, (\ref{distance Y3}) may not hold, since it is possible that $z[t]-\lambda[t]g[t]\in Y$. Let $\tilde{t}_1 =\max\{m_1^*, m^{\prime}\}$. For each $t\ge \tilde{t}_1 =\max\{m_1^*, m^{\prime}\}$, (\ref{distance Y3}), (\ref{A1-A2}) and (\ref{limit trapped A1-A2}) hold. We have
\begin{align}
\nonumber
Dist\pth{z[t+1], Y}&\le Dist\pth{z[t], Y}-\lambda[t]\left|\frac{1}{\left |  \calR_{j_{t}^{\prime}}^{1}[t] \right |}\pth{\sum_{i\in \calR_{j_{t}^{\prime}}^{1}[t] } h_i^{\prime}(x_{j_{t}^{\prime}}[t])}\right|~~~\text{by}~(\ref{distance Y3})\\
\nonumber
&=Dist\pth{z[t], Y}-\lambda[t]\left|\frac{1}{\left |  \calR_{j_{t}^{\prime}}^{1}[t] \right |}\pth{\sum_{i\in \calR_{j_{t}^{\prime}}^{1}[t] } h_i^{\prime}(z[t])}+\frac{1}{\left |  \calR_{j_{t}^{\prime}}^{1}[t] \right |}\sum_{i\in \calR_{j_{t}^{\prime}}^{1}[t] } \pth{h_i^{\prime}(x_{j_{t}^{\prime}}[t])-h_i^{\prime}(z[t])}\right|\\
\nonumber
&\overset{(a)}{\le} Dist\pth{z[t], Y}-\lambda[t]\left|\frac{1}{\left |  \calR_{j_{t}^{\prime}}^{1}[t] \right |}\pth{\sum_{i\in \calR_{j_{t}^{\prime}}^{1}[t] } h_i^{\prime}(z[t])}\right |+\lambda[t](M[t]-m[t])L\\
&\le Dist\pth{z[t], Y}-\lambda[t]|\rho^*|+\lambda[t](M[t]-m[t])L.\label{case 2 a1}
\end{align}
Inequality $(a)$ holds because gradient $h_k^{\prime}(\cdot)$ is $L$--Lipschitz for each $k\in \calV$,
$$\left |x_{j_{t+1}^{\prime}}[t]-x_k[t]\right |\le \max_{i,\, j\in \calN[t]} \pth{x_i[t]-x_j[t]}=\max_{i\in \calN[t]} x_i[t]- \min_{j\in \calN[t]}x_j[t]=M[t]-m[t],$$
and the fact that
$$\frac{1}{\left | \calR_{j_{t}^{\prime}}[t]\right |}\sum_{k\in \calR_{j_{t}^{\prime}}[t]}1=1.$$
Next we show that the last inequality holds. Since $h_i^{\prime}(\cdot)$ is non-decreasing for each $i\in \calV$, then the function
$$\frac{1}{\left |  \calR_{j_{t}^{\prime}}^{1}[t] \right |}\pth{\sum_{i\in \calR_{j_{t}^{\prime}}^{1}[t] } h_i^{\prime}(\cdot)}$$
is non-decreasing. In addition, by (\ref{limit trapped A1-A2}) we know that
$a-\frac{3c}{2}\le z[t]\le a-\frac{c}{2}$. We get
\begin{align*}
\frac{1}{\left |  \calR_{j_{t}^{\prime}}^{1}[t] \right |}\pth{\sum_{i\in \calR_{j_{t}^{\prime}}^{1}[t] } h_i^{\prime}(z[t])}\le \frac{1}{\left |  \calR_{j_{t}^{\prime}}^{1}[t] \right |}\pth{\sum_{i\in \calR_{j_{t}^{\prime}}^{1}[t] } h_i^{\prime}(a-\frac{c}{2})}\le \sup_{p(\cdot)\in \calC} p^{\prime}(a-\frac{c}{2})=\rho^*<0.
\end{align*}
Thus,
\begin{align}
\label{g1}
\left|\frac{1}{\left |  \calR_{j_{t}^{\prime}}^{1}[t] \right |}\pth{\sum_{i\in \calR_{j_{t}^{\prime}}^{1}[t] } h_i^{\prime}(z[t])}\right|\ge |\rho^*|,
\end{align}
proving the last inequality in (\ref{case 2 a1}).  Repeatedly apply (\ref{case 2 a1}) for $t\ge \tilde{t}_1 =\max\{m_1^*, m^{\prime}\}$, we get

\begin{align}
\label{case 2 a11}
Dist\pth{z[t+1], Y}\le Dist\pth{z[\tilde{t}_1], ~Y}-\pth{\sum_{r=\tilde{t}_1}^{t}\lambda[r]}|\rho^*|+\sum_{r=\tilde{t}_1}^{t}\lambda[r](M[r]-m[r])L.
\end{align}
%
%
Taking limit on both sides of (\ref{case 2 a11}), we obtain
\begin{align}
\label{f1}
\nonumber
\lim_{t\diverge}Dist\pth{z[t+1], Y}&\le Dist\pth{z[\tilde{t}_1], ~Y}-\pth{\sum_{r=\tilde{t}_1}^{\infty}\lambda[r]}|\rho^*|+\sum_{r=\tilde{t}_1}^{\infty}\lambda[r](M[r]-m[r])L\\
\nonumber
&\le Dist\pth{z[\tilde{t}_1], ~Y}-\pth{\sum_{r=\tilde{t}_1}^{\infty}\lambda[r]}|\rho^*|+\sum_{r=0}^{\infty}\lambda[r](M[r]-m[r])L\\
\nonumber
&\overset{(a)}{=} Dist\pth{z[\tilde{t}_1], ~Y}-\infty+C_1\\
&=-\infty.
\end{align}
Equality $(a)$ is true due to (\ref{alg11}),  the fact that $|\rho^*|>0$ and that
\begin{align*}
\sum_{r=\tilde{t}_1}^{\infty}\lambda[r]=\sum_{t=0}^{\infty}\lambda[t]-\sum_{r=0}^{\tilde{t}_1-1}\lambda[r]=
\infty-\sum_{r=0}^{\tilde{t}_1-1}\lambda[r]=\infty.
\end{align*}
On the other hand, we know $ \lim_{t\diverge}Dist\pth{z[t], Y}= c>0$. This is a contradiction.
Thus,
$$ \lim_{t\diverge} Dist\pth{z[t], Y}=c=0.$$

\paragraph{{\bf When (A.2) holds but (A.1) does not hold:}} That is, there does not exist a subsequence $\{z[t_k]\}_{k=0}^{\infty}$ such that $z[t_k]<\min Y$ for all $k\ge 0$; and there exists a subsequence $\{z[t_k^{\prime}]\}_{k=0}^{\infty}$ such that $z[t_k^{\prime}]>\max Y$ for all $k\ge 0$. Recall that $z[t]\notin Y$ for each $t\ge m$. Then there exists $m_2\ge m$ such that $z[t]>\max Y$ for each $m_2$. From the proof of Lemma \ref{accu point}, we get
\begin{align}
\label{A2-A1}
\lim_{t\diverge} z[t]~=~\max Y +c=b+c.
\end{align}

Since (\ref{A2-A1}) holds, there exists $m_2^*\ge m_2\ge m$ such that for all $t\ge m_2^*\ge m_2\ge m$, the following holds.
\begin{align}
\label{limit trapped A2-A1}
|z[t]-\pth{b+c}|\le \epsilon=\frac{c}{2}~~~\iff~~~b+\frac{c}{2}\le z[t]\le b+\frac{3c}{2}.
\end{align}
Since $c>0$, we have $b+\frac{c}{2}>b$. Then, for each $p(\cdot)\in \calC$, $p^{\prime}(b+\frac{c}{2})>0$. Then,
$$\tilde{\rho}\triangleq \inf_{p(\cdot)\in \calC} p^{\prime}(b+\frac{c}{2})\ge 0.$$
Let $K=\sum_{j\in \calF} {\bf 1}\{h_j^{\prime}(b+\frac{c}{2})\le 0\}$.
Define $\tilde{q}(x)$ as follows,
$$\tilde{q}(x)=\frac{1}{|\calN|+K}\pth{\sum_{j\in \calN} h_j(x)+\sum_{j\in \calF} h_j(x){\bf 1}\{h_j^{\prime}(b+\frac{c}{2})\le 0\}}.$$
It can be easily seen that $\tilde{q}(\cdot)\in \calC$ is a valid function and $$\inf_{p(\cdot)\in \calC} p^{\prime}(b+\frac{c}{2})={\tilde{q}}^{\prime}(b+\frac{c}{2})>0.$$
Then, $\tilde{\rho}={\tilde{q}}^{\prime}(b+\frac{c}{2})>0$.\\

Note that when $t\ge m_2^*\ge m_2\ge m$, (\ref{distance Y3}) may not hold, since it is possible that $z[t]-\lambda[t]g[t]\in Y$. Let $\tilde{t}_2 =\max\{m_2^*, m^{\prime}\}$. For each $t\ge \tilde{t}_2 =\max\{m_2^*, m^{\prime}\}$, (\ref{distance Y3}), (\ref{A2-A1}) and (\ref{limit trapped A2-A1}) hold. We have
\begin{align}
\label{case 2 a2}
\nonumber
Dist\pth{z[t+1], Y}&\le Dist\pth{z[t], Y}-\lambda[t]\left|\frac{1}{\left |  \calR_{j_{t}^{\prime}}^{1}[t] \right |}\pth{\sum_{i\in \calR_{j_{t}^{\prime}}^{1}[t] } h_i^{\prime}(x_{j_{t}^{\prime}}[t])}\right|~~~\text{by}~(\ref{distance Y3})\\
\nonumber
&\le Dist\pth{z[t], Y}-\lambda[t]\left|\frac{1}{\left |  \calR_{j_{t}^{\prime}}^{1}[t] \right |}\pth{\sum_{i\in \calR_{j_{t}^{\prime}}^{1}[t] } h_i^{\prime}(z[t])}\right |+\lambda[t](M[t]-m[t])L\\
&\le Dist\pth{z[t], Y}-\lambda[t]|\tilde{\rho}|+\lambda[t]L(M[t]-m[t]).
\end{align}
Next we show that the last inequality holds. Recall that the function
$$\frac{1}{\left |  \calR_{j_{t}^{\prime}}^{1}[t] \right |}\pth{\sum_{i\in \calR_{j_{t}^{\prime}}^{1}[t] } h_i^{\prime}(\cdot)}$$
is non-decreasing. 
We get
\begin{align}
\label{g2}
\nonumber
\frac{1}{\left |  \calR_{j_{t}^{\prime}}^{1}[t] \right |}\pth{\sum_{i\in \calR_{j_{t}^{\prime}}^{1}[t] } h_i^{\prime}(z[t])}&\ge \frac{1}{\left |  \calR_{j_{t}^{\prime}}^{1}[t] \right |}\pth{\sum_{i\in \calR_{j_{t}^{\prime}}^{1}[t] } h_i^{\prime}(b+\frac{c}{2})}~~~\text{by}~(\ref{limit trapped A2-A1})\\
&\ge \inf_{p(\cdot)\in \calC} p^{\prime}(b+\frac{c}{2})=\tilde{\rho}>0,
\end{align}
proving the last inequality in (\ref{case 2 a2}).  Repeatedly apply (\ref{case 2 a2}) for $t\ge \tilde{t}_2 =\max\{m_2^*, m^{\prime}\}$, we get

\begin{align}
\label{case 2 a21}
Dist\pth{z[t+1], Y}\le Dist\pth{z[\tilde{t}_2], ~Y}-\pth{\sum_{r=\tilde{t}_2}^{t}\lambda[r]}|\tilde{\rho}|+\sum_{r=\tilde{t}_2}^{t}\lambda[r]L(M[r]-m[r]).
\end{align}
%

%

Taking limit on both sides of (\ref{case 2 a21}), we obtain
\begin{align*}
\lim_{t\diverge}Dist\pth{z[t+1], Y}&\le Dist\pth{z[\tilde{t}_2], ~Y}-\pth{\sum_{r=\tilde{t}_2}^{\infty}\lambda[r]}|\tilde{\rho}|+\sum_{r=\tilde{t}_2}^{\infty}\lambda[r](M[r]-m[r])L\\
\nonumber
&\le Dist\pth{z[\tilde{t}_2], ~Y}-\pth{\sum_{r=\tilde{t}_2}^{\infty}\lambda[r]}|\tilde{\rho}|+\sum_{r=0}^{\infty}\lambda[r](M[r]-m[r])L\\
&\le Dist\pth{z[\tilde{t}_2], ~Y}-\infty+C_1\\
&=-\infty.
\end{align*}
This inequality is obtained similarly to the inequality (\ref{f1}).
On the other hand, we know $ \lim_{t\diverge}Dist\pth{z[t], Y}= c>0$. This is a contradiction.
Thus,
$$ \lim_{t\diverge} Dist\pth{z[t], Y}=c=0.$$

\paragraph{{\bf Both (A.1) and (A.2) hold:}}
Let $\{z[t_k]\}_{k=0}^{\infty}$ be a maximal subsequence of $\{z[t]\}_{t=0}^{\infty}$ such that $t_k\ge m^{\prime}$ and $z[t_k]<\min Y$ for all $k\ge 0$. Let $\{z[t_k^{\prime}]\}_{k=0}^{\infty}$ be a maximal subsequence of $\{z[t]\}_{t=0}^{\infty}$ such that $t_k^{\prime}\ge m^{\prime}$ and $t_k^{\prime}>\max Y$ for all $k\ge 0$.
Recall that $z[t]\notin Y$ for each $t\ge m$. Then,
$$
\{z[t_k]\}_{k=0}^{\infty}\cup \{z[t_k^{\prime}]\}_{k=0}^{\infty}=\{z[t]\}_{t\ge m}^{\infty}.$$

By Lemma \ref{accu point}, we know
\begin{align}
\label{A1+A2}
\lim_{k\diverge} z[t_k]=\min Y-c=a-c~~~~\text{and}~~~\lim_{k\diverge} z[t_k^{\prime}]=\max Y+c=b+c.
\end{align}
Since $$\{z[t_k]\}_{k=0}^{\infty}\cup \{z[t_k^{\prime}]\}_{k=0}^{\infty}=\{z[t]\}_{t\ge m}^{\infty},$$
 there exist $m_3\ge m$ such that for each $t\ge m_3$,
\begin{align}
\label{limit trapped A1+A2}
a-\frac{3c}{2}\le z[t]\le a-\frac{c}{2}~~~\text{or}~~b+\frac{c}{2}\le z[t]\le b+\frac{3c}{2}.
\end{align}
Recall that
\begin{align*}
\rho^*= \sup_{p(\cdot)\in \calC} p^{\prime}(a-\frac{c}{2})=q^{\prime}(a-\frac{c}{2})~~~\text{and}~~~\tilde{\rho}\triangleq \inf_{p(\cdot)\in \calC} p^{\prime}(b+\frac{c}{2})={\tilde{q}}^{\prime}(b+\frac{c}{2}).
\end{align*}

Recall that for each $t\ge m^{\prime}\ge m_0,$ $z[t]$ is not a resilient point with respect to $g[t]$, and $z[t]-\lambda[t]g[t]\notin Y$. Thus (\ref{distance Y3}) holds, i.e.,
\begin{align*}
Dist\pth{z[t+1], Y}\le Dist\pth{z[t], Y}-\lambda[t]\left|\frac{1}{\left |  \calR_{j_{t}^{\prime}}^{1}[t] \right |}\pth{\sum_{i\in \calR_{j_{t}^{\prime}}^{1}[t] } h_i^{\prime}(x_{j_{t}^{\prime}}[t])}\right|.
\end{align*}

Since (\ref{limit trapped A1+A2}), by (\ref{g1}) and (\ref{g2}), we get
\begin{align*}
\left |\frac{1}{\left |  \calR_{j_{t}^{\prime}}^{1}[t] \right |}\pth{\sum_{i\in \calR_{j_{t}^{\prime}}^{1}[t] } h_i^{\prime}(z[t])}\right |\ge |\rho^*|~~~\text{or}~~~\left |\frac{1}{\left |  \calR_{j_{t}^{\prime}}^{1}[t] \right |}\pth{\sum_{i\in \calR_{j_{t}^{\prime}}^{1}[t] } h_i^{\prime}(z[t])}\right |\ge |\tilde{\rho}|.
\end{align*}
Thus,
\begin{align}
\label{g3}
~\left |\frac{1}{\left |  \calR_{j_{t}^{\prime}}^{1}[t] \right |}\pth{\sum_{i\in \calR_{j_{t}^{\prime}}^{1}[t] } h_i^{\prime}(z[t])}\right |\ge \min\{|\rho^*|, |\tilde{\rho}|\}.
\end{align}

We get
\begin{align}
\label{case 2 a3}
\nonumber
Dist\pth{z[t+1], Y}&\le Dist\pth{z[t], Y}-\lambda[t]\left|\frac{1}{\left |  \calR_{j_{t}^{\prime}}^{1}[t] \right |}\pth{\sum_{i\in \calR_{j_{t}^{\prime}}^{1}[t] } h_i^{\prime}(x_{j_{t}^{\prime}}[t])}\right|~~~\text{by}~(\ref{distance Y3})\\
\nonumber
&\le Dist\pth{z[t], Y}-\lambda[t]\left|\frac{1}{\left |  \calR_{j_{t}^{\prime}}^{1}[t] \right |}\pth{\sum_{i\in \calR_{j_{t}^{\prime}}^{1}[t] } h_i^{\prime}(z[t])}\right |+\lambda[t](M[t]-m[t])L\\
&\le Dist\pth{z[t], Y}-\lambda[t]\min\{|\rho^*|, |\tilde{\rho}|\}+\lambda[t]L(M[t]-m[t])~~~\text{by}~(\ref{g3})
\end{align}

Repeatedly apply (\ref{case 2 a3}) for $t\ge \tilde{t}_3 =\max\{m_3^*, m^{\prime}\}$, we get
\begin{align}
\label{case 2 a31}
Dist\pth{z[t+1], Y}\le Dist\pth{z[\tilde{t}_3], ~Y}-\pth{\sum_{r=\tilde{t}_3}^{t}\lambda[r]}\min\{|\rho^*|, |\tilde{\rho}|\} +\sum_{r=\tilde{t}_3}^{t}\lambda[r]L(M[r]-m[r]).
\end{align}

Taking limit on both sides of (\ref{case 2 a3}), we obtain
\begin{align*}
\lim_{t\diverge}Dist\pth{z[t+1], Y}&\le Dist\pth{z[\tilde{t}_3], ~Y}-\pth{\sum_{r=\tilde{t}_3}^{\infty}\lambda[r]}\min\{|\rho^*|, |\tilde{\rho}|\}+\sum_{r=\tilde{t}_3}^{\infty}\lambda[r]L(M[r]-m[r])\\
&\le Dist\pth{z[\tilde{t}_3], ~Y}-\infty+C_1\\
&=-\infty.
\end{align*}
This inequality is obtained similarly to the inequality (\ref{f1}).

On the other hand, we know $\lim_{t\diverge}Dist\pth{z[t], Y}= c>0$. A contradiction is proved.
Thus,
$$ \lim_{t\diverge} Dist\pth{z[t], Y}=c=0.$$

The proof is complete.

\eproof
\end{proof}

\subsection{Algorithm 2}

In Algorithm 1, in each iteration $t\ge 1$, there are two rounds of information exchange. Next we will present a simple algorithm which only requires one message sent by each agent per iteration. In this algorithm, each agent $j$ maintains one local estimate $x_j$, where $x_j[0]$ is an arbitrary input at agent $j$.

%

\paragraph{}
\hrule
~
\vspace*{4pt}

{\bf Algorithm 2} for agent $j$ for iteration $t\ge 1$:
~
\vspace*{4pt}\hrule

\begin{list}{}{}

\item[{\bf Step 1:}]
Compute $h_j^{\prime}(x_j[t-1])$-- the gradient of local function $h_j(\cdot)$ at point $x_j[t-1]$, and send the tuple $\pth{x_j[t-1],h_j^{\prime}(x_j[t-1])}$ to all the agents (including agent $j$ itself).\\
~
\item[{\bf Step 2:}]
Let $R_j[t-1]$ denote the set of tuples of the form $\pth{x_i[t-1],h_i^{\prime}(x_i[t-1])}$ received as a result of step 1. 
Update $x_j$ as
\begin{align}
\label{update x crash 2}
x_j[t]=\frac{1}{\left | \calR_j[t-1]\right |} \pth{\sum_{i\in \calR_j[t-1]} \pth{x_i[t-1]-\lambda[t-1] h_i^{\prime}(x_i[t-1])}}.
\end{align}
~

\end{list}

~
\hrule

~

~

Note that $\calR_j[t-1]\subseteq \calN[t-1]$. In addition, set $Y$ is the same as that defined earlier for Algorithm 1.
\begin{lemma}
\label{consensus alg2}
Under Algorithm 2, the sequence $\{M[t]-m[t]\}_{t=0}^{\infty}$ converges and
$$\lim_{t\diverge} \pth{M[t]-m[t]}=0.$$
\end{lemma}
Recall that $M[t]=\max_{i\in \calN[t]} x_i[t]$ and $m[t]=\min_{i\in \calN[t]} x_i[t]$. Lemma \ref{consensus alg2} implies that asymptotic consensus is achieved under Algorithm 2. The proof of Lemma \ref{consensus alg2} is similar to the proof of Lemma \ref{consensus alg1}, and is omitted. 

\begin{lemma}
\label{finite ub crash}
Under Algorithm 2, the following holds.
$$\sum_{t=0}^{\infty} \lambda[t]\pth{M[t]-m[t]}<\infty.$$
\end{lemma}
The proof of Lemma \ref{finite ub crash} is the similar to the proof of Lemma \ref{alg1 finite ub crash}, and is omitted.
By Lemma \ref{finite ub crash}, we know there exists some constant $C_2$ such that for any constant $t\ge 0$,
\begin{align}
\label{alg22}
\sum_{\tau=t}^{\infty}\lambda[\tau]L\pth{M[\tau]-m[\tau]}\le \sum_{\tau=0}^{\infty}\lambda[\tau]L\pth{M[\tau]-m[\tau]}\le C_2.
\end{align}
The following corollary is an immediate consequence of Lemma \ref{finite ub crash}.
\begin{corollary}
\label{cor1}
Under Algorithm 2,
$$
\lim_{t\diverge} \lambda[t]\pth{M[t]-m[t]}=0,
$$
and
$$
\lim_{t\diverge} \sum_{\tau=t}^{\infty}\lambda[\tau]\pth{M[\tau]-m[\tau]}=0.
$$
\end{corollary}
The proof of Corollary \ref{cor1} is similar to the proof of Corollary \ref{alg1 cor1}, and is omitted.

In our convergence analysis, we will use the well-know ``almost supermartingale" convergence theorem in \cite{Robbins1985}, which can also be found as Lemma 11, in Chapter 2.2 \cite{poljak1987introduction}. We present a simpler deterministic version of the theorem in the next lemma.

\begin{lemma}\cite{Robbins1985}
\label{SB stochatic convergence}
Let $\{ a_t\}_{t=0}^{\infty}, \{ b_t\}_{t=0}^{\infty}$, and $\{ c_t\}_{t=0}^{\infty}$ be non-negative sequences. Suppose that
$$a_{t+1}\le a_{t}-b_{t}+c_{t}~~~~~\text{for all}~t\ge 0,$$
and $\sum_{t=0}^{\infty} c_t<\infty$. Then $\sum_{t=0}^{\infty} b_t<\infty$ and the sequence $\{a_t\}_{t=0}^{\infty}$ converges to a non-negative value.
\end{lemma}

Recall that set $Y$ is the same as that defined earlier for Algorithm 1.
We define $z[t]$ and $x_{j_t}$ similar to that for Algorithm 1. In particular, let $\{z[t]\}_{t=0}^{\infty}$ be a sequence of estimates such that
\begin{align}
\label{alg2 crash sequence z}
z[t]=x_{j_{t}}[t], ~~\text{where}~j_{t}\in \argmax_{j\in \calN[t]} Dist\pth{x_j[t], Y}.
\end{align}
From the definition, there is a sequence of agents $\{j_t\}_{t=0}^{\infty}$ associated with the sequence $\{z[t]\}_{t=0}^{\infty}$.

\begin{theorem}
\label{optimal alg2}
The sequence $\{Dist\pth{z[t], Y}\}_{t=0}^{\infty}$ converges and $$\lim_{t\diverge} Dist\pth{z[t], Y}=0.$$
\end{theorem}

\begin{proof}

We first try to derive an iteration relation similar to that in (\ref{basic iteration alg1}).
%
%
\begin{align}
\label{alg2 iter1}
\nonumber
&Dist\pth{z[t+1], Y}=Dist\pth{x_{j_{t+1}}[t+1], Y}~~~\text{by (\ref{alg2 crash sequence z})}\\
\nonumber
&=Dist\pth{\frac{1}{\left | \calR_{j_{t+1}}[t]\right |}\sum_{i\in \calR_{j_{t+1}}[t]} \pth{x_i[t]-\lambda[t]h_i^{\prime}(x_i[t])},~ Y}~~~\text{by}~(\ref{update x crash 2})\\
\nonumber
&=Dist\pth{\frac{1}{\left | \calR_{j_{t+1}}[t]\right |}\sum_{i\in \calR_{j_{t+1}}[t]} \pth{x_i[t]-
\lambda[t]\frac{1}{\left | \calR_{j_{t+1}}[t]\right |}\sum_{k\in \calR_{j_{t+1}}[t]}h_k^{\prime}(x_k[t])},~ Y}\\
\nonumber
&\le  \frac{1}{\left | \calR_{j_{t+1}}[t]\right |}\sum_{i\in \calR_{j_{t+1}}[t]} \, Dist\pth{x_i[t]-
\lambda[t]\frac{1}{\left | \calR_{j_{t+1}}[t]\right |}\sum_{k\in \calR_{j_{t+1}}[t]}h_k^{\prime}(x_k[t]), ~Y}~~~\text{by convexity of}~Dist\pth{\cdot, Y}\\
&\le \max_{i\in \calR_{j_{t+1}}[t]}Dist\pth{x_i[t]-
\lambda[t]\frac{1}{\left | \calR_{j_{t+1}}[t]\right |}\sum_{k\in \calR_{j_{t+1}}[t]}h_k^{\prime}(x_k[t]), ~Y}
\end{align}

Let
\begin{align}
\label{aux 2}
j_{t+1}^{\prime}\in \argmax_{i\in \calR_{j_{t+1}}[t]}
Dist\pth{x_i[t]-
\lambda[t]\frac{1}{\left | \calR_{j_{t+1}}[t]\right |}\sum_{k\in \calR_{j_{t+1}}[t]}h_k^{\prime}(x_k[t]), ~Y}.
\end{align}

Note that $j_{t+1}^{\prime}\in \calR_{j_{t+1}}[t]\subseteq \calN[t]$, i.e.,  $j_{t+1}^{\prime}\in \calN[t]$.

We get
\begin{align}
\label{alg2 iter}
\nonumber
Dist\pth{z[t+1], Y}&\le \max_{i\in \calR_{j_{t+1}}[t]}Dist\pth{x_i[t]-
\lambda[t]\frac{1}{\left | \calR_{j_{t+1}}[t]\right |}\sum_{k\in \calR_{j_{t+1}}[t]}h_k^{\prime}(x_k[t]), ~Y}~~~\text{by}~(\ref{alg2 iter1})\\
\nonumber
&=Dist\pth{x_{j_{t+1}^{\prime}}[t]-\lambda[t]\frac{1}{\left | \calR_{j_{t+1}}[t]\right |}\sum_{k\in \calR_{j_{t+1}}[t]}h_k^{\prime}(x_k[t]), ~Y}~~~\text{by}~(\ref{aux 2})\\
&=\inf_{y\in Y} \left |    x_{j_{t+1}^{\prime}}[t]-\lambda[t]\frac{1}{\left | \calR_{j_{t+1}}[t]\right |}\sum_{k\in \calR_{j_{t+1}}[t]}h_k^{\prime}(x_k[t])-y      \right |.
\end{align}
For each $y\in Y$, we have
\begin{align}
\label{crash distance y2}
\nonumber
&\left |    x_{j_{t+1}^{\prime}}[t]-\lambda[t]\frac{1}{\left | \calR_{j_{t+1}}[t]\right |}\sum_{k\in \calR_{j_{t+1}}[t]}h_k^{\prime}(x_k[t])-y      \right |\\
\nonumber
&=\left |  x_{j_{t+1}^{\prime}}[t]-\lambda[t]\frac{1}{\left | \calR_{j_{t+1}}[t]\right |}\sum_{k\in \calR_{j_{t+1}}[t]}h_k^{\prime}(x_{j_{t+1}^{\prime}}[t])-y + \lambda[t]\frac{1}{\left | \calR_{j_{t+1}}[t]\right |}\sum_{k\in \calR_{j_{t+1}}[t]}\pth{h_k^{\prime}(x_{j_{t+1}^{\prime}}[t])-h_k^{\prime}(x_k[t])} \right | \\
\nonumber
&\le \left | x_{j_{t+1}^{\prime}}[t]-\lambda[t]\frac{1}{\left | \calR_{j_{t+1}}[t]\right |}\sum_{k\in \calR_{j_{t+1}}[t]}h_k^{\prime}(x_{j_{t+1}^{\prime}}[t])-y \right |+ \left |\lambda[t]\frac{1}{\left | \calR_{j_{t+1}}[t]\right |}\sum_{k\in \calR_{j_{t+1}}[t]}\pth{h_k^{\prime}(x_{j_{t+1}^{\prime}}[t])-h_k^{\prime}(x_k[t])} \right |\\
\nonumber
\nonumber
&\le \left | x_{j_{t+1}^{\prime}}[t]-\lambda[t]\frac{1}{\left | \calR_{j_{t+1}}[t]\right |}\sum_{k\in \calR_{j_{t+1}}[t]}h_k^{\prime}(x_{j_{t+1}^{\prime}}[t])-y \right |+ \lambda[t]\frac{1}{\left | \calR_{j_{t+1}}[t]\right |}\sum_{k\in \calR_{j_{t+1}}[t]}\left | h_k^{\prime}(x_{j_{t+1}^{\prime}}[t])-h_k^{\prime}(x_k[t]) \right |\\
\nonumber
&\overset{(a)}{\le} \left | x_{j_{t+1}^{\prime}}[t]-\lambda[t]\frac{1}{\left | \calR_{j_{t+1}}[t]\right |}\sum_{k\in \calR_{j_{t+1}}[t]}h_k^{\prime}(x_{j_{t+1}^{\prime}}[t])-y \right |+ \lambda[t]\frac{1}{\left | \calR_{j_{t+1}}[t]\right |}\sum_{k\in \calR_{j_{t+1}}[t]} L\left |x_{j_{t+1}^{\prime}}[t]-x_k[t] \right |\\
&\overset{(b)}{\le}\left | x_{j_{t+1}^{\prime}}[t]-\lambda[t]\frac{1}{\left | \calR_{j_{t+1}}[t]\right |}\sum_{k\in \calR_{j_{t+1}}[t]}h_k^{\prime}(x_{j_{t+1}^{\prime}}[t])-y \right |+\lambda[t] L (M[t]-m[t]).
\end{align}

Inequality $(a)$ holds because gradient $h_k^{\prime}(\cdot)$ is $L$--Lipschitz for each $k\in \calV$. Inequality $(b)$ holds from the fact that
$$\left |x_{j_{t+1}^{\prime}}[t]-x_k[t]\right |\le \max_{i,\, j\in \calN[t]} \pth{x_i[t]-x_j[t]}=\max_{i\in \calN[t]} x_i[t]- \min_{j\in \calN[t]}x_j[t]=M[t]-m[t],$$
and that
$$\frac{1}{\left | \calR_{j_{t+1}}[t]\right |}\sum_{k\in \calR_{j_{t+1}}[t]}1=1.$$ 
Using (\ref{crash distance y2}), the inequality (\ref{alg2 iter}) can be further bounded as
\begin{align}
\label{alg2 iter2}
\nonumber
Dist\pth{z[t+1], Y}&\le \inf_{y\in Y} \left |    x_{j_{t+1}^{\prime}}[t]-\lambda[t]\frac{1}{\left | \calR_{j_{t+1}}[t]\right |}\sum_{k\in \calR_{j_{t+1}}[t]}h_k^{\prime}(x_k[t])-y      \right |\\
&\le \inf_{y\in Y}\left | x_{j_{t+1}^{\prime}}[t]-\lambda[t]\frac{1}{\left | \calR_{j_{t+1}}[t]\right |}\sum_{k\in \calR_{j_{t+1}}[t]}h_k^{\prime}(x_{j_{t+1}^{\prime}}[t])-y \right |+\lambda[t] L (M[t]-m[t])~~~\text{by}~(\ref{crash distance y2}).
\end{align}

\vskip 2\baselineskip
Note that for each $t\ge 0$, there exists a non-faulty agent $j_{t+1}^{\prime}$ such that (\ref{crash distance y2}) holds, and there exists a sequence of agents $\{j_{t}^{\prime}\}_{t=1}^{\infty}$.
Let $\{x[t]\}_{t=0}^{\infty}$ be a sequence of estimates such that
\begin{align}
\label{alg2 a1}
x[t]=x_{j_{t+1}^{\prime}}[t].
\end{align}
Let $\{g[t]\}_{t=0}^{\infty}$ be a sequence of gradients such that
\begin{align}
\label{alg2 a2}
g[t]=\frac{1}{\left | \calR_{j_{t+1}}[t]\right |}\sum_{k\in \calR_{j_{t+1}}[t]}h_k^{\prime}(x_{j_{t+1}^{\prime}}[t]).
\end{align}
\paragraph{\bf If $x[t]=x_{j_{t+1}^{\prime}}[t]$ is a resilient point} with respect to the gradient $g[t]=\frac{1}{\left | \calR_{j_{t+1}}[t]\right |}\sum_{k\in \calR_{j_{t+1}}[t]}h_k^{\prime}(x_{j_{t+1}^{\prime}}[t]))$, by Definition \ref{BS def resilient}, we bound (\ref{alg2 iter2}) further as
\begin{align}
\label{alg2 distance Y1}
\nonumber
&Dist\pth{z[t+1], Y}\le \inf_{y\in Y}\left |x_{j_{t+1}^{\prime}}[t]-\lambda[t]\frac{1}{\left | \calR_{j_{t+1}}[t]\right |}\sum_{k\in \calR_{j_{t+1}}[t]}h_k^{\prime}(x_{j_{t+1}^{\prime}}[t]))-y\right |+L \lambda[t] \pth{M[t]-m[t]} \\
&\le L \lambda[t]+L \lambda[t] \pth{M[t]-m[t]}.
\end{align}

\paragraph{\bf If $x[t]=x_{j_{t+1}^{\prime}}[t]$ is not a resilient point}
 with respect to the gradient $g[t]=\frac{1}{\left | \calR_{j_{t+1}}[t]\right |}\sum_{k\in \calR_{j_{t+1}}[t]}h_k^{\prime}(x_{j_{t+1}^{\prime}}[t])$, then from Definition \ref{BS def resilient}, we know that
 \begin{itemize}
\item[$C1$:] if $x_{j_{t+1}^{\prime}}[t]\in Y$, then $x_{j_{t+1}^{\prime}}[t]-\frac{\lambda[t]}{\left | \calR_{j_{t+1}}[t]\right |}\sum_{k\in \calR_{j_{t+1}}[t]}h_k^{\prime}(x_{j_{t+1}^{\prime}}[t])\in Y,$
\item[$C2$:] if $x_{j_{t+1}^{\prime}}[t]< \min Y$, then $x_{j_{t+1}^{\prime}}[t]-\frac{\lambda[t]}{\left | \calR_{j_{t+1}}[t]\right |}\sum_{k\in \calR_{j_{t+1}}[t]}h_k^{\prime}(x_{j_{t+1}^{\prime}}[t])\le \max Y,$
\item[$C3$:] if $x_{j_{t+1}^{\prime}}[t]> \max Y$, then $x_{j_{t+1}^{\prime}}[t]-\frac{\lambda[t]}{\left | \calR_{j_{t+1}}[t]\right |}\sum_{k\in \calR_{j_{t+1}}[t]}h_k^{\prime}(x_{j_{t+1}^{\prime}}[t])\ge \min Y$.
\end{itemize}

\vskip \baselineskip

We consider two scenarios: scenario 1
$$x_{j_{t+1}^{\prime}}[t]-\frac{\lambda[t]}{\left | \calR_{j_{t+1}}[t]\right |}\sum_{k\in \calR_{j_{t+1}}[t]}h_k^{\prime}(x_{j_{t+1}^{\prime}}[t])\in Y,$$
and scenario 2
$$x_{j_{t+1}^{\prime}}[t]-\frac{\lambda[t]}{\left | \calR_{j_{t+1}}[t]\right |}\sum_{k\in \calR_{j_{t+1}}[t]}h_k^{\prime}(x_{j_{t+1}^{\prime}}[t])\notin Y.$$
The first scenario can possibly appear in each of $C1, C2,$ and $C3$. In contrast, the second scenario can only appear in $C2$ and $C3$.

\paragraph{Scenario 1:} Assume that $$x_{j_{t+1}^{\prime}}[t]-\frac{\lambda[t]}{\left | \calR_{j_{t+1}}[t]\right |}\sum_{k\in \calR_{j_{t+1}}[t]}h_k^{\prime}(x_{j_{t+1}^{\prime}}[t])\in Y, $$
it holds that
$$\inf_{y\in Y} \left | x_{j_{t+1}^{\prime}}[t]-\frac{\lambda[t]}{\left | \calR_{j_{t+1}}[t]\right |}\sum_{k\in \calR_{j_{t+1}}[t]}h_k^{\prime}(x_{j_{t+1}^{\prime}}[t])-y\right |=0\le Dist\pth{z[t], Y}. $$
Thus, (\ref{alg2 iter2}) can be bounded as
\begin{align}
\nonumber
&Dist\pth{z[t+1], Y}\le \inf_{y\in Y}\left |x_{j_{t+1}^{\prime}}[t]-\lambda[t]\frac{1}{\left | \calR_{j_{t+1}}[t]\right |}\sum_{k\in \calR_{j_{t+1}}[t]}h_k^{\prime}(x_{j_{t+1}^{\prime}}[t]))-y\right |+L \lambda[t] \pth{M[t]-m[t]} \\
&\le 0+L \lambda[t] \pth{M[t]-m[t]}\label{alg2 distance Y2'}\\
&\le Dist\pth{z[t], Y}+L \lambda[t] \pth{M[t]-m[t]}\label{alg2 distance Y2}.
\end{align}

\paragraph{Scenario 2:}
Assume that
$$x_{j_{t+1}^{\prime}}[t]-\frac{\lambda[t]}{\left | \calR_{j_{t+1}}[t]\right |}\sum_{k\in \calR_{j_{t+1}}[t]}h_k^{\prime}(x_{j_{t+1}^{\prime}}[t])~\notin~ Y=[\min Y, \max Y].$$
As commented earlier, either $C 2$ holds or $C 3$ holds. In addition, from the assumption of scenario 2, $C 2$ and $C 3$ can be further refined as follows.

\begin{itemize}
\item[$C 2^{\prime}$:]
$x_{j_{t+1}^{\prime}}[t]<\min Y~~~\text{and}~~~ x_{j_{t+1}^{\prime}}[t]-\frac{\lambda[t]}{\left | \calR_{j_{t+1}}[t]\right |}\sum_{k\in \calR_{j_{t+1}}[t]}h_k^{\prime}(x_{j_{t+1}^{\prime}}[t])~<\min Y$
\item[$C 3^{\prime}$:]
$x_{j_{t+1}^{\prime}}[t]>\max Y~~~\text{and}~~~ x_{j_{t+1}^{\prime}}[t]-\frac{\lambda[t]}{\left | \calR_{j_{t+1}}[t]\right |}\sum_{k\in \calR_{j_{t+1}}[t]}h_k^{\prime}(x_{j_{t+1}^{\prime}}[t])~>\max Y$
\end{itemize}

Similar to (\ref{negative gradient}), it can be shown that for both $C 2^{\prime}$ and $C 3^{\prime}$, the following holds.
\begin{align}
\label{alg2 negative gradient}
\left | x_{j_{t+1}^{\prime}}[t]-\frac{\lambda[t]}{\left | \calR_{j_{t+1}}[t]\right |}\sum_{k\in \calR_{j_{t+1}}[t]}h_k^{\prime}(x_{j_{t+1}^{\prime}}[t])-y\right |=\left | x_{j_{t+1}^{\prime}}[t]-y\right |-\lambda[t]\left |\frac{1}{\left | \calR_{j_{t+1}}[t]\right |}\sum_{k\in \calR_{j_{t+1}}[t]}h_k^{\prime}(x_{j_{t+1}^{\prime}}[t])\right |.
\end{align}

Thus, under scenario 2, we can bound (\ref{alg2 iter2}) as
\begin{align}
\nonumber
&Dist\pth{z[t+1], Y}\le \inf_{y\in Y}\left |x_{j_{t+1}^{\prime}}[t]-\lambda[t]\frac{1}{\left | \calR_{j_{t+1}}[t]\right |}\sum_{k\in \calR_{j_{t+1}}[t]}h_k^{\prime}(x_{j_{t+1}^{\prime}}[t]))-y\right |+L \lambda[t] \pth{M[t]-m[t]} \\
\nonumber
&= \inf_{y\in Y}\left | x_{j_{t+1}^{\prime}}[t]-y\right |-\lambda[t]\left |\frac{1}{\left | \calR_{j_{t+1}}[t]\right |}\sum_{k\in \calR_{j_{t+1}}[t]}h_k^{\prime}(x_{j_{t+1}^{\prime}}[t])\right |+L \lambda[t] \pth{M[t]-m[t]}~~~\text{by}~(\ref{alg2 negative gradient})\\
\nonumber
&=Dist\pth{x_{j_{t+1}^{\prime}}[t], Y}-\lambda[t]\left |\frac{1}{\left | \calR_{j_{t+1}}[t]\right |}\sum_{k\in \calR_{j_{t+1}}[t]}h_k^{\prime}(x_{j_{t+1}^{\prime}}[t])\right |+L \lambda[t] \pth{M[t]-m[t]}\\
&\le Dist\pth{z[t], Y}-\lambda[t]\left |\frac{1}{\left | \calR_{j_{t+1}}[t]\right |}\sum_{k\in \calR_{j_{t+1}}[t]}h_k^{\prime}(x_{j_{t+1}^{\prime}}[t])\right |+L \lambda[t] \pth{M[t]-m[t]}\label{alg2 distance Y3}\\
&\le Dist\pth{z[t], Y}+L \lambda[t] \pth{M[t]-m[t]}\label{alg2 distance Y4}.
\end{align}

By (\ref{alg2 distance Y1}), (\ref{alg2 distance Y2}) and (\ref{alg2 distance Y4}), for each $t\ge 0$, we obtain the following iteration relation
\begin{align}
\label{basic iteration alg2 crash}
Dist\pth{z[t+1], Y}\le \max \left \{\lambda[t] L, ~Dist\pth{z[t],Y}\right\}+\lambda[t]L\pth{M[t]-m[t]}.
\end{align}

\hrule

\vskip 2\baselineskip
Recall (\ref{alg2 a1}), (\ref{alg2 a2}) that $x[t]=x_{j_{t+1}^{\prime}}[t]$ and $g[t]=\frac{1}{\left | \calR_{j_{t+1}}[t]\right |}\sum_{k\in \calR_{j_{t+1}}[t]}h_k^{\prime}(x_{j_{t+1}^{\prime}}[t])$. Similar to the proof of Theorem \ref{optimal alg1}, we consider two cases : case (i) there are infinitely many points in  $\{x[t]\}_{t=0}^{\infty}$ that are resilient with respect to $\{g[t]\}_{t=0}^{\infty}$, and case (ii) there are finitely many points in  $\{x[t]\}_{t=0}^{\infty}$ that are resilient with respect to $\{g[t]\}_{t=0}^{\infty}$, respectively.

\paragraph{{\bf Case (i):}} There are infinitely many points in  $\{x[t]\}_{t=0}^{\infty}$ that are resilient with respect to $\{g[t]\}_{t=0}^{\infty}$.

Let $\{t_i\}_{i=0}^{\infty}$ be the maximal sequence of such indices. Since $x[t_i]$ is a resilient point with respect to $g[t]$ for each $i$, then for each $t_i$, by (\ref{alg2 distance Y1}), we have
\begin{align}
\label{crash case 1 s1}
Dist\pth{z[t_i+1],Y}\le \lambda[t_i] L+\lambda[t_i]L\pth{M[t_i]-m[t_i]},
\end{align}
and for each $t\not= t_i\, \forall i$, by (\ref{alg2 distance Y2}) and (\ref{alg2 distance Y4}), we get
\begin{align}
\label{crash case 1 s2}
Dist\pth{z[t+1],Y}\le Dist\pth{z[t], Y}+\lambda[t]L\pth{M[t]-m[t]},
\end{align}

Taking limit sup on both sides of (\ref{crash case 1 s1}), we get
\begin{align}
\label{alg2 casei}
\nonumber
\limsup_{i\diverge} ~Dist\pth{z[t_i+1],Y}&\le \limsup_{i\diverge}\lambda[t_i] L+\limsup_{i\diverge}\lambda[t_i]L\pth{M[t_i]-m[t_i]}\\
&=~0+0=0~~~\text{by Corollary \ref{cor1}}.
\end{align}
In addition, $\liminf_{i\diverge} ~Dist\pth{z[t_i+1],Y}\ge 0$. Thus, the limit of $Dist\pth{z[t_i+1],Y}$ exists, and
$$\lim_{i\diverge} ~Dist\pth{z[t_i+1],Y}=0.$$
For each $\tau>t_0$ and $\tau \notin \{t_i\}_{i=0}^{\infty}$, there exists $t_{i(\tau)}$ such that $t_{i(\tau)}< \tau< t_{i(\tau)+1}$. Repeatedly applying (\ref{crash case 1 s2}), we get
\begin{align}
\label{alg2 a3}
\nonumber
Dist\pth{z[\tau+1], Y}
&\le ~Dist\pth{z[t_{i(\tau)}+1],Y}+\sum_{r=t_{i(\tau)}+1}^{\tau}\lambda[r]L\pth{M[r]-m[r]}\\
\nonumber
&\le \lambda[t_{i(\tau)}] L+\lambda[t_{i(\tau)}]\pth{M[t_{i(\tau)}]-m[t_{i(\tau)}]}L +\sum_{r=t_{i(\tau)}+1}^{\tau} \lambda[r]\pth{M[r]-m[r]}L~~~\text{by}~(\ref{crash case 1 s1})\\
\nonumber
&=\lambda[t_{i(\tau)}] L +\sum_{r=t_{i(\tau)}}^{\tau} \lambda[r]\pth{M[r]-m[r]}L\\
&\le \lambda[t_{i(\tau)}] L +\sum_{r=t_{i(\tau)}}^{\infty} \lambda[r]\pth{M[r]-m[r]}L ~~~\text{since}~\lambda[r]\pth{M[r]-m[r]}L\ge 0, \forall\, r
\end{align}
Taking limit sup on both sides of (\ref{alg2 a3}), we get
\begin{align*}
\limsup_{\tau\diverge} Dist\pth{z[\tau+1], Y}&\le =\lim_{\tau\diverge}\lambda[t_{i(\tau)}] L+\lim_{\tau\diverge}\sum_{r=t_{i(\tau)}}^{\infty} \lambda[r]\pth{M[r]-m[r]}L\\
& =0+0=0~~~~\text{by Corollary~\ref{cor1}}
\end{align*}
To apply Corollary \ref{cor1} here we have to have $t_{i(\tau)} \diverge$  when $\tau \diverge$. This is true since there are infinite resilient points.

Using a similar argument used earlier in the proof of Theorem \ref{optimal alg1}, we conclude that $\lim_{t\diverge} Dist\pth{z[t], Y}$ exists and  $\lim_{t\diverge} Dist\pth{z[t], Y}=0$.

\paragraph{{\bf Case (ii):}} There are finitely many points in  $\{x[t]\}_{t=0}^{\infty}$ that are resilient with respect to $\{g[t]\}_{t=0}^{\infty}$.

By the assumption in case (ii) we know that there exists a time index $m_0$ such that for all $t\ge m_0$, each $x[t]$ is not a resilient point with respect to $g[t]$. Thus, for $t\ge m_0$, either (\ref{alg2 distance Y2}) or (\ref{alg2 distance Y4}) holds. Thus, for $t\ge m_0$, we have
\begin{align}
\label{CS t case ii 1}
Dist\pth{z[t+1], Y}\le ~Dist\pth{z[t],Y}+\lambda[t]L\pth{M[t]-m[t]}.
\end{align}
 Define $\{a_r\}_{r=0}^{\infty}$, $\{b_r\}_{r=0}^{\infty}$, and $\{c_r\}_{r=0}^{\infty}$ as follows.
\begin{align*}
&a_r=Dist\pth{z[m_0+r],Y},\\
&b_r=0,\\
&c_r=\lambda[m_0+r]L\pth{M[m_0+r]-m[m_0+r]}.
\end{align*}
By Lemma \ref{finite ub crash} and Lemma \ref{SB stochatic convergence}, we know the limit of $Dist\pth{z[t], Y}$ exists. Let $c\ge 0$ be a nonnegative constant such that
\begin{align}
\label{crash case 2 limit}
\lim_{t\diverge} Dist\pth{z[t], Y}=c.
\end{align}

Repeatedly applying (\ref{CS t case ii 1}), we get
\begin{align}
\label{crash case 2 bounded}
\nonumber
Dist\pth{z[t+1], Y}&\le ~Dist\pth{z[t],Y}+\lambda[t]L\pth{M[t]-m[t]}\\
\nonumber
&\le ~Dist\pth{z[m_0],Y}+\sum_{r=m_0}^{t}\lambda[r]L\pth{M[r]-m[r]}\\
\nonumber
&\le ~Dist\pth{z[m_0],Y}+\sum_{r=m_0}^{\infty}\lambda[r]L\pth{M[r]-m[r]}\\
\nonumber
&\le ~Dist\pth{z[m_0],Y}+\sum_{r=0}^{\infty}\lambda[r]L\pth{M[r]-m[r]}\\
&\le  ~Dist\pth{z[m_0],Y}+C_2~~~\text{by (\ref{alg22})}.
\end{align}
Thus, by (\ref{crash case 2 bounded}), we know that for each $t\ge m_0$
$$Dist\pth{z[t+1], Y}\le Dist\pth{z[m_0], Y}+C_2.$$
Thus, $$\lim_{t\diverge} Dist\pth{z[t], Y}=c<\infty.$$

\paragraph{Case (ii.a):}
Assume that there are infinitely many time indices $t\ge m_0$ such that
$$x[t]-\lambda[t] g[t]~=~x_{j_{t+1}^{\prime}}[t]-\lambda[t]\frac{1}{\left | \calR_{j_{t+1}}[t]\right |}\sum_{k\in \calR_{j_{t+1}}[t]}h_k^{\prime}(x_{j_{t+1}^{\prime}}[t])\in ~Y.$$
Let $\{t_k\}_{k=0}^{\infty}$ be the maximal sequence of such indices. By (\ref{alg2 distance Y2'}), we have
\begin{align}
\label{alg2 a4}
Dist\pth{z[t_k+1], Y}&\le 0+L \lambda[t_k] \pth{M[t_k]-m[t_k]}.
\end{align}
Taking limit on both sides of (\ref{alg2 a4}), we get
\begin{align*}
\lim_{k\diverge}Dist\pth{z[t_k+1], Y}&\le 0+L \lim_{k\diverge}\lambda[t_k] \pth{M[t_k]-m[t_k]}\\
&=0+0=0 ~~~\text{by Corollary \ref{cor1}}
\end{align*}
On the other hand, $\lim_{k\diverge}Dist\pth{z[t_k+1], Y}=c\ge 0$. Thus,
$$c=\lim_{t\diverge} Dist\pth{z[t], Y}=\lim_{k\diverge} Dist\pth{z[t_k+1], Y}=0,$$
proving the theorem.



\paragraph{Case (ii.b):}
Assume that there are only finitely many time indices $t\ge m_0$ such that
$$x[t]-\lambda[t] g[t]~=~x_{j_{t+1}^{\prime}}[t]-\lambda[t]\frac{1}{\left | \calR_{j_{t+1}}[t]\right |}\sum_{k\in \calR_{j_{t+1}}[t]}h_k^{\prime}(x_{j_{t+1}^{\prime}}[t])\in ~Y.$$
Then, there exists $m^{\prime}\ge m_0$ such that for each $t\ge m^{\prime}\ge m_0$, $x[t]$ is not a resilient point with respect to $g[t]$, and
$$x[t]-\lambda[t] g[t]~=~x_{j_{t+1}^{\prime}}[t]-\lambda[t]\frac{1}{\left | \calR_{j_{t+1}}[t]\right |}\sum_{k\in \calR_{j_{t+1}}[t]}h_k^{\prime}(x_{j_{t+1}^{\prime}}[t])\notin ~Y.$$
Thus, for each $t\ge m^{\prime}\ge m_0$, (\ref{alg2 distance Y3}) holds, i.e.,
\begin{align*}
Dist\pth{z[t+1], Y}\le Dist\pth{z[t], Y}-\lambda[t]\left |\frac{1}{\left | \calR_{j_{t+1}}[t]\right |}\sum_{k\in \calR_{j_{t+1}}[t]}h_k^{\prime}(x_{j_{t+1}^{\prime}}[t])\right |+L \lambda[t] \pth{M[t]-m[t]}.
\end{align*}

Recall that $0\le c< \infty$ is a nonnegative constant such that
$\lim_{t\diverge}~Dist\pth{z[t],Y}=c$.
{\bf Next we show that $c=0$.} We prove this by contradiction. Suppose $c>0$.
By Lemma \ref{accu point}, we know that either (A.1) is true or (A.2) is true.\\
(A.1) There exists a subsequence $\{z[t_k]\}_{k=0}^{\infty}$ such that $z[t_k]<\min Y$ for all $k\ge 0$.\\
(A.2) There exists a subsequence $\{z[t_k^{\prime}]\}_{k=0}^{\infty}$ such that $z[t_k^{\prime}]>\max Y$ for all $k\ge 0$.\\
In addition, at least one of $(\min Y-c)$ or $(\max Y+c)$ is an accumulation point of $\{z[t]\}_{t=0}^{\infty}$, and no other accumulation points exist.

Let $a=\min Y$, $b=\max Y$ and $\epsilon=\frac{c}{2}$. 
%
%
It can be seen from the proof of Lemma \ref{accu point} that there exists $m$ such that $z[t]\notin Y$ for each $t\ge m$.
We consider three scenarios: (A.1) is true but (A.2) is not true, (A.2) is true but (A.1) is not true, both (A.1) and (A.2) are true.
\paragraph{{\bf When (A.1) holds but (A.2) does not hold:}} That is, there exists a subsequence $\{z[t_k]\}_{k=0}^{\infty}$ such that $z[t_k]<\min Y$ for all $k\ge 0$; and there does not exist a subsequence $\{z[t_k^{\prime}]\}_{k=0}^{\infty}$ such that $z[t_k^{\prime}]>\max Y$ for all $k\ge 0$. Then there exists $m_1\ge m$ such that $z[t]<\min Y$ for each $t\ge m_1\ge m$. From the proof of Lemma \ref{accu point}, we know
\begin{align*}
\lim_{t\diverge} z[t]~=~\min Y -c~=~a-c.
\end{align*}
Since (\ref{A1-A2}) holds, there exists $m_1^*\ge m_1\ge m$ such that for all $t\ge m_1^*\ge m_1\ge m$, the following holds.
\begin{align}
|z[t]-\pth{a-c}|\le \epsilon=\frac{c}{2}~~~\iff~~~a-\frac{3c}{2}\le z[t]\le a-\frac{c}{2}.
\end{align}
Since $c>0$, we have $a-\frac{c}{2}<a$. Then, for each $p(\cdot)\in \calC$, $p^{\prime}(a-\frac{c}{2})<0$. Then,
$$\rho^*\triangleq \sup_{p(\cdot)\in \calC} p^{\prime}(a-\frac{c}{2})\le 0.$$
Let $K=\sum_{j\in \calF} {\bf 1}\{h_j^{\prime}(a-\frac{c}{2})\ge 0\}$.
Define $q(x)$ as follows,
$$q(x)=\frac{1}{|\calN|+K}\pth{\sum_{j\in \calN} h_j(x)+\sum_{j\in \calF} h_j(x){\bf 1}\{h_j^{\prime}(a-\frac{c}{2})\ge 0\}}.$$
It can be easily seen that $q(\cdot)\in \calC$ is a valid function and $$\rho^*=\sup_{p(\cdot)\in \calC} p^{\prime}(a-\frac{c}{2})=q^{\prime}(a-\frac{c}{2})<0.$$

Note that when $t\ge m_1^*\ge m_1\ge m$, (\ref{alg2 distance Y3}) may not hold, since it is possible that $z[t]-\lambda[t]g[t]\in Y$. Let $\tilde{t}_1 =\max\{m_1^*, m^{\prime}\}$. For each $t\ge \tilde{t}_1 =\max\{m_1^*, m^{\prime}\}$, (\ref{alg2 distance Y3}), (\ref{A1-A2}) and (\ref{limit trapped A1-A2}) hold. We have
\begin{align}
\label{alg2 case 2 a1}
\nonumber
Dist\pth{z[t+1], Y}&\le Dist\pth{z[t], Y}-\lambda[t]\left|\frac{1}{\left |  \calR_{j_{t+1}}[t] \right |}\pth{\sum_{i\in \calR_{j_{t+1}}[t] } h_i^{\prime}(x_{j_{t+1}^{\prime}}[t])}\right|+\lambda[t]L\pth{M[t]-m[t]}~~~\text{by}~(\ref{alg2 distance Y3})\\
\nonumber
&\le Dist\pth{z[t], Y}-\lambda[t]\left|\frac{1}{\left |  \calR_{j_{t+1}}[t] \right |}\pth{\sum_{i\in \calR_{j_{t+1}}[t] } h_i^{\prime}(z[t])}\right|+2\lambda[t]L\pth{M[t]-m[t]}\\
&\le Dist\pth{z[t], Y}-\lambda[t]|\rho^*|+2\lambda[t]L\pth{M[t]-m[t]}.
\end{align}
Repeatedly applying (\ref{alg2 case 2 a1}) for $t\ge \tilde{t}_1 =\max\{m_1^*, m^{\prime}\}$, we get

\begin{align}
\label{alg2 case 2 a11}
Dist\pth{z[t+1], Y}\le Dist\pth{z[\tilde{t}_1], ~Y}-\pth{\sum_{r=\tilde{t}_1}^{t}\lambda[r]}|\rho^*|+2\sum_{r=\tilde{t}_1}^{t}\lambda[r]L\pth{M[r]-m[r]}.
\end{align}
%
%
Taking limit on both sides of (\ref{alg2 case 2 a11}), we obtain
\begin{align}
\label{alg2 f1}
\nonumber
\lim_{t\diverge}Dist\pth{z[t+1], Y}&\le Dist\pth{z[\tilde{t}_1], ~Y}-\pth{\sum_{r=\tilde{t}_1}^{\infty}\lambda[r]}|\rho^*|+2\sum_{r=\tilde{t}_1}^{\infty}\lambda[r]L\pth{M[r]-m[r]}\\
\nonumber
&\le Dist\pth{z[\tilde{t}_1], ~Y}-\pth{\sum_{r=\tilde{t}_1}^{\infty}\lambda[r]}|\rho^*|+2C_2~~~\text{by (\ref{alg22})}\\
\nonumber
&= Dist\pth{z[\tilde{t}_1], ~Y}-\infty+2C_2\\
&=-\infty.
\end{align}
%
On the other hand, we know $ \lim_{t\diverge}Dist\pth{z[t], Y}= c>0$. This is a contradiction.
Thus,
$$ \lim_{t\diverge} Dist\pth{z[t], Y}=c=0.$$

\vskip 2\baselineskip

Similarly, we can show the case when (A.2) holds but (A.1) does not hold, and the case when both (A.1) and (A.2) hold. \\

The proof of the theorem is complete.

\eproof
\end{proof}

\section{Synchronous Byzantine Iterative Algorithm}
\label{sec: algorithm byzantine}
In this section, we present an iterative algorithm, in which each non-faulty agent sends only one message per iteration, and keeps minimal memory across iterations. We assume each local cost function $h_j(\cdot)$ has $L$--Lipschitz continuous derivative. 

\paragraph{}
\vspace*{8pt}\hrule
~

{\bf Algorithm 3} for agent $j$ for iteration $t\ge 1$:
~
\vspace*{4pt}\hrule

\begin{list}{}{}
\item[{\bf Step 1:}]
Compute $h_j^{\prime}\pth{x_j[t-1]}$ -- the gradient of the local cost function $h_j(\cdot)$ at point $x_j[t-1]$, and send the estimate and gradient pair $(x_j[t-1], h_j^{\prime}\pth{x_j[t-1]})$ to all the agents (including itself). \\

~
\item[{\bf Step 2:}]
Let $\calR_j[t-1]$ denote the set of tuples of the form $\pth{x_i[t-1], \, h_i^{\prime}(x_i[t-1])}$ received as a result of step 1.\\

In step 2, agent $j$ should be able to receive a tuple $(w_i[t-1], g_i[t-1])$ from each agent $i\in \calV$. For non-faulty agent $i\in \calN$, $w_i[t-1]=x_i[t-1]$ and $g_i[t-1]=h_i^{\prime}\pth{x_i[t-1]}$. If a faulty agent $k\in \calF$ does not send a tuple to agent $j$, then agent $j$ assumes $(w_k[t-1], g_k[t-1])$ to be some default tuple. \footnote{In contrast to Algorithms 1, 2 and 3 in \cite{su2015byzantine}, the adopted default tuple in Algorithm 3 here is not necessarily known to all agents. In addition, the default tuple may vary across iterations. }\\

~

\item[{\bf Step 3:}] Sort the first entries of the received tuples in $\calR_j[t-1]$ in a non-increasing order (breaking ties arbitrarily), and erase the smallest $f$ values and the largest $f$ values. Let $\calR_j^{1}[t-1]$ be the identifiers of the $n-2f$ agents from whom the remaining first entries were received. Similarly, sort the second entries of the received tuples in $\calR_j[t-1]$ in a non-increasing order (breaking ties arbitrarily), and erase the smallest $f$ values and the largest $f$ values. Let $\calR_j^{2}[t-1]$ be the identifiers of the $n-2f$ agents from whom the remaining second entries were received.
Denote the largest and smallest gradients among the remaining values by $\hat{g}_j[t-1]$ and $\check{g}_j[t-1]$, respectively. Set $\widetilde{g}_j[t-1]=\frac{1}{2}\pth{\hat{g}_j[t-1]+\check{g}_j[t-1]}$.

Update its state as follows.
\begin{align}
\label{Byzantine Iterative}
x_j[t]=\frac{1}{n-2f} \pth{\sum_{i\in \calR_j^1[t-1]}w_i[t-1]}-\lambda[t-1]\widetilde{g}_j[t-1].
\end{align}

\end{list}

\hrule

~

Let $\widetilde{\calC}$ be the collection of functions defined as follows:
\begin{align}
\nonumber
\widetilde{\calC}\triangleq \{~~~p(x): p(x)&=\sum_{i\in \calN} \alpha_i h_i(x), ~~\forall i\in\calN, ~ \alpha_i\geq 0,\\
\nonumber
&\sum_{i\in \calN}\alpha_i=1,\text{~~and~~}\\
&\sum_{i\in\calN} {\bf 1}\left(\alpha_i\ge \frac{1}{2(|\calN|-f)}\right) ~ \geq ~|\calN|-f ~~~ \}
\label{BS valid collection}
\end{align}
Each $p(x)\in \widetilde{\calC}$ is called a valid function. Note that the function $\frac{1}{|\calN|}\sum_{i\in \calN}h_i(x)\in \widetilde{\calC}$ since $n\ge 3f+1$ and $|\calN|\ge 2f+1$. For ease of future reference, we let $\widetilde{p}(x)=\frac{1}{|\calN|}\sum_{i\in \calN}h_i(x)$. Define $\widetilde{Y}\triangleq \cup_{p(x)\in \widetilde{\calC}} \argmin~ p(x)$.

\begin{lemma}\cite{su2015byzantine}
\label{BS valid convex}
$\widetilde{Y}$ is a convex set.
\end{lemma}

\begin{lemma}
\label{BS valid closed}
$\widetilde{Y}$ is a closed set.
\end{lemma}
Lemma \ref{BS valid closed} is proved in Appendix \ref{app: BS closed}. By Lemma \ref{BS valid closed}, Definition \ref{BS def resilient} is well-defined over $\tY$.

\subsection{Update Dynamic -- Matrix Representation}

\begin{definition}\cite{vaidya2012iterative}
For a given graph $G(\calV, \calE)$, a reduced graph $\calH$ is a subgraph of $G(\calV, \calE)$ obtained by (i) removing all the faulty agents from $\calV$ along with their edges; (ii) removing any additional up to $f$ incoming edges at each non-faulty agent.
\end{definition}
Let us denote the collection of all the reduced graphs for a given $G(\calV, \calE)$ by $R_\calF$. Thus, $\calV-\calF$ is the set of agents in each element in $R_\calF$. Let $\tau=|R_\calF|$. It is easy to see that
$\tau$ depends on $\calF$, and it is finite.\\

Without loss of generality, assume agents indexed from 1 through $n-\phi$ are non-faulty, and agents indexed from $n-\phi+1$ to $n$ are faulty.
Let ${\bf x}[t-1]\in \reals^{n-\phi}$ be a real vector of the local estimates at the beginning of iteration $t$ with ${\bf x}_j[t-1]=x_j[t-1]$ being the local estimate of agent $j\in \calN$, and let $\widetilde{\bf g}[t-1]\in  \reals^{n-\phi}$ be a vector of the local gradients at iteration $t$ with $\widetilde{\bf g}_j[t-1]=\widetilde{ g}_j[t-1], j\in \calN$.
Since the underlying communication network is a complete  graph with $n\ge 3f+1$, as shown in \cite{Vaidya2012MatrixConsensus}, the update of ${\bf x}\in \reals^{n-\phi}$ in each iteration can be written compactly in a matrix form.
\begin{align}
\label{matrix representation}
{\bf x}[t]={\bf M}[t-1]{\bf x}[t-1]-\lambda[t-1]\widetilde{\bf g}[t-1].
\end{align}

The construction of ${\bf M}[t]$ and relevant properties are given in \cite{Vaidya2012MatrixConsensus}. Let $\calH\in \calR_{\calF}$ be a reduced graph of the given communication graph, with ${\bf H}$ as the adjacency matrix. It is shown in \cite{Vaidya2012MatrixConsensus} that in every iteration $t$, and for every ${\bf M}[t]$, there exists a reduced graph $\calH[t]\in \calR_{\calF}$ with adjacency matrix ${\bf H}[t]$ such that
\begin{align}
\label{iterative BS matrix lb}
{\bf M}[t]\ge \beta {\bf H}[t],
\end{align}
where $0<\beta<1$ is a constant. The definition of $\beta$ can be found in \cite{Vaidya2012MatrixConsensus}.

Equation (\ref{matrix representation}) can be further expanded out as

\begin{align}
\label{MR evo BS}
\nonumber
{\bf x}[t]~&=~{\bf M}[t-1]{\bf x}[t-1]-\lambda[t-1] \widetilde{\bf g}[t-1]\\
\nonumber
&=~{\bf M}[t-1]\pth{ {\bf M}[t-2]{\bf x}[t-2]-\lambda[t-2] \widetilde{\bf g}[t-2]}-\lambda[t-1] \widetilde{\bf g}[t-1]\\
\nonumber
&=~{\bf M}[t-1]{\bf M}[t-2]{\bf x}[t-2]-\lambda[t-2] {\bf M}[t-1]\widetilde{\bf g}[t-2]-\lambda[t-1] \widetilde{\bf g}[t-1]\\
\nonumber
&=~\cdots\\
\nonumber
&=~\pth{{\bf M}[t-1]{\bf M}[t-2]\cdots {\bf M}[0]{\bf x}[0]}-\lambda[0] \pth{{\bf M}[t-1]{\bf M}[t-2]\cdots {\bf M}[1] \widetilde{\bf g}[0]}-\cdots-\\
\nonumber
&\quad ~-\lambda[t-1]\widetilde{\bf g}[t-1]\\
&={\bf  \Phi}(t-1, 0){\bf x}[0]-\sum_{r=0}^{t-1} \lambda[r]{\bf \Phi}(t-1, r+1)\widetilde{\bf g}[r],
\end{align}
where ${\bf \Phi}(t-1, r)={\bf M}[t-1]{\bf M}[t-2]\cdots {\bf M}[r]$ is a backward product, and by convention, ${\bf \Phi}(t-1, t-1)={\bf M}[t-1]$ and ${\bf \Phi}(t-1, t)={\bf I}$.

\subsection{Correctness of Algorithm 3}
Using coefficients of ergodicity theorem, it is showed in \cite{Vaidya2012MatrixConsensus} that ${\bf \Phi}(t,r)$ is weak-ergodic \cite{Vaidya2012MatrixConsensus}, and that the rate of the convergence  is exponential \cite{Anthonisse1977360}, as formally stated in Theorem \ref{convergencerate}. Recall that $\tau=|R_{\calF}|$, $n-\phi$ is the total number of non-faulty agents, and $0<\beta<1$ is a constant for which (\ref{iterative BS matrix lb}) holds.

\begin{theorem}\cite{Anthonisse1977360}
\label{convergencerate}
Let $\nu=\tau(n-\phi)$ and $\gamma=1-\beta^{\nu}$. For any sequence ${\bf \Phi}(t, r)$,
\begin{align}
\left | {\bf \Phi}_{ik}(t, r)-{\bf \Phi}_{jk}(t, r)\right |\le \gamma^{\lceil\frac{t-r+1}{\nu}\rceil},
\end{align}
for all $t\ge r$.
\end{theorem}

\begin{lemma}
\label{SB asymptotic consensus fb}
For all $i, j\in \calN$ and for each $t\ge 1$,
\begin{align*}
|x_i[t]-x_j[t]|\le (n-\phi)  \max\{|u|, |U|\} \gamma^{\lceil \frac{t}{\nu}\rceil}+L\sum_{r=0}^{t-1}\lambda[r](n-\phi) \gamma^{\lceil \frac{t-1-r}{\nu}\rceil},
\end{align*}
and for all $i, j\in \calN$ and for $t=0$,
\begin{align*}
|x_i[0]-x_j[0]|\le U-u.
\end{align*}
\end{lemma}

The proof of Lemma \ref{SB asymptotic consensus fb} can be found in Appendix \ref{app: SB asymptotic consensus fb}.

\begin{corollary}
\label{BS consensus}
For $i,j\in \calN$,
$$\lim_{t\diverge}|x_i[t]-x_j[t]|=0.$$
\end{corollary}
We present the proof of Corollary \ref{BS consensus} in Appendix \ref{app: BS consensus}.

%
%
%
%
%

  Let $M[t]=\max_{i\in \calN}x_i[t]$ and $m[t]=\min_{i\in \calN}x_i[t]$. The following lemma holds.
\begin{lemma}
\label{finiteness of series SB}
Under Algorithm 3, the following holds.
$$\sum_{t=0}^{\infty} \lambda[t]\pth{M[t]-m[t]}<\infty.$$
\end{lemma}
The proof of Lemma \ref{finiteness of series SB} is similar to the proof of Lemma \ref{finite ub crash}. For completeness, we present the proof in Appendix \ref{app: finiteness of series SB}.

\begin{proposition}
\label{p1}
Let $a, b, c, d\in \reals$ such that
$b<a, b\le c\le \frac{1}{2}\pth{a+b}, \frac{1}{2}\pth{a+b}<a \le d,$ and there exists $0\le \xi\le 1$, for which $\frac{1}{2}\pth{a+b}=\xi d+(1-\xi)c$ holds. Then
$$\frac{1}{2}\le \xi \le 1.$$
\end{proposition}
\begin{proof}
Suppose, on the contrary, that $0\le \xi<\frac{1}{2}$. Since $c\ge b$ and $d\ge a$, we have
\begin{align}
\label{cc1}
\frac{1}{2}\pth{c+d}\ge \frac{1}{2}\pth{a+b}
\end{align}
On the other hand, by the assumptions that $a>b$, and that $\frac{1}{2}\pth{a+b}\ge c$,  it holds that
\begin{align}
\label{c2}
d\ge a >\frac{1}{2}\pth{a+b}\ge c,
\end{align}
i.e., $d> c$. Then
\begin{align}
\label{c3}
\nonumber
\frac{1}{2}\pth{c+d}
&=\frac{1}{2}c+\frac{1}{2}d\\
\nonumber
&=\xi c+\pth{\frac{1}{2}-\xi}c+\frac{1}{2}d\\
\nonumber
&< \xi c+\pth{\frac{1}{2}-\xi}d+\frac{1}{2}d~~~\text{by (\ref{c2})}\\
\nonumber
&=\xi c+\pth{\frac{1}{2}-\xi+\frac{1}{2}}d\\
\nonumber
&=\xi c+\pth{1-\xi}d\\
&=\frac{1}{2}\pth{a+b},
\end{align}
i.e., $\frac{1}{2}\pth{c+d}<\frac{1}{2}\pth{a+b}$.
The relations in (\ref{cc1}) and (\ref{c3}) contradict each other. Thus, the assumption that $0\le \xi<\frac{1}{2}$ does not hold, i.e., $\frac{1}{2}\le \xi\le 1$, proving the proposition.

\eproof
\end{proof}

\begin{lemma}
\label{BS valid gradient}
For each non-faulty agent $j\in \calN$ and each iteration $t\ge 1$, there exists a valid function $p(x)=\sum_{i\in \calN} \alpha_i\, h_i(x)\in \calC$ such that
$$\widetilde{g}_j[t-1]=\sum_{i\in \calN} \alpha_i\, h_i^{\prime}(x_i[t-1]).$$
\end{lemma}
\begin{proof}
Recall that $\calR_j^2[t-1]$ denotes the set of agents from whom the remaining $n-2f$ gradient values (second entries of the tuples) were received in iteration $t$, and let us denote by $\calL_j[t-1]$ and $\calS_j[t-1]$ the set of agents from whom the largest $f$ gradient values and the smallest $f$ gradient values were received in iteration $t$. 

Let $i^*, j^*\in \calR_j^2[t-1]$ such that $g_{i^*}[t-1]=\check{g}_j[t-1]$ and $g_{j^*}[t-1]=\hat{g}_j[t-1]$.
Recall that $|\calF|=\phi$. Let $\calL_j^*[t-1]\subseteq \calL_j[t-1]-\calF$ and $\calS_j^*[t-1]\subseteq \calS_j[t-1]-\calF$ such that
$$|\calL_j^*[t-1]|=f-\phi+|\calR_j^2[t-1]\cap \calF|,$$
and
$$|\calS_j^*[t-1]|=f-\phi+|\calR_j^2[t-1]\cap \calF|.$$

We consider two cases: (i) $\hat{g}_j[t-1]>\check{g}_j[t-1]$ and (ii) $\hat{g}_j[t-1]=\check{g}_j[t-1]$, separately.

\paragraph{{\bf Case (i)}: $\hat{g}_j[t-1]>\check{g}_j[t-1]$.}

By definition of $\calL_j^*[t-1]$ and $\calS_j^*[t-1]$, we have
\begin{align}
\label{case 1 valid gradient}
\frac{1}{f-\phi+|\calR_j^2[t-1]\cap \calF|}\sum_{i\in \calS_j^*[t-1]}g_i[t-1]\le \widetilde{g}_j[t-1]\le \frac{1}{f-\phi+|\calR_j^2[t-1]\cap \calF|}\sum_{i\in \calL_j^*[t-1]}g_i[t-1].
\end{align}
Thus, there exists $0\le \xi\le 1$ such that
\begin{align}
\label{extrem nonfaulty}
\nonumber
\widetilde{g}_j[t-1]&=\xi\pth{\frac{1}{f-\phi+|\calR_j^2[t-1]\cap \calF|}\sum_{i\in \calS_j^*[t-1]}g_i[t-1]}+(1-\xi)\pth{\frac{1}{f-\phi+|\calR_j^2[t-1]\cap \calF|}\sum_{i\in \calL_j^*[t-1]}g_i[t-1]}\\
&=\frac{\xi}{f-\phi+|\calR_j^2[t-1]\cap \calF|}\sum_{i\in \calS_j^*[t-1]}g_i[t-1]+\frac{1-\xi}{f-\phi+|\calR_j^2[t-1]\cap \calF|}\sum_{i\in \calL_j^*[t-1]}g_i[t-1].
\end{align}
By symmetry, WLOG, assume $\xi\ge \frac{1}{2}$. \\

Let $k\in \calR_j^2[t-1]-\calF$. By symmetry, WLOG, assume $g_k[t-1]\le \widetilde{g}_j[t-1]$. Since $|\calL_j[t-1]\cup \{j^*\}|=f+1$, there exists a non-faulty agent $j^{\prime}_k\in \calL_j[t-1]\cup \{j^*\}$. Thus, $g_{j^{\prime}_k}[t-1]\ge \hat{g}_j[t-1]>\widetilde{g}_j[t-1]$, and there exists $0\le \xi_k\le 1$ such that
\begin{align}
\label{BS middle nonfaulty}
\frac{1}{2}\pth{\hat{g}_j[t-1]+\check{g}_j[t-1]}=\widetilde{g}_j[t-1]=\xi_k g_k[t-1]+(1-\xi_k) g_{j^{\prime}_k}[t-1].
\end{align}
Let $a=\hat{g}_j[t-1], b=\check{g}_j[t-1], c=g_k[t-1],$ and $d=g_{j^{\prime}_k}[t-1]$. By Proposition \ref{p1}, we know that
$\frac{1}{2}\le \xi_k\le 1$.


Since $|\calN|-f=n-\phi-f=n-2f+f-\phi=\left|\calR_j^2[t-1]\right|+f-\phi
=\left|\calR_j^2[t-1]-\calF\right|+\left|\calR_j^2[t-1]\cap\calF\right|+f-\phi$, we get
\begin{align*}
\nonumber
\widetilde{g}_j[t-1]&=\frac{|\calN|-f}{|\calN|-f}\widetilde{g}_j[t-1]\\
&=\frac{|\calR_j^2[t-1]-\calF|}{|\calN|-f}\widetilde{g}_j[t-1]+\frac{f-\phi+|\calR_j^2[t-1]\cap \calF|}{|\calN|-f}\widetilde{g}_j[t-1]\\
\nonumber
&=\frac{1}{|\calN|-f}\pth{\sum_{k\in \calR_j^2[t-1]-\calF}\widetilde{g}_j[t-1]}+\frac{f-\phi+|\calR_j^2[t-1]\cap \calF|}{|\calN|-f}\widetilde{g}_j[t-1]\\
\nonumber
&=\frac{1}{|\calN|-f}\sum_{k\in \calR_j^2[t-1]-\calF}\pth{\xi_k g_k[t-1]+(1-\xi_k) g_{j^{\prime}_k}[t-1]}\\
\nonumber
&\quad +\frac{\xi}{|\calN|-f}\sum_{i\in \calS_j^*[t-1]}g_i[t-1]+\frac{1-\xi}{|\calN|-f}\sum_{i\in \calL_j^*[t-1]}g_i[t-1]~~~\text{by}~(\ref{extrem nonfaulty})~\text{and}~(\ref{BS middle nonfaulty})\\
\nonumber
&=\frac{1}{|\calN|-f}\sum_{k\in \calR_j^2[t-1]-\calF}\pth{\xi_k\, h_k^{\prime}(x_k[t-1])+(1-\xi_k)\, h_{j^{\prime}_k}^{\prime}(x_{j^{\prime}_k}[t-1])}\\
&\quad +\frac{\xi}{|\calN|-f}\sum_{i\in \calS_j^*[t-1]}h_i^{\prime}(x_i[t-1])+\frac{1-\xi}{|\calN|-f}\sum_{i\in \calL_j^*[t-1]}h_i^{\prime}(x_i[t-1]).
\end{align*}

Define $q(x)$ as follows.
\begin{align}
\label{BS rewritten 1}
\nonumber
q(x)&=\frac{1}{|\calN|-f}\sum_{k\in \calR_j^2[t-1]-\calF}\pth{\xi_k\, h_k(x)+(1-\xi_k)\, h_{j^{\prime}_k}(x)}\\
&\quad +\frac{\xi}{|\calN|-f}\sum_{i\in \calS_j^*[t-1]}h_i(x)+\frac{1-\xi}{|\calN|-f}\sum_{i\in \calL_j^*[t-1]}h_i(x).
\end{align}
In (\ref{BS rewritten 1}), for each $k\in \calR_j^2[t-1]-\calF$, it holds that $\frac{\xi_k}{|\calN|-f}\ge \frac{1}{2\pth{|\calN|-f}}$. For each $i\in \calS_j^*[t-1]$, it holds that $\frac{\xi}{|\calN|-f}\ge \frac{1}{2\pth{|\calN|-f}}$. In addition, we have
\begin{align*}
|\pth{\calR_j^2[t-1]-\calF}\cup \calS_j^*[t-1]|&=|\calR_j^2[t-1]-\calF|+|\calS_j^*[t-1]|\\
&=|\calR_j^2[t-1]|-|\calR_j^2[t-1]\cap \calF|+|\calS_j^*[t-1]|\\
&=n-2f-|\calR_j^2[t-1]\cap \calF|+f-\phi+|\calR_j^2[t-1]\cap \calF|\\
&=n-\phi-f=|\calN|-f.
\end{align*}

Thus, in (\ref{BS rewritten 1}), at least $|\calN|-f$ non-faulty agents corresponding to agents $k\in \pth{\calR_j^2[t-1]-\calF}\cup \calS_j^*[t-1]$ are assigned with weights lower bounded by $\frac{1}{2(|\calN|-f)}$. 

\paragraph{{\bf Case (ii)}: $\hat{g}_j[t-1]=\check{g}_j[t-1]$.}
Let $k\in \calR_j^2[t-1]-\calF$. Since $\hat{g}_j[t-1]\ge g_k[t-1]\ge \check{g}_j[t-1]$ and $\hat{g}_j[t-1]=\check{g}_j[t-1]$, it holds that $\hat{g}_j[t-1]=g_k[t-1]=\check{g}_j[t-1]$.
Consequently, we have
$$\widetilde{g}_j[t-1]=\frac{1}{2}\pth{\hat{g}_j[t-1]+\check{g}_j[t-1]}=g_k[t-1].$$
So we can rewrite $\widetilde{g}_j[t-1]$ as follows.
\begin{align*}
\nonumber
\widetilde{g}_j[t-1]&=\frac{|\calN|-f}{|\calN|-f}\,\widetilde{g}_j[t-1]\\
\nonumber
&=\frac{1}{|\calN|-f}\pth{\sum_{k\in \calR_j^2[t-1]-\calF}\widetilde{g}_j[t-1]}+\frac{f-\phi+|\calR_j^2[t-1]\cap \calF|}{|\calN|-f}\widetilde{g}_j[t-1]\\
\nonumber
&=\frac{1}{|\calN|-f}\sum_{k\in \calR_j^2[t-1]-\calF}g_k[t-1]+\frac{\xi}{|\calN|-f}\sum_{i\in \calS_j^*[t-1]}g_i[t-1]+\frac{1-\xi}{|\calN|-f}\sum_{i\in \calL_j^*[t-1]}g_i[t-1]\\
\nonumber
&=\frac{1}{|\calN|-f}\sum_{k\in \calR_j^2[t-1]-\calF}h_k^{\prime}(x_k[t-1])\\
&\quad+\frac{\xi}{|\calN|-f}\sum_{i\in \calS_j^*[t-1]}h_i^{\prime}(x_i[t-1])+\frac{1-\xi}{|\calN|-f}\sum_{i\in \calL_j^*[t-1]}h_i^{\prime}(x_i[t-1]).
\end{align*}

Define $q(x)$ as follows.
\begin{align}
\label{BS rewritten 2}
q(x)=\frac{1}{|\calN|-f}\sum_{k\in \calR_j^2[t-1]-\calF}h_k(x)+\frac{\xi}{|\calN|-f}\sum_{i\in \calS_j^*[t-1]}h_i(x)+\frac{1-\xi}{|\calN|-f}\sum_{i\in \calL_j^*[t-1]}h_i(x).
\end{align}
In (\ref{BS rewritten 2}), for each $k\in \calR_j^2[t-1]-\calF$, it holds that $\frac{1}{|\calN|-f}\ge \frac{1}{2\pth{|\calN|-f}}$. For each $i\in \calS_j^*[t-1]$, it holds that $\frac{\xi}{|\calN|-f}\ge \frac{1}{2\pth{|\calN|-f}}$. In addition, we have
\begin{align*}
|\pth{\calR_j^2[t-1]-\calF}\cup \calS_j^*[t-1]|=|\calN|-f.
\end{align*}
Thus, in (\ref{BS rewritten 2}), at least $|\calN|-f$ non-faulty agents corresponding to $\pth{\calR_j^2[t-1]-\calF}\cup \calS_j^*[t-1]$ are assigned with weights lower bounded by $\frac{1}{2(|\calN|-f)}$.\\ 

Case (i) and Case (ii) together prove the lemma.

\eproof
\end{proof}

%

\begin{proposition}
\label{BS p1}
For each non-faulty agent $j\in \calN$ and each $t\ge 1$, there exists a set of convex coefficients $\beta_i$'s over non-faulty agents, i.e., $\beta_i\ge 0$ for each $i\in \calN$ and $\sum_{i\in \calN} \beta_i=1$, such that the following holds
$$\frac{1}{n-2f}\sum_{i\in \calR_j^1[t-1]} w_i[t-1]=\sum_{i\in \calN}\beta_i x_i[t-1].$$
\end{proposition}
\begin{proof}
Note that $\calR_j^1[t-1]=\pth{\calR_j^1[t-1]-\calF}\cup \pth{\calR_j^1[t-1]\cap\calF}$. We consider two cases: (i) $\calR_j^1[t-1]\cap\calF=\O$ and (ii) $\calR_j^1[t-1]\cap\calF\not=\O$, separately.

\paragraph{{\bf Case (i):} $\calR_j^1[t-1]\cap\calF=\O$.}
When $\calR_j^1[t-1]\cap\calF=\O$, every agent in $\calR_j^1[t-1]$ is non-faulty, i.e., $\calR_j^1[t-1]\subseteq \calN$. Then we get
\begin{align}
\label{BS p1 c1}
\frac{1}{n-2f}\sum_{i\in \calR_j^1[t-1]} w_i[t-1]=\frac{1}{n-2f}\sum_{i\in \calR_j^1[t-1]} x_i[t-1] ~~~\text{since}~ w_i[t-1]=x_i[t-1]~\text{for each}~ i\in \calN
\end{align}
Let $\beta_i=\frac{1}{n-2f}$ for each $i\in \calR_j^1[t-1]\subseteq \calN$, and $\beta_i=0$ for each $i\in \calN-\calR_j^1[t-1]$. The obtained $\beta_i$'s is a valid collection of convex coefficients, since $\beta_i=\frac{1}{n-2f}\ge 0$ for each $i\calN$, and
$$\sum_{i\in \calN} \beta_i=\sum_{i\in \calR_j^1[t-1]} \beta_i=\sum_{i\in \calR_j^1[t-1]} \frac{1}{n-2f}=\frac{1}{n-2f}|\calR_j^1[t-1]|=\frac{1}{n-2f}(n-2f)=1.$$

\paragraph{{\bf Case (ii):} $\calR_j^1[t-1]\cap\calF\not=\O$.}
Let $\calL_j[t-1]$ be the set of the identifiers of the $f$ agents from whom the $f$ largest first entries ($w_i[t-1]$'s) are received, and let $\calS_j[t-1]$ be the set of the identifiers of the $f$ agents from whom the $f$ smallest first entries ($w_i[t-1]$'s) are received. Since $\calR_j^1[t-1]\cap\calF\not=\O$,  it holds that $\calL_j[t-1]\cap\calN \not=\O$ and $\calS_j[t-1]\cap \calN\not=\O$. Let $l$ and $s$ be two non-faulty agents such that $l\in \calL_j[t-1]\cap\calN$ and $s\in \calS_j[t-1]\cap\calN$. By definition of $\calR_j^1[t-1]$, for each $k\in \calR_j^1[t-1]\cap\calF$, we have
\begin{align}
x_s[t-1]=~w_s[t-1]~\le w_k[t-1]~\le w_l[t-1]=x_l[t-1].
\end{align}
Then,
$$ |\calR_j^1[t-1]\cap\calF|\, x_s[t-1] \le \sum_{k\in \calR_j^1[t-1]\cap\calF}w_k[t-1]~\le |\calR_j^1[t-1]\cap\calF|\,x_l[t-1].$$
Thus, there exists $0\le \zeta\le 1$ such that
\begin{align}
\label{bbb}
\sum_{k\in \calR_j^1[t-1]\cap\calF}w_k[t-1]=\zeta\pth{|\calR_j^1[t-1]\cap\calF|\, x_s[t-1]}+(1-\zeta)\pth{|\calR_j^1[t-1]\cap\calF|\, x_l[t-1]}.
\end{align}
Thus,
\begin{align}
\label{BS p1 c2}
\nonumber
&\frac{1}{n-2f}\sum_{i\in \calR_j^1[t-1]} w_i[t-1]=\frac{1}{n-2f}\pth{\sum_{i\in \calR_j^1[t-1]-\calF} w_i[t-1]+\sum_{i\in \calR_j^1[t-1]\cap \calF} w_i[t-1]} \\
\nonumber
&=\frac{1}{n-2f}\pth{\sum_{i\in \calR_j^1[t-1]-\calF} x_i[t-1]+\sum_{i\in \calR_j^1[t-1]\cap \calF} w_i[t-1]}~~~\text{since}~ w_i[t-1]=x_i[t-1]~\text{for each}~ i\in \calN\\
\nonumber
&=\frac{1}{n-2f}\pth{\sum_{i\in \calR_j^1[t-1]-\calF} x_i[t-1]+    \zeta\pth{|\calR_j^1[t-1]\cap\calF|\, x_s[t-1]}+(1-\zeta)\pth{|\calR_j^1[t-1]\cap\calF|\, x_l[t-1]} }~~\text{by}~(\ref{bbb})\\
&=\frac{1}{n-2f}\sum_{i\in \calR_j^1[t-1]-\calF} x_i[t-1]+\frac{\zeta|\calR_j^1[t-1]\cap\calF|}{n-2f} x_s[t-1]+\frac{(1-\zeta)|\calR_j^1[t-1]\cap\calF|}{n-2f} x_l[t-1].
\end{align}

Let $\beta_s=\frac{\zeta|\calR_j^1[t-1]\cap\calF|}{n-2f}$, $\beta_l=\frac{(1-\zeta)|\calR_j^1[t-1]\cap\calF|}{n-2f}$, let $\beta_i=\frac{1}{n-2f}$ for each $i\in \calR_j^1[t-1]-\calF$, and let $\beta_i=0$ for all other non-faulty agents. The obtained $\beta_i$'s is a valid collection of convex coefficients since
\begin{align*}
\beta_s+\beta_l +\sum_{i\in \calR_j^1[t-1]-\calF} \beta_i&=\frac{\zeta|\calR_j^1[t-1]\cap\calF|}{n-2f}+\frac{(1-\zeta)|\calR_j^1[t-1]\cap\calF|}{n-2f}+\sum_{i\in \calR_j^1[t-1]-\calF} \frac{1}{n-2f}\\
&=\frac{|\calR_j^1[t-1]\cap\calF|}{n-2f}+\frac{|\calR_j^1[t-1]-\calF|}{n-2f}\\
&=\frac{|\calR_j^1[t-1]|}{n-2f}=\frac{n-2f}{n-2f}=1.
\end{align*}

Case (i) and case (ii) together prove the proposition.

\eproof
\end{proof}

We define $z[t]$ and $x_{j_t}$ similar to that for Algorithm 1. In particular, let $\{z[t]\}_{t=0}^{\infty}$ be a sequence of estimates such that
\begin{align}
\label{alg3 crash sequence z}
z[t]=x_{j_{t}}[t], ~~\text{where}~j_{t}\in \argmax_{j\in \calN[t]} Dist\pth{x_j[t], Y}.
\end{align}
From the definition, there is a sequence of agents $\{j_t\}_{t=0}^{\infty}$ associated with the sequence $\{z[t]\}_{t=0}^{\infty}$.

\begin{theorem}
\label{talgo BS}
The sequence $\{Dist\pth{z[t], Y}\}_{t=0}^{\infty}$ converges and $$\lim_{t\diverge} Dist\pth{z[t], Y}=0.$$
\end{theorem}
\begin{proof}

\begin{align}
\label{BS a1}
\nonumber
Dist\pth{z[t+1], Y}&=Dist\pth{x_{j_{t+1}}[t], Y}~~~\text{by (\ref{alg3 crash sequence z})}\\
\nonumber
&=Dist\pth{\frac{1}{n-2f}\sum_{i\in \calR^{1}_{j_{t+1}}[t]} w_i[t]-\lambda[t]\widetilde{g}_{j_{t+1}}[t],~ Y}~~~\text{by}~(\ref{Byzantine Iterative})\\
\nonumber
&=Dist\pth{\sum_{i\in \calN} \beta_i x_i[t]-\lambda[t]\widetilde{g}_{j_{t+1}}[t],~ Y}~~~\text{by Proposition}~ \ref{BS p1}\\
\nonumber
&=Dist\pth{\sum_{i\in \calN} \beta_i \pth{x_i[t]-\lambda[t]\widetilde{g}_{j_{t+1}}[t]},~ Y}~~~\text{since}~ \sum_{i\in \calN} \beta_i=1\\
\nonumber
&\le \sum_{i\in \calN} \beta_i \, Dist\pth{x_i[t]-\lambda[t]\widetilde{g}_{j_{t+1}}[t], ~Y}~~~\text{by convexity of}~Dist\pth{\cdot, Y}\\
&\le \max_{i\in \calN} Dist\pth{x_i[t]-\lambda[t]\widetilde{g}_{j_{t+1}}[t], ~Y}.
\end{align}
By Lemma \ref{BS valid gradient}, there exists a valid function $p_{t}(\cdot)=\sum_{q\in \calN} \alpha_q h_q(\cdot)\in \calC$ such that
\begin{align}
\label{valid gradient inter}
\widetilde{g}_{j_{t+1}}[t]=\sum_{q\in \calN} \alpha_q h_q^{\prime}(x_q[t]).
\end{align}
In addition, let $$j_{t+1}^{\prime}\in \argmax_{i\in \calN} Dist\pth{x_i[t]-\lambda[t]\widetilde{g}_{j_{t+1}}[t], ~Y}.$$
We get
\begin{align}
\label{BS distance y2}
\nonumber
Dist\pth{z[t+1], Y}
&\le \max_{i\in \calN} Dist\pth{x_i[t]-\lambda[t]\widetilde{g}_{j_{t+1}}[t], ~Y}~~~\text{by (\ref{BS a1})}\\
\nonumber
&= Dist\pth{x_{j_{t+1}^{\prime}}[t]-\lambda[t]\widetilde{g}_{j_{t+1}}[t], ~Y}\\
\nonumber
&=Dist\pth{x_{j_{t+1}^{\prime}}[t]-\lambda[t]\sum_{q\in \calN} \alpha_q h_q^{\prime}(x_q[t]), ~Y}~~~\text{by Lemma \ref{BS valid gradient}}\\
\nonumber
&=\inf_{y\in Y}\left |x_{j_{t+1}^{\prime}}[t]-\lambda[t]\sum_{q\in \calN} \alpha_q h_q^{\prime}(x_q[t])-y \right |\\
&\le \inf_{y\in Y}\left |x_{j_{t+1}^{\prime}}[t]-\lambda[t]p_{t+1}^{\prime}(x_{j_{t+1}^{\prime}}[t])-y\right |+\lambda[t]L (M[t]-m[t]).
\end{align}
where $p_{t}$ is defined in (\ref{valid gradient inter}).
Note that for each $t\ge 0$, there exists a non-faulty agent $j_{t}^{\prime}$ such that (\ref{BS distance y2}) holds, and there exists a sequence of agents $\{j_{t}^{\prime}\}_{t=0}^{\infty}$.
Let $\{x[t]\}_{t=0}^{\infty}$ be a sequence of estimates such that $x[t]=x_{j_{t+1}^{\prime}}[t]$.
Let $\{g[t]\}_{t=0}^{\infty}$ be a sequence of gradients such that $g[t]=p_{t}^{\prime}(x_{j_{t+1}^{\prime}}[t])$.\\

The remaining of the proof is identical to the proof of Theorem \ref{optimal alg2}.

\eproof
\end{proof}

\section{Discussion and Conclusion}
\label{conclusion and discussion}
So far, a synchronous system is considered. In an asynchronous system, when there are up to $f$ crash faults, Problem 1 is not solvable, since it is possible that every agent in the system is non-faulty, but $f$ agents are slow. In this case, the system will mistakenly ``treat" the slow agents as crashed agents. Consequently, the weights of the slow agents may be strictly smaller than the other agents.
Despite the impossibility of solving Problem 1 in asynchronous system, nevertheless, Problem 2 can be solved
with $\beta\ge \frac{1}{n}$ and $\gamma\ge |\calN|-f$. In particular, Algorithm 2 can be easily adapted for asynchronous system by modifying the receiving step (step 2). For completeness, we list out the algorithm for crash faults.
\paragraph{}
\hrule
~
\vspace*{4pt}

{\bf Algorithm 4} (crash faults) for agent $j$ for iteration $t\ge 1$ :
~
\vspace*{4pt}\hrule

\begin{list}{}{}

\item[{\bf Step 1:}]
Compute $h_j^{\prime}(x_j[t-1])$-- the gradient of local function $h_j(\cdot)$ at point $x_j[t-1]$, and send the triple $\pth{x_j[t-1],\, h_j^{\prime}(x_j[t-1]), \, t}$ to all the agents (including agent $j$ itself).\\
~
\item[{\bf Step 2:}]
Upon receiving $\pth{x_i[t-1],\, h_i^{\prime}(x_i[t-1]), \, t}$ from $n-f$ non-faulty agents (including agent $j$ itself) -- these received tuples form a multiset $\calR_j[t-1]$,  update $x_j$ as
\begin{align}
\label{asyn update x crash 2}
x_j[t]=\frac{1}{\left | \calR_j[t-1]\right |} \pth{\sum_{i\in \calR_j[t-1]} \pth{x_i[t-1]-\lambda[t-1] h_i^{\prime}(x_i[t-1])}}.
\end{align}
~
\end{list}

~
\hrule

~

~
Note that $\left |\calR_j[t-1] \right|=n-f$.
Since at most $f$ agents may crash, agent $j$ can receive messages from at least $n-f$ agents in step 2. Thus, Algorithm 3 will always proceed to the next iteration.
We are able to show the following theorem.
\begin{theorem}
\label{alg4 optimal}
Algorithm 4 solves Problem 2 with $\beta=\frac{1}{n}$ and $\gamma=n-f$.
\end{theorem}

The collection of valid function is defined as follows.
\begin{align*}
\calC\triangleq \Big{\{}~~p(x)~:~ p(x)&=\sum_{i\in \calV} \alpha_i h_i(x), \forall i\in \calV, \alpha_i\ge 0, \sum_{i\in \calN} \alpha_i=1, ~~\text{and}\\
&\sum_{i\in \calV}{\bf 1}\pth{\alpha_i\ge \frac{1}{n}}\ge n-f
~~\Big{\}}
\end{align*}

The proof of Theorem \ref{alg4 optimal} is similar to the proof of Theorem \ref{optimal alg2}.  \\

In an asynchronous system, when there are up to $f$ Byzantine faults, simple iterative algorithms like Algorithm 3 may not exist, observing that it is impossible to achieve Byzantine consensus with single round of message exchange with only $n=3f+1$ agents.
In contrast, when the algorithm introduced in \cite{abraham2005optimal} is used as a communication mechanism in each iteration, we believe that Algorithm 3 can be modified such that it can solve Problem 2 with $\beta\ge \frac{1}{2(|\calN|-f)}$ and $\gamma\ge |\calN|-2f$.
There may be a tradeoff between the system size $n$ and the communication load in each iteration. We leave this problem for future exploration.

Note that the definition of admissibility of the local functions in this report is slightly different from that in \cite{su2015byzantine}. Comparing to \cite{su2015byzantine}, stronger assumptions  are used in proving the correctness of the three iterative algorithms developed in this work. In particular, we require that the local functions have to have $L$--Lipschitz derivatives. Whether such assumptions are necessary or not is still open, and we leave this for future exploration as well.

\bibliographystyle{plain}

\bibliography{PSDA_DL}

\newpage
\appendix


\centerline{\large\bf Appendices}

~
\section{Lemma \ref{valid convex crash}}
\label{appendix: valid convex crash}

 \begin{proof}
Let $x_1, x_2\in Y$ such that $x_1\not=x_2$. By definition of $Y$, there exist valid functions
$$p_1(x)=C_1\pth{\sum_{i\in \calN} h_i(x)+\sum_{i\in \calF} \alpha_i h_i(x)}, ~~~~\text{and}~~~ p_2(x)=C_2\pth{\sum_{i\in \calN} h_i(x)+\sum_{i\in \calF} \beta_i h_i(x)},$$
such that $x_1\in \argmin~ p_1(x)$ and $x_2\in \argmin~ p_2(x)$, respectively. Note that it is possible that $p_1(\cdot)=p_2(\cdot)$, and that $p_i(\cdot)=\widetilde{p}(\cdot)$ for $i=1$ or $i=2$.\\

Given $0\le \alpha\le 1$, let $x_{\alpha}=\alpha x_1+(1-\alpha) x_2$. We consider two cases:
\begin{itemize}
\item[(i)] ~$x_{\alpha}\in \argmin ~p_1(x)\cup \argmin ~p_2(x) \cup \argmin ~ \widetilde{p}(x)$, and
\item[(ii)] $x_{\alpha}\notin \argmin ~p_1(x)\cup \argmin ~p_2(x) \cup \argmin ~ \widetilde{p}(x)$.
\end{itemize}

\paragraph{{\bf Case (i)}:~$x_{\alpha}\in \argmin ~p_1(x)\cup \argmin ~p_2(x) \cup \argmin ~ \widetilde{p}(x)$.}

When $x_{\alpha}\in \argmin ~p_1(x)\cup \argmin ~p_2(x) \cup \argmin ~ \widetilde{p}(x)$, by definition of $Y$, we have
$$ x_{\alpha}\in \argmin ~p_1(x)\cup \argmin ~p_2(x) \cup \argmin ~ \widetilde{p}(x)\subseteq Y.$$
Thus, $x_{\alpha}\in Y$.

\paragraph{{\bf Case (ii)}:~$x_{\alpha}\notin \argmin ~p_1(x)\cup \argmin ~p_2(x) \cup \argmin ~ \widetilde{p}(x)$.}

By symmetry, WLOG, assume that $x_1<x_2$. By definition of $x_{\alpha}$ and the assumption of case (ii),
it holds that $x_1<x_{\alpha}<x_2$.
In particular, it must be that
$$x_{\alpha}> \max \pth{\argmin p_1(x)}~\text{and}~ x_{\alpha}< \min \pth{\argmin p_1(x)},$$
which imply that $p_1^{\prime}(x_{\alpha})>0$ and $p_2^{\prime}(x_{\alpha})<0$. There are two possibilities for $\widetilde{p}^{\prime}(x_{\alpha})$: either $\widetilde{p}^{\prime}(x_{\alpha})>0$ or $\widetilde{p}^{\prime}(x_{\alpha})<0$. Note that $\widetilde{p}^{\prime}(x_{\alpha})\not=0$, since
 $x_\alpha\notin \argmin ~ \widetilde{p}(x)$.\\

Assume that $\widetilde{p}^{\prime}(x_{\alpha})<0$. Then, there exists $0\le \zeta\le 1$ such that
$$
 \zeta ~p_1^{\prime}(x_{\alpha}) + (1-\zeta)~\widetilde{p}^{\prime}(x_{\alpha})=0.$$
By definition of $p_1(x)$ and $\widetilde{p}(x)$, we have
 \begin{align*}
 0~&=~\zeta ~ p_1^{\prime}(x_{\alpha}) + (1-\zeta)~\widetilde{p}^{\prime}(x_{\alpha})\\
 ~&=\zeta ~C_1\pth{\sum_{i\in \calN} h^{\prime}_i(x_{\alpha})+\sum_{i\in \calF} \alpha_i h^{\prime}_i(x_{\alpha})} + (1-\zeta)\pth{\frac{1}{|\calN|}\sum_{i\in \calN} h^{\prime}_i(x_{\alpha})}\\
 ~&=\pth{\zeta C_1+(1-\zeta)\frac{1}{|\calN|}} \sum_{i\in \calN} h^{\prime}_i(x_{\alpha})
 +\zeta C_1 \sum_{i\in \calF} \alpha_i h^{\prime}_i(x_{\alpha}).
\end{align*}
Thus, $x_{\alpha}$ is an optimum of function
\begin{align}
\label{valid obj crash syn}
\pth{\zeta C_1+(1-\zeta)\frac{1}{|\calN|}} \sum_{i\in \calN} h_i(x)
 +\zeta C_1 \sum_{i\in \calF} \alpha_i h_i(x).
\end{align}
Since $p_1(x)\in \calC$, it holds that $C_1\pth{|\calN|+\sum_{i\in \calF} \alpha_i}=1$. Then we get
\begin{align*}
\pth{\zeta C_1+(1-\zeta)\frac{1}{|\calN|}} |\calN|+\zeta C_1 \sum_{i\in \calF} \alpha_i&=\zeta C_1\pth{|\calN|+ \sum_{i\in \calF} \alpha_i}+(1-\zeta)\frac{1}{|\calN|}|\calN|\\
&=\zeta \, 1+(1-\zeta)\, 1=1.
\end{align*}
In addition, since $\frac{1}{n}\le C_1\le \frac{1}{|\calN|}$, we get
\begin{align*}
\frac{1}{n}~=~ \zeta \frac{1}{n}+(1-\zeta)\frac{1}{n}\le \zeta C_1+(1-\zeta)\frac{1}{|\calN|} \le \zeta \frac{1}{|\calN|}+(1-\zeta)\frac{1}{|\calN|}=\frac{1}{|\calN|}.
\end{align*}
So function (\ref{valid obj crash syn}) is a valid function.

Similarly, we can show that the above result holds when $\widetilde{p}^{\prime}(x_{\alpha})>0$ is positive.

Therefore, set $Y$ is convex.
\eproof
\end{proof}

\section{Lemma \ref{valid closed crash}}
\label{appendix: valid closed crash}

Define an auxiliary function $r(x)$ as follows
\begin{align}
\label{aux f}
r(x)\triangleq \sum_{i\in \calN} h_i^{\prime}(x)+\sum_{i\in \calF} \pth{h_i^{\prime}(x){\bf 1}\{h_i^{\prime}(x)>0\}}.
\end{align}
\begin{proposition}
\label{aux continuous}
Function $r(x)$ is continuous and non-decreasing.
\end{proposition}

\begin{proof}
Since $h_i(x)$ is convex for each $i\in \calV$, it holds that $h_i^{\prime}(x)$ is non-decreasing. In addition, ${\bf 1}\{h_i^{\prime}(x)>0\}$ is also non-decreasing for each $i\in \calV$. Thus, function $r(x)$ is non-decreasing.  \\

For each $i\in \calV$, since $h_i(\cdot)$ is differentiable and continuous, it follows that $h_i^{\prime}(\cdot)$ is continuous. That is, $\forall\, \frac{\epsilon}{n}>0, ~\exists ~ \delta>0$, and for each $i\in \calV$, such that
$$
|x-c|<\delta ~~\Longrightarrow~ \left |h_i^{\prime}(x)-h_i^{\prime}(c) \right|\le \frac{\epsilon}{n}.
$$
Then
\begin{align}
\label{conti}
\nonumber
|r(x)-r(c)|&=\left |  \sum_{i\in \calN} h_i^{\prime}(x)+\sum_{i\in \calF} \pth{h_i^{\prime}(x){\bf 1}\{h_i^{\prime}(x)>0\}}-\pth{\sum_{i\in \calN} h_i^{\prime}(c)+\sum_{i\in \calF} \pth{h_i^{\prime}(c){\bf 1}\{h_i^{\prime}(c)>0\}}}  \right |\\
\nonumber
&=\left |  \sum_{i\in \calN} \pth{h_i^{\prime}(x)-h_i^{\prime}(c)}+\sum_{i\in \calF} \pth{h_i^{\prime}(x){\bf 1}\{h_i^{\prime}(x)>0\}-h_i^{\prime}(c){\bf 1}\{h_i^{\prime}(c)>0\}}  \right |\\
\nonumber
&\le \sum_{i\in \calN} |h_i^{\prime}(x)-h_i^{\prime}(c)|+\sum_{i\in \calF} \left | h_i^{\prime}(x){\bf 1}\{h_i^{\prime}(x)>0\}-h_i^{\prime}(c){\bf 1}\{h_i^{\prime}(c)>0\} \right |\\
&\le |\calN|\frac{\epsilon}{n}+\sum_{i\in \calF} \left | h_i^{\prime}(x){\bf 1}\{h_i^{\prime}(x)>0\}-h_i^{\prime}(c){\bf 1}\{h_i^{\prime}(c)>0\} \right |.
\end{align}

\paragraph{}
When ${\bf 1}\{h_i^{\prime}(x)>0\}={\bf 1}\{h_i^{\prime}(c)>0\}$, it holds that
\begin{align}
\label{cont1}
\nonumber
\left | h_i^{\prime}(x){\bf 1}\{h_i^{\prime}(x)>0\}-h_i^{\prime}(c){\bf 1}\{h_i^{\prime}(c)>0\} \right |&\le \max\{~0,  | h_i^{\prime}(x)-h_i^{\prime}(c)|\}\\
&\le  | h_i^{\prime}(x)-h_i^{\prime}(c)|<\frac{\epsilon}{n}.
\end{align}

\paragraph{}
Consider the case when ${\bf 1}\{h_i^{\prime}(x)>0\}\not={\bf 1}\{h_i^{\prime}(c)>0\}$. Assume $x<c$. As $h_i^{\prime}(\cdot)$ is non-decreasing, we have
$$ {\bf 1}\{h_i^{\prime}(x)>0\}=0\not=1={\bf 1}\{h_i^{\prime}(c)>0\}. $$
Then,
\begin{align}
\label{cont2}
\nonumber
\left | h_i^{\prime}(x){\bf 1}\{h_i^{\prime}(x)>0\}-h_i^{\prime}(c){\bf 1}\{h_i^{\prime}(c)>0\} \right |&=\left | 0-h_i^{\prime}(c)\right |=h_i^{\prime}(c)\\
\nonumber
&\le  h_i^{\prime}(c)-h_i^{\prime}(x)~~\text{since}~ h_i^{\prime}(x)\le 0\\
&=|h_i^{\prime}(c)-h_i^{\prime}(x)|<\frac{\epsilon}{n}.
\end{align}
Similarly, we can show $\left | h_i^{\prime}(x){\bf 1}\{h_i^{\prime}(x)>0\}-h_i^{\prime}(c){\bf 1}\{h_i^{\prime}(c)>0\} \right |<\frac{\epsilon}{n}$, for the case when $x\ge c$.

By (\ref{cont1}) and (\ref{cont2}), we can bound (\ref{conti}) as
\begin{align*}
|r(x)-r(c)|&\le |\calN|\frac{\epsilon}{n}+\sum_{i\in \calF} \left | h_i^{\prime}(x){\bf 1}\{h_i^{\prime}(x)>0\}-h_i^{\prime}(c){\bf 1}\{h_i^{\prime}(c)>0\} \right |\\
&\le |\calN|\frac{\epsilon}{n}+|\calF|\frac{\epsilon}{n}=\epsilon.
\end{align*}

\eproof
\end{proof}

\begin{proposition}
\label{app: p1}
For each valid function $p(x)\in \calC$, $\argmin_{x\in \reals}\, p(x)$ is compact.
\end{proposition}
\begin{proof}
Since $\argmin_{x\in \reals}\, h_i(x)$ is compact, and $p(x)$ is a convex combination of the local functions, it follows trivially that $\argmin_{x\in \reals}\, p(x)$ is bounded.
Thus, to show $\argmin_{x\in \reals}\, p(x)$ is compact, it remains to show that $\argmin_{x\in \reals}\, p(x)$ is closed.

Let $\{x_t\}_{t=0}^{\infty}\subseteq \argmin_{x\in \reals}\, p(x)$ be a sequence such that 
\begin{align}
\label{app: converge1}
\lim_{t\diverge} x_t=x^*. 
\end{align}
Recall that $h_i(\cdot)$ is continuous for each $i\in \calV$. Then $p(x)$ is also continuous. Thus, (\ref{app: converge1}) implies that 
\begin{align}
\label{app: converge2}
\lim_{t\diverge} p(x_t)=p(x^*).
\end{align}

Therefore, $x^*\in \argmin_{x\in \reals} p(x)$ and $\argmin_{x\in \reals} p(x)$ is compact.

\eproof
\end{proof}

\section*{Proof of Lemma \ref{valid closed crash}}
\begin{proof}
By Lemma \ref{valid convex crash}, we know that $Y$ is convex. To show $Y$ is closed, it is enough to show that $Y$ is bounded and both $\min Y$ and $\max Y$ exist.

For small enough $x$, $h^{\prime}_i(x)<0$ for each $i\in \calV$. Thus, $r(x)<0$ for small enough $x$. Similarly, $r(x)>0$ for large enough $x$.
By Proposition \ref{aux continuous}, we know that function $r(x)$ is non-decreasing and continuous. Thus, there exists $x_0\in \reals$ such that
\begin{align*}
0~=~r(x_0)=\sum_{i\in \calN} h_i^{\prime}(x_0)+\sum_{i\in \calF} \pth{h_i^{\prime}(x_0){\bf 1}\{h_i^{\prime}(x_0)>0\}}.
\end{align*}
Let
$$p_1(x)=C_0\pth{\sum_{i\in \calN} h_i(x)+\sum_{i\in \calF} \pth{h_i (x){\bf 1}\{h_i^{\prime}(x_0)>0\}}},$$
where $C_0\pth{|\calN|+\sum_{i\in \calF} {\bf 1}\{h_i^{\prime}(x_0)>0\}}=1$. Since
$$0\le \left |\sum_{i\in \calF} {\bf 1}\{h_i^{\prime}(x_0)>0\}\right|\le |\calF|,$$
it holds that $\frac{1}{n}\le C_0\le \frac{1}{|\calN|}$. Thus, $p_1(x)\in \calC$ is a valid function.

 Let $a=\min \pth{\argmin\, p_1(x)}$. By Proposition \ref{app: p1}, $\argmin\, p_1(x)$ is compact. Thus, $a$ is well-defined.  By definition $a\in Y$. Next we show that $a=\min Y$.

Suppose, on the contrary that, there exists $\tilde{a}<a$ such that $\tilde{a}\in Y$.  Since $\tilde{a}\in Y$, there exists $q(x)\in \calC$ such that $\tilde{a}\in \argmin\, q(x)$. That is,
\begin{align}
\label{close crash1}
0=q^{\prime}(\tilde{a})=C\pth{\sum_{i\in \calN} h_i^{\prime}\pth{\tilde{a}}+\sum_{i\in \calF} \alpha_i h_i^{\prime}\pth{\tilde{a}}}.
\end{align}
As $C>0$, from (\ref{close crash1}), we have
\begin{align*}
0&=\sum_{i\in \calN} h_i^{\prime}\pth{\tilde{a}}+\sum_{i\in \calF} \alpha_i h_i^{\prime}\pth{\tilde{a}}\\
&\le  \sum_{i\in \calN} h_i^{\prime}\pth{\tilde{a}}+\sum_{i\in \calF} \alpha_i h_i^{\prime}\pth{\tilde{a}}{\bf 1}\{ h_i^{\prime}(\tilde{a})>0\} ~~~\text{since}~h_i^{\prime}\pth{\tilde{a}}{\bf 1}\{ h_i^{\prime}(\tilde{a})>0\}\ge h_i^{\prime}\pth{\tilde{a}}\\
&\le \sum_{i\in \calN} h_i^{\prime}\pth{\tilde{a}}+\sum_{i\in \calF} h_i^{\prime}\pth{\tilde{a}}{\bf 1}\{ h_i^{\prime}(\tilde{a})>0\} ~~~\text{since}~0\le \alpha_i\le 1\\
&=r(\tilde{a})\\
&\le r(x_0)~~~\text{since}~\tilde{a}<a\le x_0~\text{and monotonicity of}~ r(\cdot)\\
&=0.
\end{align*}
Thus, $r(\tilde{a})=0= r(x_0)$. Since $h_i^{\prime}(\cdot)$ and ${\bf 1}\{h_i^{\prime}(\cdot)>0\}$
are both non-decreasing for each $i\in \calV$, we get
\begin{align}
\label{identical}
h_i^{\prime}(\tilde{a})=h_i^{\prime}(x_0), ~\forall~i\in \calN, ~~\text{and}~~{\bf 1}\{h_i^{\prime}(\tilde{a})>0\}={\bf 1}\{h_i^{\prime}(x_0)>0\}, \forall~i\in \calF.
\end{align}
We obtain
\begin{align*}
p_1^{\prime}(\tilde{a})&=C_1\pth{\sum_{i\in \calN} h_i^{\prime}(\tilde{a})+\sum_{i\in \calF} \pth{h_i^{\prime}(\tilde{a}){\bf 1}\{h_i^{\prime}(x_0)>0\}}}\\
&=C_1\pth{\sum_{i\in \calN} h_i^{\prime}(\tilde{a})+\sum_{i\in \calF} \pth{h_i^{\prime}(\tilde{a}){\bf 1}\{h_i^{\prime}(\tilde{a})>0\}}}~~\text{by}~(\ref{identical})\\
&=C_1r(\tilde{a})=0.
\end{align*}
That is, $\tilde{a}\in \argmin\, p_1(x)$, contradicting the fact that $\tilde{a}<a=\min\pth{\argmin\, p_1(x)}$.

Therefore, $a=\min Y$, i.e., $\min Y$ exists. Similarly, we can show that $\max Y$ also exists.\\

Therefore, set $Y$ is closed.

\eproof
\end{proof}

\section{Proof of Proposition \ref{crash sum 0}}
\label{appendix: crash sum 0}
\begin{proof}

For any $t\ge 1$, we have
\begin{align*}
\ell(t)&=\sum_{r=0}^{t-1}\lambda[r] b^{t-r}\\
&=\sum_{r=0}^{\lceil \frac{t}{2}\rceil}\lambda[r] b^{t-r}+\sum_{r=\lceil \frac{t}{2}\rceil+1}^{t-1}\lambda[r] b^{t-r}\\
&\le \sum_{r=0}^{\lceil \frac{t}{2}\rceil}\lambda[0] b^{t-r}+\lambda[\lceil \frac{t}{2}\rceil]\sum_{r=\lceil \frac{t}{2}\rceil+1}^{t-1}b^{t-r}~~~\text{since}~\lambda[t]\le \lambda[t-1],~\forall t\ge 1\\
&\le \lambda[0] \frac{b^{t-\lceil \frac{t}{2}\rceil}}{1-b}+\frac{b\lambda[\lceil \frac{t}{2}\rceil]}{1-b}\\
&\le \lambda[0] \frac{b^{\frac{t}{2}-1}}{1-b}+\frac{b\lambda[\lceil \frac{t}{2}\rceil]}{1-b}.
\end{align*}
Thus, we get
$$\limsup_{t\diverge} \ell(t)\le \lim_{t\diverge}\pth{\lambda[0] \frac{b^{\frac{t}{2}-1}}{1-b}+\frac{b\lambda[\lceil \frac{t}{2}\rceil]}{1-b}}
=\lambda[0] \frac{1}{1-b}\lim_{t\diverge} b^{\frac{t}{2}-1}+\frac{b}{1-b}\lim_{t\diverge}\lambda[\lceil \frac{t}{2}\rceil]\overset{(a)}{=}0+0=0.$$
Equality $(a)$ follows from the fact that $0\le b<1$ and the fact that $\lim_{t\diverge}\lambda[\lceil \frac{t}{2}\rceil]=0$.
On the other hand, by definition of $\ell(t)$ we know
$\ell(t)\ge 0$ for each $t\ge 1$. Thus $\liminf_{t\diverge} \ell(t)\ge 0.$\\

Therefore, the limit of $\ell(t)$ exists and $\lim_{t\diverge}\ell(t)=0$.

\eproof
\end{proof}

\section{Proof of Lemma \ref{SB asymptotic consensus fb}}
\label{app: SB asymptotic consensus fb}
\begin{proof}

When $t=0$, for all $i, j\in \calN$ we have
\begin{align*}
|x_i[0]-x_j[0]|\le \max_{i\in \calN} x_i[0]-\min_{j\in \calN} x_j[0]=U-u.
\end{align*}
Recall (\ref{MR evo BS}). For $t\ge 1$,
\begin{align*}
{\bf x}[t]={\bf \Phi}(t-1, 0){\bf x}[0]-\sum_{r=0}^{t-1} \lambda[r]{\bf \Phi}(t-1, r+1) \widetilde{\bf g}[r],
\end{align*}
Then each $x_i[t]$ can be written as
\begin{align*}
x_i [t]=\sum_{k=1}^{n-\phi}{\bf \Phi}_{ik}(t-1, 0)x_k [0]-\sum_{r=0}^{t-1} \pth{\lambda[r]\sum_{k=1}^{n-\phi}{\bf \Phi}_{ik}(t-1, r+1) \widetilde{g}_k[r]}.
\end{align*}
Thus
\begin{align}
\label{lemma consensus BS}
\nonumber
|x_i[t]-x_j[t]|&=  \Bigg{|} \sum_{k=1}^{n-\phi}{\bf \Phi}_{ik}(t-1, 0)x_k [0]-\sum_{r=0}^{t-1} \pth{\lambda[r]\sum_{k=1}^{n-\phi}{\bf \Phi}_{ik}(t-1, r+1) \widetilde{g}_k[r]}\\
\nonumber
&\quad-\sum_{k=1}^{n-\phi}{\bf \Phi}_{jk}(t-1, 0)x_k [0]+\sum_{r=0}^{t-1} \pth{\lambda[r]\sum_{k=1}^{n-\phi}{\bf \Phi}_{jk}(t-1, r+1) \widetilde{g}_k[r]} \Bigg{|}\\
\nonumber
&\le \left |   \sum_{k=1}^{n-\phi}{\bf \Phi}_{ik}(t-1, 0)x_k [0]-\sum_{k=1}^{n-\phi}{\bf \Phi}_{jk}(t-1, 0)x_k [0]  \right |\\
&\quad + \left |  \sum_{r=0}^{t-1} \pth{\lambda[r]\sum_{k=1}^{n-\phi}{\bf \Phi}_{jk}(t-1, r+1) \widetilde{g}_k[r]}-\sum_{r=0}^{t-1} \pth{\lambda[r]\sum_{k=1}^{n-\phi}{\bf \Phi}_{ik}(t-1, r+1) \widetilde{g}_k[r]}  \right |.
\end{align}
We bound the two terms in (\ref{lemma consensus BS}) separately. For the first term in (\ref{lemma consensus BS}), we have
\begin{align}
\label{lemma consensus BS t1}
\nonumber
\left |   \sum_{k=1}^{n-\phi}{\bf \Phi}_{ik}(t-1, 0)x_k [0]-\sum_{k=1}^{n-\phi}{\bf \Phi}_{jk}(t-1, 0)x_k [0]  \right |&=\left |\sum_{k=1}^{n-\phi}\pth{{\bf \Phi}_{ik}(t-1, 0)-{\bf \Phi}_{jk}(t-1, 0)}x_k [0]\right |\\
\nonumber
&\le \sum_{k=1}^{n-\phi}\left |{\bf \Phi}_{ik}(t-1, 0)-{\bf \Phi}_{jk}(t-1, 0)\right | \left |x_k [0]\right |\\
\nonumber
&\le \sum_{k=1}^{n-\phi} \gamma^{\lceil \frac{t}{\nu}\rceil} \left |x_k [0]\right | ~~~\text{by Theorem \ref{convergencerate}}\\
&\le (n-\phi)  \max\{|u|, |U|\} \gamma^{\lceil \frac{t}{\nu}\rceil}.
\end{align}
In addition, the second term in (\ref{lemma consensus BS}) can be bounded as follows.
\begin{align}
\label{lemma consensus BS t2}
\nonumber
&\left |  \sum_{r=0}^{t-1} \pth{\lambda[r]\sum_{k=1}^{n-\phi}{\bf \Phi}_{jk}(t-1, r+1) \widetilde{g}_k[r]}-\sum_{r=0}^{t-1} \pth{\lambda[r]\sum_{k=1}^{n-\phi}{\bf \Phi}_{ik}(t-1, r+1) \widetilde{g}_k[r]}  \right |\\
\nonumber
&=\left |  \sum_{r=0}^{t-1} \pth{\lambda[r]\sum_{k=1}^{n-\phi}{\bf \Phi}_{jk}(t-1, r+1) -{\bf \Phi}_{ik}(t-1, r+1)} \widetilde{g}_k[r]  \right |\\
\nonumber
&\le  \sum_{r=0}^{t-1} \pth{\lambda[r]\sum_{k=1}^{n-\phi}\left |{\bf \Phi}_{jk}(t-1, r+1) -{\bf \Phi}_{ik}(t-1, r+1) \right |} |\widetilde{g}_k[r]|\\
&\le L\sum_{r=0}^{t-1}\lambda[r](n-\phi) \gamma^{\lceil \frac{t-1-r}{\nu}\rceil} ~~~\text{by Theorem \ref{convergencerate} and the fact that $ |\widetilde{g}_k[r]|\le L$}
\end{align}
From (\ref{lemma consensus BS t1}) and (\ref{lemma consensus BS t2}), the LHS of (\ref{lemma consensus BS}) can be upper bounded by
\begin{align*}
|x_i[t]-x_j[t]|\le (n-\phi)  \max\{|u|, |U|\} \gamma^{\lceil \frac{t}{\nu}\rceil}+L\sum_{r=0}^{t-1}\lambda[r](n-\phi) \gamma^{\lceil \frac{t-1-r}{\nu}\rceil}.
\end{align*}
The proof is complete.

\eproof
\end{proof}

\section{Proof of Corollary \ref{BS consensus}}
\label{app: BS consensus}
\begin{proof}
By Lemma \ref{SB asymptotic consensus fb}, for each $t\ge 1$,
\begin{align*}
|x_i[t]-x_j[t]|&\le (n-\phi)  \max\{|u|, |U|\} \gamma^{\lceil \frac{t}{\nu}\rceil}+L\sum_{r=0}^{t-1}\lambda[r](n-\phi) \gamma^{\lceil \frac{t-1-r}{\nu}\rceil}\\
&\le (n-\phi)  \max\{|u|, |U|\} \gamma^{ \frac{t}{\nu}}+L\sum_{r=0}^{t-1}\lambda[r](n-\phi) \gamma^{ \frac{t-1-r}{\nu}},
\end{align*}
and for all $i, j\in \calN$. Taking limit sup on both sides, we get
\begin{align*}
\limsup_{t\diverge}|x_i[t]-x_j[t]|&\le
(n-\phi)  \max\{|u|, |U|\} \limsup_{t\diverge}\gamma^{ \frac{t}{\nu}}
+L(n-\phi)\limsup_{t\diverge}\pth{\sum_{r=0}^{t-1}\lambda[r] \gamma^{ \frac{t-1-r}{\nu}}}\\
&=0+L(n-\phi)\limsup_{t\diverge}\pth{\sum_{r=0}^{t-1}\lambda[r] \gamma^{ \frac{t-1-r}{\nu}}}\\
&=0+0=0~~~\text{by Lemma \ref{crash sum 0}},
\end{align*}
proving the corollary.

\eproof
\end{proof}

\section{Proof of Lemma \ref{finiteness of series SB}}
\label{app: finiteness of series SB}
\begin{proof}
By Lemma \ref{SB asymptotic consensus fb}, for $t\ge 1$ we have
$$M[t]-m[t]\le (n-\phi)  \max\{|u|, |U|\} \gamma^{\lceil \frac{t}{\nu}\rceil}+L\sum_{r=0}^{t-1}\lambda[r](n-\phi) \gamma^{\lceil \frac{t-1-r}{\nu}\rceil}.
$$
Thus, we get
\begin{align}
\label{finiteness SB}
\nonumber
\sum_{t=1}^{\infty} \lambda[t]\pth{M[t]-m[t]}&\le ~\sum_{t=1}^{\infty} \lambda[t] \pth{(n-\phi)  \max\{|u|, |U|\} \gamma^{\lceil \frac{t}{\nu}\rceil}+L\sum_{r=0}^{t-1}\lambda[r](n-\phi) \gamma^{\lceil \frac{t-1-r}{\nu}\rceil}}\\
&~=(n-\phi)  \max\{|u|, |U|\} \sum_{t=1}^{\infty} \lambda[t] \gamma^{\lceil \frac{t}{\nu}\rceil}+L(n-\phi)\sum_{t=1}^{\infty} \lambda[t] \sum_{r=0}^{t-1}\lambda[r] \gamma^{\lceil \frac{t-1-r}{\nu}\rceil}.
\end{align}
Since $\lambda[t]\le \lambda [0]$ for each $t\ge 0$, we have
\begin{align}
\label{SB finiteness t1}
\nonumber
(n-\phi)  \max\{|u|, |U|\} \sum_{t=1}^{\infty} \lambda[t] \gamma^{\lceil \frac{t}{\nu}\rceil}&\le (n-\phi)  \max\{|u|, |U|\} \lambda[0] \sum_{t=1}^{\infty} \gamma^{\lceil \frac{t}{\nu}\rceil}\\
\nonumber
\nonumber
&\le (n-\phi)  \max\{|u|, |U|\} \lambda[0] \sum_{t=1}^{\infty} \gamma^{ \frac{t}{\nu}}\\
&\le (n-\phi)  \max\{|u|, |U|\} \lambda[0]\frac{1}{1-\gamma^{ \frac{1}{\nu}}}~<~\infty.
\end{align}

\begin{align}
\label{SB finiteness t2}
\nonumber
L(n-\phi)\sum_{t=1}^{\infty} \lambda[t] \sum_{r=0}^{t-1}\lambda[r] \gamma^{\lceil \frac{t-1-r}{\nu}\rceil}&= L(n-\phi)\sum_{t=1}^{\infty}  \sum_{r=0}^{t-1}\lambda[t]\lambda[r] \gamma^{\lceil \frac{t-1-r}{\nu}\rceil}\\
\nonumber
&\le \frac{L(n-\phi)}{2}\sum_{t=1}^{\infty}  \sum_{r=0}^{t-1}\pth{\lambda^2[t]+\lambda^2[r]} \gamma^{\lceil \frac{t-1-r}{\nu}\rceil}~~~\text{since}~\lambda[t]\lambda[r]\le \frac{\lambda^2[t]+\lambda^2[r]}{2}\\
&=\frac{L(n-\phi)}{2}\sum_{t=1}^{\infty}  \lambda^2[t]\sum_{r=0}^{t-1} \gamma^{\lceil \frac{t-1-r}{\nu}\rceil}+\frac{L(n-\phi)}{2}\sum_{t=1}^{\infty}  \sum_{r=0}^{t-1}\lambda^2[r] \gamma^{\lceil \frac{t-1-r}{\nu}\rceil}
\end{align}
The first term on the RHS of (\ref{SB finiteness t2}) can be bounded as
\begin{align}
\label{SB finiteness t2-1}
\nonumber
\frac{L(n-\phi)}{2}\sum_{t=1}^{\infty}  \lambda^2[t]\sum_{r=0}^{t-1} \gamma^{\lceil \frac{t-1-r}{\nu}\rceil}&\le \frac{L(n-\phi)}{2}\sum_{t=1}^{\infty}  \lambda^2[t]\sum_{r=0}^{t-1} \gamma^{\frac{t-1-r}{\nu}}\\
\nonumber
& \le \frac{L(n-\phi)}{2}\sum_{t=1}^{\infty}  \lambda^2[t] \frac{1}{1-\gamma^{\frac{1}{\nu}}}\\
\nonumber
&=\frac{L(n-\phi)}{2\pth{1-\gamma^{\frac{1}{\nu}}}}\sum_{t=1}^{\infty}  \lambda^2[t]\\
&<\infty ~~~\text{since}~ \sum_{t=1}^{\infty}  \lambda^2[t]<\infty.
\end{align}

For the second term on the RHS of (\ref{SB finiteness t2}), for any fixed $T$, we get
\begin{align*}
\nonumber
\frac{L(n-\phi)}{2}\sum_{t=1}^{T}  \sum_{r=0}^{t-1}\lambda^2[r] \gamma^{\lceil \frac{t-1-r}{\nu}\rceil}&\le \frac{L(n-\phi)}{2}\sum_{t=1}^{T}  \sum_{r=0}^{t-1}\lambda^2[r] \gamma^{\frac{t-1-r}{\nu}}\\
\nonumber
&=\frac{L(n-\phi)}{2} \sum_{r=0}^{T-1}\lambda^2[r]\sum_{t=0}^{T-1-r} \gamma^{\frac{t}{\nu}}\\
&\le \frac{L(n-\phi)}{2(1-\gamma^{\frac{1}{\nu}})}  \sum_{r=0}^{T-1}\lambda^2[r].
\end{align*}
Thus, we get
\begin{align}
\label{SB finiteness t2-2}
\frac{L(n-\phi)}{2}\sum_{t=0}^{\infty}  \sum_{r=0}^{t-1}\lambda^2[r] \gamma^{\lceil \frac{t-1-r}{\nu}\rceil}\le \frac{L(n-\phi)}{2(1-\gamma^{\frac{1}{\nu}})}  \sum_{r=0}^{\infty}\lambda^2[r]<\infty.
\end{align}
By (\ref{SB finiteness t1}), (\ref{SB finiteness t2-1}) and (\ref{SB finiteness t2-2}), we get
$$\sum_{t=1}^{\infty} \lambda[t] \pth{M[t]-m[t]}<\infty.$$
In addition,
\begin{align*}
\sum_{t=0}^{\infty} \lambda[t] \pth{M[t]-m[t]}&= \lambda[0]\pth{M[0]-m[0]}+\sum_{t=1}^{\infty} \lambda[t] \pth{M[t]-m[t]}\\
&= \lambda[0] (U-u)+\sum_{t=1}^{\infty} \lambda[t] \pth{M[t]-m[t]}<\infty,
\end{align*}
proving the lemma.

\eproof
\end{proof}

\section{Proof of Lemma \ref{BS valid closed}}
\label{app: BS closed}

Define an auxiliary function $\tr(x)$ as follows. For each $x\in \reals$, let $h_{i_1(x)}^{\prime}(x), \cdots, h_{i_{|\calN|}(x)}^{\prime}(x)$ be a non-decreasing order of $h_j^{\prime}(x)$, for $j\in \calN$. Define $\tr(x)$ as follows,
\begin{align}
\label{BS aux f}
\tr(x)=\pth{1-\frac{|\calN|-f-1}{2(|\calN|-f)}} h^{\prime}_{i_1(x)}(x)+\frac{1}{2(|\calN|-f)}\sum_{j=2}^{|\calN|-f} h^{\prime}_{i_j(x)}(x).
\end{align}
Intuitively speaking, $\tr(x)$ is the largest gradient value among all valid functions in $\tC$ at point $x$.

\begin{proposition}
\label{BS aux continuous}
Function $\tr(\cdot)$ is continuous and non-decreasing.
\end{proposition}

\begin{proof}
Let $x\le y \in \reals$.
\begin{align*}
\tr(y)-\tr(x)&=\pth{1-\frac{|\calN|-f-1}{2(|\calN|-f)}} h^{\prime}_{i_1(y)}(y)+\frac{1}{2(|\calN|-f)}\sum_{j=2}^{|\calN|-f} h^{\prime}_{i_j(y)}(y)\\
&\quad-\pth{1-\frac{|\calN|-f-1}{2(|\calN|-f)}} h^{\prime}_{i_1(x)}(x)-\frac{1}{2(|\calN|-f)}\sum_{j=2}^{|\calN|-f} h^{\prime}_{i_j(x)}(x)\\
&\ge \pth{1-\frac{|\calN|-f-1}{2(|\calN|-f)}} h^{\prime}_{i_1(x)}(y)+\frac{1}{2(|\calN|-f)}\sum_{j=2}^{|\calN|-f} h^{\prime}_{i_j(x)}(y)\\
&\quad-\pth{1-\frac{|\calN|-f-1}{2(|\calN|-f)}} h^{\prime}_{i_1(x)}(x)-\frac{1}{2(|\calN|-f)}\sum_{j=2}^{|\calN|-f} h^{\prime}_{i_j(x)}(x)\\
&=\pth{1-\frac{|\calN|-f-1}{2(|\calN|-f)}}\pth{h^{\prime}_{i_1(x)}(y)-h^{\prime}_{i_1(x)}(x)}
+\frac{1}{2(|\calN|-f)}\sum_{j=2}^{|\calN|-f}\pth{h^{\prime}_{i_j(x)}(y)-h^{\prime}_{i_j(x)}(x)}\\
&\le 0+0~~~\text{since}~x\le y~\text{and}~~h_i^{\prime}(\cdot)~ \text{is non-decreasing}
\end{align*}
Thus, function $\tr(\cdot)$ is non-decreasing. \\

Next we show that function $\tr(\cdot)$ is continuous.

For each $i\in \calV$, since $h_i(\cdot)$ is differentiable, it follows that $h_i^{\prime}(\cdot)$ is continuous. That is, $\forall\, \epsilon >0, ~\exists ~ \delta>0$, and for each $i\in \calV$, such that
$$
|x-c|<\delta ~~\Longrightarrow~ \left |h_i^{\prime}(x)-h_i^{\prime}(c) \right|\le \epsilon.
$$

Assume $c\le x< c+\delta$.
Then
\begin{align}
\label{BS conti}
\nonumber
|\tr(x)-\tr(c)|&=\tr(x)-\tr(c)~~~\text{by monotonicity of }\tr(\cdot)\\
\nonumber
&=\pth{1-\frac{|\calN|-f-1}{2(|\calN|-f)}} h^{\prime}_{i_1(x)}(x)+\frac{1}{2(|\calN|-f)}\sum_{j=2}^{|\calN|-f} h^{\prime}_{i_j(x)}(x)\\
\nonumber
&\quad-\pth{1-\frac{|\calN|-f-1}{2(|\calN|-f)}} h^{\prime}_{i_1(c)}(c)-\frac{1}{2(|\calN|-f)}\sum_{j=2}^{|\calN|-f} h^{\prime}_{i_j(c)}(c)\\
\nonumber
&\le \pth{1-\frac{|\calN|-f-1}{2(|\calN|-f)}} h^{\prime}_{i_1(x)}(x)+\frac{1}{2(|\calN|-f)}\sum_{j=2}^{|\calN|-f} h^{\prime}_{i_j(x)}(x)\\
\nonumber
&\quad-\pth{1-\frac{|\calN|-f-1}{2(|\calN|-f)}} h^{\prime}_{i_1(x)}(c)-\frac{1}{2(|\calN|-f)}\sum_{j=2}^{|\calN|-f} h^{\prime}_{i_j(x)}(c)\\
\nonumber
&\le \pth{1-\frac{|\calN|-f-1}{2(|\calN|-f)}} \pth{h^{\prime}_{i_1(x)}(x)-h^{\prime}_{i_1(x)}(c)}+\frac{1}{2(|\calN|-f)}\sum_{j=2}^{|\calN|-f} \pth{h^{\prime}_{i_j(x)}(x)-h^{\prime}_{i_1(x)}(c)}\\
&< \pth{1-\frac{|\calN|-f-1}{2(|\calN|-f)}} \epsilon +\frac{1}{2(|\calN|-f)}\sum_{j=2}^{|\calN|-f} \epsilon=\epsilon.
\end{align}
Similarly, we can show that when $c-\delta< x\le c$, $|\tr(x)-\tr(c)|<\epsilon$.\\

Thus, function $\tr(\cdot)$ is continuous. \\

The proof is complete.

\eproof
\end{proof}

\section*{Proof of Lemma \ref{BS valid closed}}
\begin{proof}
By Lemma \ref{BS valid convex}, we know that $\tY$ is convex. To show $Y$ is closed, it is enough to show that $\tY$ is bounded and both $\min\, \tY$ and $\max\, \tY$ exist.

By Proposition \ref{BS aux continuous}, we know that function $\tr(x)$ is non-decreasing and continuous. Thus, there exists $x_0\in \reals$ such that
\begin{align*}
0~=~\tr(x_0)=\pth{1-\frac{|\calN|-f-1}{2(|\calN|-f)}} h^{\prime}_{i_1(x_0)}(x_0)+\frac{1}{2(|\calN|-f)}\sum_{j=2}^{|\calN|-f} h^{\prime}_{i_j(x_0)}(x_0).
\end{align*}
Let
\begin{align}
\label{order}
q(x)=\pth{1-\frac{|\calN|-f-1}{2(|\calN|-f)}} h_{i_1(x_0)}(x)+\frac{1}{2(|\calN|-f)}\sum_{j=2}^{|\calN|-f} h_{i_j(x_0)}(x).
\end{align}
By construction, $q(x)\in \calC$ is a valid function. Note that due to the possibility of existence of ties in top $|\calN|-f$ rankings of the order $h_{i_1(x)}^{\prime}(x), \cdots, h_{i_{|\calN|}(x)}^{\prime}(x)$, for a given $x_0$, there may be multiple orders over $h_i^{\prime}(x_0), \forall i\in \calN$ of the top $|\calN|-f$ elements. Let $\calO$ be the collection of all such orders. Note that there is an one-to-one correspondence of an order and a valid function defined in (\ref{order}). Let
$$a=\min_{o\in \calO} \min\pth{\argmin\, q_o(x)},$$
 which is well-defined since $\argmin\, q_o(x)$ is compact, and $|\calO|$ is finite.\\

By definition $a\in \tY$. Next we show that $a=\min\, \tY$.

Suppose, on the contrary that, there exists $\tilde{a}<a$ such that $\tilde{a}\in \tY$.  Since $\tilde{a}\in \tY$, there exists $\tilde{q}(x)=\sum_{i\in \calN} \alpha_i h_i(x)\in \tilde{\calC}$ such that $\tilde{a}\in \argmin\, \tilde{q}(x)$. That is,
\begin{align}
\label{BS close}
\tilde{q}^{\prime}(\tilde{a})=0.
\end{align}
We have
\begin{align*}
0=q^{\prime}(\tilde{a})&=\sum_{i\in \calN} \alpha_i h^{\prime}_i(\tilde{a})\\
&\le \pth{1-\frac{|\calN|-f-1}{2(|\calN|-f)}} h^{\prime}_{i_1(\tilde{a})}(\tilde{a})+\frac{1}{2(|\calN|-f)}\sum_{j=2}^{|\calN|-f} h^{\prime}_{i_j(\tilde{a})}(\tilde{a})\\
&=\tr(\tilde{a})\le \tr(x_0)=0~~~\text{by monotonicity of } \tr(\cdot)
\end{align*}
Thus, $\tr(\tilde{a})=0=\tr(x_0)$.
In addition, we have
\begin{align*}
0&=\pth{1-\frac{|\calN|-f-1}{2(|\calN|-f)}} h^{\prime}_{i_1(\tilde{a})}(\tilde{a})+\frac{1}{2(|\calN|-f)}\sum_{j=2}^{|\calN|-f} h^{\prime}_{i_j(\tilde{a})}(\tilde{a})\\
&\le \pth{1-\frac{|\calN|-f-1}{2(|\calN|-f)}} h^{\prime}_{i_1(\tilde{a})}(x_0)+\frac{1}{2(|\calN|-f)}\sum_{j=2}^{|\calN|-f} h^{\prime}_{i_j(\tilde{a})}(x_0)~~~\text{by monotonicity of } h_i^{\prime}(\cdot)\\
&\le \pth{1-\frac{|\calN|-f-1}{2(|\calN|-f)}} h^{\prime}_{i_1(x_0)}(x_0)+\frac{1}{2(|\calN|-f)}\sum_{j=2}^{|\calN|-f} h^{\prime}_{i_j(x_0)}(x_0),
\end{align*}
which implies that $i_1(\tilde{a}), \cdots, i_{|\calN|-f}(\tilde{a})$ is an order in $\calO$. Thus, it can be seen that $\tilde{a}\ge a=\min_{o\in \calO} \min\pth{\argmin\, q_o(x)}$, contradicting the assumption that $\tilde{a}< a$. \\

Therefore, $a=\min \tY$, i.e., $\min\, \tY$ exists. Similarly, we can show that $\max\, \tY$ also exists. \\

Therefore, set $\tY$ is closed.

\eproof
\end{proof}

\end{document}